\documentclass[
 amsmath,amssymb,superscript,
 aps,pra,
onecolumn,
]{revtex4-2}
    
\usepackage{xcolor}
\usepackage{graphicx} 
\usepackage{dcolumn} 
\usepackage{bm} 
\usepackage{braket} 
\usepackage{amsthm}
\usepackage{bbm}
\usepackage[normalem]{ulem}

\usepackage{multirow}
\usepackage{makecell}
\usepackage{longtable}
\usepackage{tabularx}
\usepackage{array}
\usepackage{booktabs}

\usepackage[ruled]{algorithm2e}

\newtheorem{theorem}{Theorem}

\newtheorem{lemma}{Lemma}

\newtheorem*{theorem-non}{Theorem}
\newtheorem*{lemma-non}{Lemma}
\newtheorem*{corollary-non}{Corollary}
\newtheorem*{proposition-non}{Proposition}

 \usepackage[colorlinks=true, linkcolor=blue, citecolor=blue, urlcolor=blue]{hyperref}
 
\usepackage{appendix}

\newcommand{\bx}{\bm{x}}

\newcommand{\by}{\bm{y}}
\newcommand{\bw}{\bm{\mathrm{w}}}

\newcommand{\yi}{y^{(i)}}
\newcommand{\bxi}{\bm{x}^{(i)}}

\newcommand{\hatsigma}{\hat{\sigma}_n}

\newcommand{\hatkappa}{\kappa_{\Lambda}}

\newcommand{\tilderho}{\tilde{\rho}}

\newcommand{\bomega}{\bm{\omega}}

\newcommand{\llangle}{\langle\langle}
\newcommand{\rrangle}{\rangle\rangle}

\newcommand{\PX}{p_\text{X}}
\newcommand{\PY}{p_\text{Y}}
\newcommand{\PZ}{p_\text{Z}}

\DeclareMathOperator{\CRZ}{CRZ}
\DeclareMathOperator{\RX}{RX}
\DeclareMathOperator{\RZZ}{RZZ}
\DeclareMathOperator{\RZ}{RZ}
\DeclareMathOperator{\RY}{RY}
\DeclareMathOperator{\RP}{RP}

\DeclareMathOperator{\CNOT}{CNOT}
\DeclareMathOperator{\bUnitary}{\mathfrak{U}}

\DeclareMathOperator{\Tr}{Tr}

\DeclareMathOperator{\Var}{Var}

\DeclareMathOperator{\CI}{\text{CI}}
\DeclareMathOperator{\Unif}{\textnormal{Unif}}

\begin{document}

\title{Demonstration of Efficient Predictive Surrogates for Large-scale Quantum Processors} 

\author{Wei-You Liao$^{1}$}
\thanks{Equal contribution authors}

\author{Yuxuan Du$^{2,\textcolor{blue}{*}}$}
\thanks{Correspondence to:  duyuxuan123@gmail.com,\\dacheng.tao@ntu.edu.sg,\\quanhhl@ustc.edu.cn}

\author{Xinbiao Wang$^{2}$}
\thanks{Equal contribution authors}

\author{Tian-Ci Tian$^{1}$}
 
\author{Yong Luo$^{3}$}

\author{Bo Du$^{3}$}

\author{Dacheng Tao$^{2}$}
\thanks{Correspondence to:  duyuxuan123@gmail.com,\\dacheng.tao@ntu.edu.sg,\\quanhhl@ustc.edu.cn}

\author{He-Liang Huang$^{1}$}
\thanks{Correspondence to:  duyuxuan123@gmail.com,\\dacheng.tao@ntu.edu.sg,\\quanhhl@ustc.edu.cn}

\affiliation{$^1$Henan Key Laboratory of Quantum Information and Cryptography, Zhengzhou, Henan 450000, China}

\affiliation{$^2$College of Computing and Data Science, Nanyang Technological University, Singapore 639798, Singapore}

\affiliation{$^3$Institute of Artificial Intelligence, School of Computer Science, Wuhan University, Wuhan 430072, China}

\begin{abstract}
The ongoing development of quantum processors is driving breakthroughs in scientific discovery. Despite this progress, the formidable cost of fabricating large-scale quantum processors means they will remain rare for the foreseeable future, limiting their widespread application. To address this bottleneck, we introduce the concept of predictive surrogates, which are classical learning models designed to emulate the mean-value behavior of a given quantum processor with provably computational efficiency. In particular, we propose two predictive surrogates that can substantially reduce the need for quantum processor access in diverse practical scenarios. To demonstrate their potential in advancing digital quantum simulation, we use these surrogates to emulate a quantum processor with up to $20$ programmable superconducting qubits, enabling efficient pre-training of variational quantum eigensolvers for families of transverse-field Ising models and identification of non-equilibrium Floquet symmetry-protected topological phases. Experimental results reveal that the predictive surrogates not only reduce measurement overhead by orders of magnitude, but can also surpass the performance of conventional, quantum-resource-intensive approaches. Collectively, these findings establish predictive surrogates as a practical pathway to broadening the impact of advanced quantum processors.

\end{abstract}

\maketitle
\begin{figure*}[t]
	\centering\includegraphics[width=0.95\textwidth]{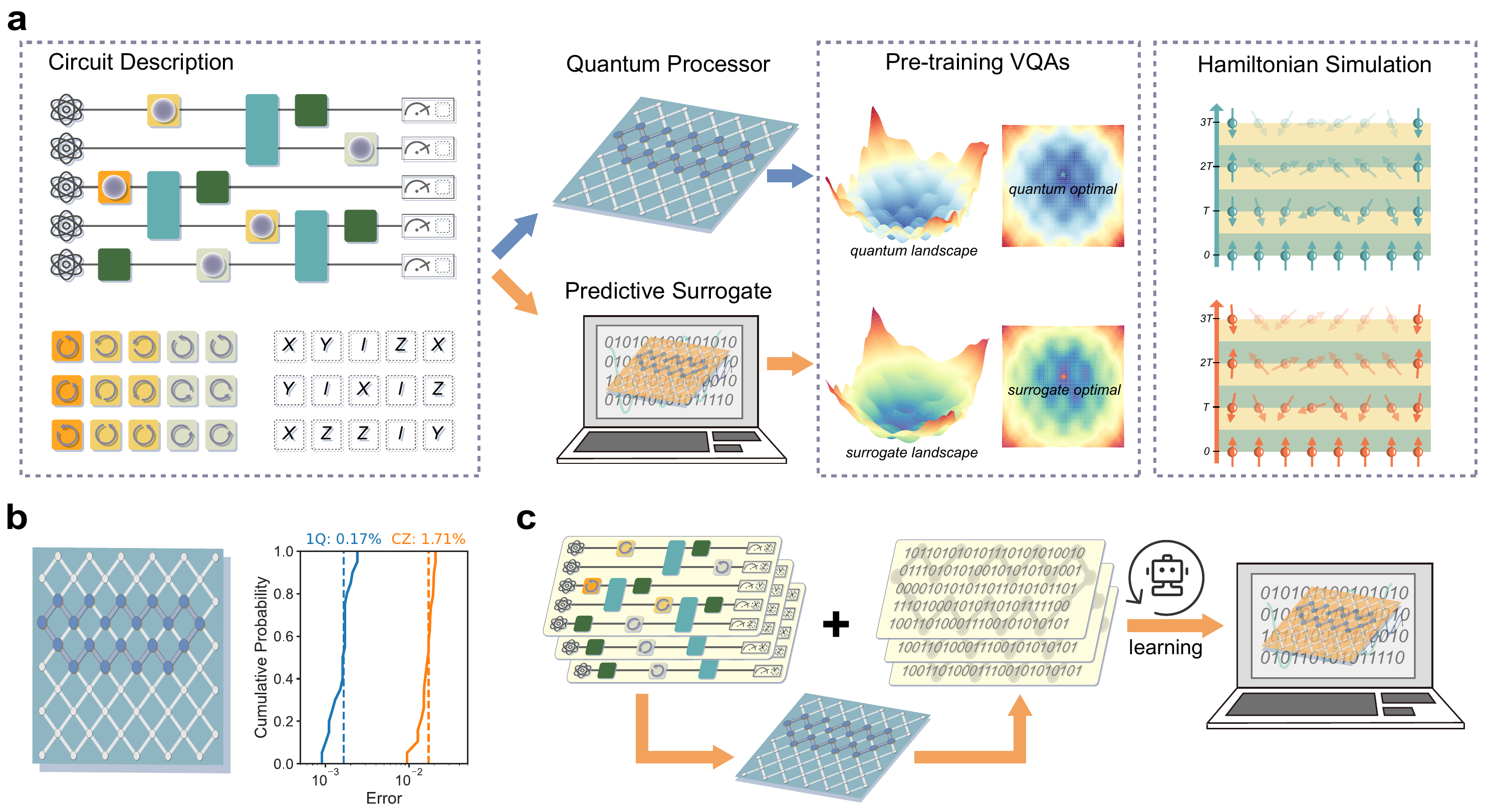}
	\caption{\small{\textbf{Framework of constructing predictive surrogates for quantum processors.} \textbf{a.}   Given the classical description of a quantum circuit $U(\bx)$ with a fixed gate layout and tunable gate parameters $\bx\in \mathcal{X}$, a predictive surrogate trained on data obtained from a quantum processor can emulate its mean-value behavior, i.e., $\Tr(\tilderho(\bx)O)$, for parameters $\bx$ spanning the entire space $\mathcal{X}$. Here, $\tilderho(\bx)$ represents the evolved state when executing $U(\bx)$ on a noisy quantum processor with a fixed initial state $\rho_0$, and $O$ refers to an observable. The optimized predictive surrogate can greatly facilitate tasks that typically require extensive access to the quantum processor, e.g., the pre-training of variational quantum algorithms and the characterization of phase diagrams in quantum many-body systems. \textbf{b.} Overview of the employed quantum processor. The left panel illustrates the 2D architecture of the superconducting quantum processor used in this study. The blue circles highlight the employed subset of superconducting qubits in experiments, while the white circles represent the remaining idle qubits on the chip. The right panel indicates the cumulative probability distributions of Pauli noise errors for parallel single-qubit gates (blue) and parallel CZ gates (orange) measured on the employed quantum processor. Besides, the dashed vertical lines denote the mean values. \textbf{c.} The construction of the predictive surrogates consists of two steps. First, the classical learner collects the training dataset, i.e., labeled data obtained from quantum processors. Then, the classical learner leverages the trigonometric monomial expansion on the collected training dataset to implement the effective predictive surrogates.}}
	\label{fig:scheme}
\end{figure*}

\section{Introduction} 
Modern quantum processors with hundreds of physical qubits~\cite{bluvstein2024logical,Google2024Quantum,gao2024establishing} have evidenced their scientific value across various domains such as tackling synthetic sampling problems~\cite{arute2019quantum,zhong2020quantum,wu2021strong,morvan2024phase}. At the same time, advancements in quantum error correction are steadily pushing the boundaries of fault-tolerant quantum computing~\cite{abobeih2022fault,krinner2022realizing,sivak2023real,google2023suppressing,bausch2024learning}. Despite the progress, the high cost of fabrication and maintenance hints that access to these advanced quantum computers will remain scarce for the foreseeable future. As a result,  quantum-resource-efficient alternatives capable of emulating the behavior of large-scale quantum processors towards practical problems are highly demanded. A critical direction in this pursuit is the development of effective and scalable methods to emulate the mean-value behavior of quantum processors, given its broad applications, including quantum physics~\cite{mi2021information,mi2022time,semeghini2021probing,andersen2025thermalization}, quantum chemistry~\cite{cao2023generation,o2023purification,guo2024experimental}, quantum system certification~\cite{eisert2020quantum}, and quantum machine learning~\cite{biamonte2017quantum,abbas2024challenges}.

Existing strategies to emulate the mean-value behavior of large-scale quantum processors can be broadly categorized into two classes: classical simulators and learning-based surrogates. Concrete examples of classical simulation encompass tensor-network~\cite{markov2008simulating,villalonga2019flexible,cirac2021matrix,pan2022simulation}, near-Clifford~\cite{bravyi2016improved,bravyi2019simulation,beguvsic2023fast}, and Pauli-path simulators~\cite{rudolph2023classical,bermejo2024quantum,angrisani2024classically,angrisani2025simulating}. However, these classical simulators share a fundamental limitation, i.e., they are incapable of accurately and effectively capturing the intrinsic noise model specific to the quantum processor being explored. As a result, a non-negligible discrepancy may arise between the simulated results and the actual performance of the quantum processor. For learning-based surrogates, most of them adopt deep learning techniques to train neural networks that emulate the behavior of quantum processors~\cite{wu2023quantum,qin2024experimental,qian2024multimodal,zhao2025rethink}. While these approaches have shown exceptional empirical results, they remain largely heuristic, lacking theoretical guarantees on their efficiency. This raises a critical question: 
\begin{center}
	\textit{Can we design provably efficient surrogates that emulate the mean-value behavior of quantum processors?}
\end{center}

Here we provide an affirmative answer to this question for estimating mean values of a quantum circuit $U(\bx)$, which consists of a gate fixed layout and $d$-dimensional input parameters $\bx$. Our focus is designing efficient predictive surrogates to emulate the mean value $\Tr(\tilde{\rho}(\bx)O)$, where $O$ is an observable and the noisy state $\tilde{\rho}(\bx)$ is generated by $U(\bx)$ on an assigned quantum processor. As shown in FIG.~\ref{fig:scheme}, given only limited access to the quantum processor, the optimized predictive surrogates can accurately estimate the mean values $\Tr(\tilde{\rho}(\bx')O)$ for any new input $\bx'$ via purely classical inference. Note that many applications rely on mean-value estimation of this form. To provide concrete solutions, we devise \textit{two provably efficient predictive surrogates} tailored to distinct scenarios, and experimentally exhibit their efficacy in digital quantum simulation~\cite{fauseweh2024quantum} using a quantum processor with up to $20$ programmable superconducting qubits.

The first predictive surrogate, denoted by $h_{\textsf{cs}}$, extends the model from Ref.~\cite{du2024efficient} to the noisy case, where it serves as a classical shadow predictor~\cite{elben2022randomized}. Specifically, $h_{\textsf{cs}}$ is able to estimate mean values of $\tilde{\rho}(\bx)$ for many local observables under Pauli noise channels~\cite{terhal2015quantum}. Our theoretical analysis reveals that when $\bx$ are independently and uniformly sampled from $[-\pi, \pi]^d$, $h_{\textsf{cs}}$ is computationally efficient to achieve an $\epsilon$-level prediction error, provided that either Pauli noise rates are relatively large or the gradient norm $\|\nabla_{\bx}\Tr(\tilde{\rho}(\bx)O)\|$ is upper bounded by a small constant. We experimentally assess the performance of $h_{\textsf{cs}}$ in pre-training variational quantum eigensolvers (VQEs) \cite{tilly2022variational} when applied to find the ground state energies of a class of 1D transverse-field Ising models (TFIMs). A notable feature of pre-training VQEs with the optimized $h_{\textsf{cs}}$ is that it enables fully classical estimation of ground state energies for all TFIMs, substantially reducing measurement overhead relative to conventional VQEs. Experimental results indicate that the ground state energies estimated by $h_{\textsf{cs}}$ are more accurate than those obtained from VQEs trained directly on the quantum processor, while using just $0.023\%$ of the measurements.

The second predictive surrogate, denoted by $h_{\mathsf{qs}}$, is designed to predict mean values $\Tr(\tilde{\rho}(\bx)O)$ with $\bx$ being correlated and a fixed observable $O$. This setting is commonly encountered in the study of quantum simulation. Different from $h_{\textsf{cs}}$, $h_{\mathsf{qs}}$ supports the estimation of mean values for input parameters $\bx$ sampled from \textit{arbitrary distributions}. We prove that when the quantum processor is modeled by Pauli noise channels and $\bx$ is sampled from a small range, $h_{\mathsf{qs}}$ is computationally efficient to reach an $\epsilon$-level prediction error. We experimentally explore the potential of $h_{\mathsf{qs}}$ by applying it to the identification of Floquet symmetry-protected phases, a distinct type of non-equilibrium states of matter~\cite{khemani2016phase}. In contrast to the original approach~\cite{zhang_digital_2022}, which demands extensive quantum processor access to estimate mean values across different $\bx$ for phase identification, our method enables orders-of-magnitude reduction in measurement overhead by training a set of $\{h_{\mathsf{qs}}\}$ to emulate mean-value behavior. Experimental results demonstrate that the optimized $\{h_{\mathsf{qs}}\}$ can reliably probe the transition region of Floquet symmetry-protected phases of time-periodic Hamiltonians. These results open a new avenue for leveraging predictive surrogates to advance the study of digital quantum simulation.

\section{Main results}\label{sec:results}

\subsection{Problem setup} 
Let us begin by establishing the mathematical framework that characterizes the mean-value behavior of noisy quantum processors. Leveraging the universality of Clifford gates with $\RZ$ gates \cite{ross2014optimal}, throughout this study, we focus on quantum circuits of the form $U(\bx) = \prod_{l=1}^{d}(\RZ(\bx_l)V_l)$, where each Clifford operation $V_l$ is composed of some Clifford gates (abbreviated as $\CI$ gates) and identity gates $\mathbb{I}_2$ with $\CI=\{H, S, \CNOT\}$. 

To account for noise, the quantum circuit $U(\bx)$ and its constituent gates are often described in terms of their quantum channel representations. Specifically, we denote the channel corresponding to $U(\bx)$ as $\mathcal{U}(\bx)$, and the channels corresponding to the individual gates $\RZ(\bx_l)$ and $V_l$ as $\mathcal{R}_{\mathcal{Z}}(\bx_l)$ and $\mathcal{V}_l$, respectively. Supported by the results of twirling operation~\cite{bennett1996mixed}, we model $\mathcal{U}({\bx})$ as being affected by Pauli noise channels~\cite{chen2024tight}, i.e., 
\begin{equation}\label{eqn:noisy-unitary-channel}
	\tilde{\mathcal{U}}({\bx}) =  \bigcirc_{l=1}^d \tilde{\mathcal{R}}_{\mathcal{Z}}(\bx_l) \circ \tilde{\mathcal{V}}_{l},
\end{equation}
where $\tilde{\mathcal{R}}_{\mathcal{Z}}(\bx_l)= \mathcal{N}_P\circ \mathcal{R_Z}(\bx_l)$, $\tilde{\mathcal{V}}_{l}=\mathcal{M}\circ \mathcal{V}_l$, and $\mathcal{N}_{P}$ and $\mathcal{M}$ denote the single-qubit and multi-qubit Pauli channels, respectively. Here $\mathcal{N}_{P}(\rho) = (1-\bar{p})\rho + \PX X\rho X+ \PY Y\rho Y + \PZ Z\rho Z$ with $\PX$, $\PY$, and $\PZ$ being the noise parameters and $\bar{p}=\PX+\PY+\PZ$. With these definitions, the mean-value behavior of a noisy quantum processor is characterized by
\begin{equation}\label{eqn:noisy-mean}
	f(\tilderho(\bx),O) = \Tr(\tilderho(\bx) O),
\end{equation}
where $O$ is a given observable and $\tilderho(\bx)$ denotes the quantum state $\rho_0$ evolved under the noisy circuit $\tilde{\mathcal{U}}({\bx})$.

The objective of the predictive surrogate is illustrated in FIG.~\ref{fig:scheme}\text{a}. That is, it is designed to learn a classical model $h(\bx, O)$ whose output approximates the true mean-value $f(\tilderho(\bx), O)$. According to statistical learning theory~\cite{vapnik1999nature}, the quality of the predictive surrogate can be quantified by measuring how closely its predictions match the true mean-value behavior. A standard metric for this purpose is the average prediction error over the data distribution $\mathbb{D}$, i.e.,
\begin{equation}\label{eq:prediction_error}
	\mathsf{R}(h) = \mathbb{E}_{\bx\sim \mathbb{D}} \left| h(\bx, O) - f(\tilderho(\bx),O)  \right|^2.
\end{equation}	
The predictive surrogate $h$ is considered computationally efficient if it can be trained in polynomial time with respect to the qubit count $N$ and the input dimension $d$, and with high probability, satisfies  $\mathsf{R}(h) \leq \epsilon$. 

Prior studies have shown that constructing efficient predictive surrogates for arbitrary quantum circuits $U(\bx)$ is unlikely, as ruled out by complexity-theoretic constraints~\cite{du2024efficient}. Motivated by this, we pursue efficient predictive surrogate design under well-defined and practically relevant scenarios, and accordingly devise two provably efficient methods, denoted by $h_{\mathsf{cs}}$ and $h_{\mathsf{qs}}$.

\subsection{Predictive surrogate $h_{\mathsf{cs}}$}
The first predictive surrogate $h_{\mathsf{cs}}$ is designed for circuits $U(\bx)$ with independently tunable $\RZ$ gates and supports arbitrary, yet $K$-local and bounded, observables $\{O\}$. The observable $O$ is said to be $K$-local if it can be written as a sum of tensor products of Pauli operators, where each term acts nontrivially on at most $K$ qubits and trivially on the rest. Additionally, the norm of the observable is bounded with $\|O\|_{\infty}\le B$. In this sense, $h_{\mathsf{cs}}$ functions as a \textit{classical shadow predictor}, capable of estimating mean values of many local observables for previously unseen quantum states $\tilderho(\bx’)$ with $\bx'\sim \mathbb{D}$, all without requiring further access to quantum processors. 

The implementation of $h_{\mathsf{cs}}$ contains two steps: (i) data collection and (ii) model construction. In Step~(i), the learner inputs various $\bxi \in [-\pi, \pi]^d$ to $\tilde{\mathcal{U}}({\bxi})$ and collects classical information of the evolved state $\tilderho(\bxi)$ under Pauli-based classical shadow~\cite{huang2020predicting} with $T$ snapshots, denoted by $\tilderho_T(\bxi)$. Following this routine, the training dataset $\mathcal{T}_{\mathsf{cs}}=\{(\bxi, \tilderho_T(\bxi))\}_{i=1}^n$ with $n$ examples is constructed. Then, in Step~(ii), the learner utilizes $\mathcal{T}_{\mathsf{cs}}$ to build the predictive surrogate. Given a new input $\bx'$, the predicted classical shadow of the state $\rho(\bx')$ yields 
\begin{equation}\label{eqn:generic-learner-state}
	\hatsigma(\bx') = \frac{1}{n}\sum_{i=1}^n\kappa_{\Lambda}\left(\bx', \bxi \right)\tilderho_T(\bxi),
\end{equation}  
where $\kappa_{\Lambda}(\bx, \bxi)=\sum_{\bomega, \|\bomega\|_0 \leq \Lambda} 2^{\|\bomega\|_0}\Phi_{\bomega}(\bx)\Phi_{\bomega}(\bxi)$ is the \textit{truncated trigonometric monomial kernel} with the basis $\Phi_{\bomega}(\bx) = \prod_{i=1}^d  [\mathbbm{1}_{\bomega_i=0} + \mathbbm{1}_{\bomega_i=1}\cdot \cos(\bx)+\mathbbm{1}_{\bomega_i=-1}\cdot \sin(\bx)]$. Refer to Supplementary Material (SM)~A for more details.

Supported by $\hatsigma(\bx')$, the predictive surrogate over the specified observable $O$ is
\begin{equation}
    \label{eqn:generic-learner}
  h_{\mathsf{cs}}(\bx', O) = \frac{1}{n}\sum_{i=1}^n\kappa_{\Lambda}\left(\bx', \bxi \right)g(\bxi,O), 
\end{equation}    
where $g(\bxi,O)=\Tr(\tilderho_T(\bxi)O)$ refers to the shadow estimation of $\Tr(\tilderho(\bxi) O)$. The following theorem indicates that even in the noisy scenario, the proposed predictive surrogate can accurately approximate the mean value $f(\tilderho(\bx),O)$ in Eq.~(\ref{eqn:noisy-mean}), where the formal statement and the proof are deferred to SM~B.

\begin{theorem}[Informal]\label{thm:learning-non-trunc}
 Following notations in Eqs.~(\ref{eqn:noisy-unitary-channel})-(\ref{eqn:generic-learner}), denote $\mathfrak{C}(\Lambda) =\{\bomega|\bomega \in \{0, \pm 1\}^d, ~s.t.~\|\bomega\|_0\leq \Lambda\}$ as the truncated frequency set with $\Lambda$ being the threshold value. Suppose that the input $\bx\sim \Unif[-\pi,\pi]^d$ follows the uniform distribution and $\mathbb{E}_{\bx}\|\nabla_{\bx} \Tr(\tilderho(\bx)O)\|_2^2\leq C$. Then, under the Pauli noise channel $\mathcal{N}_P(p_X,p_Y,p_Z)$ with $p=\min\{p_X,p_Y\}$, when  $ \Lambda=4C/\epsilon$ and the number of training examples is \[n = \tilde{\Omega}\Big(\big|\mathfrak{C}\big(\min \big\{\frac{4C}{\epsilon}, \frac{1}{2(p+p_Z)} \log\big(\frac{2B}{\sqrt{\epsilon}}\big)\big\}\big)\big| \frac{2  B^2 9^K}{\epsilon}\Big),\]  with probability at least $1-\delta$, we have $\mathsf{R}(h_{\mathsf{cs}})\leq \epsilon$ even for $T=1$ snapshot per training example.	
\end{theorem}

The results suggest that the sample complexity $n$ does not explicitly depend on the qubit count $N$ and polynomially scales with $d$ for small $C$. Ref.~\cite{du2024efficient} has empirically shown that the quantity $C$ typically decreases with the increased $N$ and the circuit depth. In the extreme case of barren plateaus, $C$ becomes exponentially small~\cite{larocca2025barren}. This favorable scaling behavior ensures the efficacy of $h_{\mathsf{cs}}$ when applied to emulate scalable quantum processors. 

In addition, the collected classical shadows $\{\tilderho_T(\bxi)\}$ inherently encode the noise characteristics of the explored quantum processor. This allows $h_{\mathsf{cs}}$ to accurately emulate the mean values of noisy quantum states without explicitly characterizing the underlying Pauli noise. This distinguishes our approach from classical simulators~\cite{beguvsic2023fast,tindall2024efficient,fontana2023classical}, which require expensive computational overhead to achieve the same goal~\cite{chen2024tight}.  

Moreover, different from the ideal case, the scaling of $n$ is dominated by the noise parameters, as reflected in the term $\min \{4C/\epsilon, 1/(2(p+p_Z))\}$. As a result, higher levels of Pauli noise reduce both $1/(2(p+p_Z))$ and $C$~\cite{wang2021noise}, which in turn decreases the required $n$. This observation highlights the potential of  $h_{\mathsf{cs}}$ for advancing the practical utility of near-term quantum devices.

\noindent\textbf{Remark}. The achieved results in Theorem~\ref{thm:learning-non-trunc} can be effectively extended to a broader class of quantum circuits with arbitrary Pauli rotation gates and Clifford gates, as explained in SM~A. In SM~E, we prove the computational efficiency of $h_{\mathsf{cs}}$, scaling with $\mathcal{O}(n^2NT)$.

\subsection{Predictive surrogate $h_{\mathsf{qs}}$}
The second predictive surrogate $h_{\mathsf{qs}}$ is tailored for $U(\bx)$ with correlated parameters among $\RZ$ gates. This setting is commonly adopted in quantum simulations relevant to many-body physics~\cite{fauseweh2024quantum}  and quantum chemistry~\cite{mcardle2020quantum}, as well as in quantum machine learning~\cite{zhang2022gaussian}.
Specifically, the quantum circuit in Eq.~(\ref{eqn:noisy-unitary-channel}) yields a layer-wise layout, i.e., $U(\bx)=\prod_{l=1}^{L}(\prod_{j=1}^{d/L}\RP(\bx_j)V_j)$, where $\RP$ refers to an arbitrary Pauli rotation gate, and the entries of $\bx\in \mathbb{R}^d$ vary within each layer but remain the same across different layers. Without loss of generality, it is sufficient to set $\RP=\RZ$. The surrogate $h_{\mathsf{qs}}$ aims to predict the mean value $f(\tilderho(\bx'),O)$ in Eq.~(\ref{eqn:noisy-mean}) given a predefined  observable $O$ and any new input $\bx'$.

The implementation of $h_{\mathsf{qs}}$ contains the same two steps as with $h_{\mathsf{cs}}$, i.e., data collection and model construction, but differs in their execution. First, the learner builds the training dataset $\mathcal{T}_{\mathsf{mb}}=\{(\bxi,\yi)\}_{i=1}^n$. The inputs $\bxi$ are sampled from an arbitrary distribution within the bounded interval $[-R,R]^d$, and $\yi$ is the estimated mean value of the observable $O$ obtained by $T$ shots. Here we suppose that the maximal estimation error satisfies $\max_{i\in [n]} |\yi-\Tr(\tilderho(\bxi)O)| \le \epsilon_l$. Then, the predictive surrogate is defined by 
\begin{equation}\label{eqn:surrogate_correlated_para}
	h_{\mathsf{qs}}(\bx,\widehat{\bm{\mathrm{w}}})=\braket{\bm{\Phi}_{\mathfrak{C}(\Lambda)}(\bx), \widehat{\bm{\mathrm{w}}}},
\end{equation}
 where $\widehat{\bm{\mathrm{w}}}$ is obtained by minimizing the  loss function with $\min_{\bm{\mathrm{w}}} \frac{1}{n}\sum_{i=1}^n \left(\yi- h_{\mathsf{qs}}(\bxi,\bm{\mathrm{w}})\right)^2 + \lambda \|\bm{\mathrm{w}}\|_2$. The notation $\bm{\Phi}_{\mathfrak{C}(\Lambda)}=[\Phi_{\bomega}(\bx)]_{\bomega \in \mathfrak{C}(\Lambda)}$ represents a feature vector with $|\mathfrak{C}(\Lambda)|$ dimensions formed by the truncated trigonometric monomial $\Phi_{\bomega}(\bx)$, and $\lambda$ refers to the regularization parameter.  For convenience,  we denote the optimized surrogate $h_{\mathsf{qs}}(\bx, \widehat{\bm{\mathrm{w}}})$ as $h_{\mathsf{qs}}(\bx)$ or $h_{\mathsf{qs}}$, whenever no ambiguity arises. The subsequent theorem provides a provable guarantee for the efficiency of  $h_{\mathsf{qs}}$, where the formal description and proof are provided in SM~C.

\begin{theorem}[Informal]\label{thm:ridge_predict_error}
 Following notations in Eqs.~(\ref{eqn:noisy-unitary-channel})-(\ref{eqn:surrogate_correlated_para}), let $q=1-2(p+\PZ)$ and $\epsilon=16B^2(deq(1+R)/\Lambda)^{2\Lambda}$. Suppose $q(1+R)<1/e$ and $\epsilon_l\le\sqrt{\epsilon}/4$. Then, when the frequency is truncated to $ \Lambda> deq(1+R)$ and the number of training examples is \[n=\left(\frac{1}{q(1+R)}\right)^{4deq(1+R)}\cdot \frac{\log(1/\delta)}{9},\]with probability at least $1-\delta$, we have $\mathsf{R}(h_{\mathsf{qs}})\leq \epsilon$. 
\end{theorem}

The achieved results indicate the capabilities of $h_{\mathsf{qs}}$ when applied to approximate  $f(\tilderho(\bx),O)$ with correlated $\bx$. Concretely, the sample complexity and computational complexity of implementing $h_{\mathsf{qs}}$ scale polynomially with the input dimension $d$ when the noise parameters satisfy $q \le \mathcal{O}(1/d)$. This additional requirement compared to $h_{\mathsf{cs}}$ stems from the fact that the former considers more general settings. Different from $h_{\mathsf{cs}}$ designed for the uncorrelated $\bx$ with a uniform input distribution, $h_{\mathsf{qs}}$ is tailored to accommodate the correlated $\bx$ and arbitrary input distributions, making it applicable to a broader range of scenarios. This increased flexibility comes at the expense of involving additional conditions.

\noindent \textbf{Remark.}  In SM~D, we provide an alternative solution of $h_{\mathsf{qs}}$, which can efficiently emulate the mean-value behavior of $f(\tilderho(\bx),O)$ with low noise levels (i.e., $q>\mathcal{O}(1/d)$), provided that the range $R$ is small.

\subsection{Applications to digital quantum simulation} 
A crucial feature of predictive surrogates is that once trained, they can predict mean values during inference without any further access to the quantum processor. As a result, when the task is to estimate mean values across diverse $\{\rho(\bx)|\bx\}$, the number of required quantum measurements, and thereby overall quantum processor usage, can be greatly reduced. This advantage is most evident in measurement-intensive applications. To illustrate this, we next present how $h_{\mathsf{cs}}$ and $h_{\mathsf{qs}}$ can be employed to advance two crucial tasks in digital quantum simulation, i.e., variational quantum eigensolvers (VQEs)~\cite{tilly2022variational} and quantum phase identification, where more details are given in SM~F.
 
VQEs leverage the variational principle to compute the ground state energy of quantum systems. Given a Hamiltonian $\mathsf{H}$, VQE adopts the parameter shift rule~\cite{schuld2019evaluating} to iteratively update $\bx$ to minimize $f(\tilderho(\bx), \mathsf{H})$ in Eq.~(\ref{eqn:noisy-mean}), often incurring tremendous access to quantum processors. As indicated in Ref.~\cite{du2024efficient}, a well-optimized $h_{\mathsf{cs}}$ can alleviate the expensive measurement overhead via \textit{pre-training}. Instead of directly optimizing $f(\tilderho(\bx), \mathsf{H})$, the targeted parameters can be obtained by minimizing $h_{\mathsf{cs}}(\bx, \mathsf{H})$, i.e.,
\begin{equation}\label{eqn:pre-train-VQE}
	  \min_{\bx} h_{\mathsf{cs}}(\bx, \mathsf{H}).
\end{equation}
This process is entirely classical and does not involve any access to quantum processors. Given access to the optimized $\hat{\bx}$, the unitary $U(\hat{\bx})$ can, if desired, be deployed on quantum processors for additional fine-tuning.   

\begin{figure*}[t]
	\centering\includegraphics[width=1\textwidth]{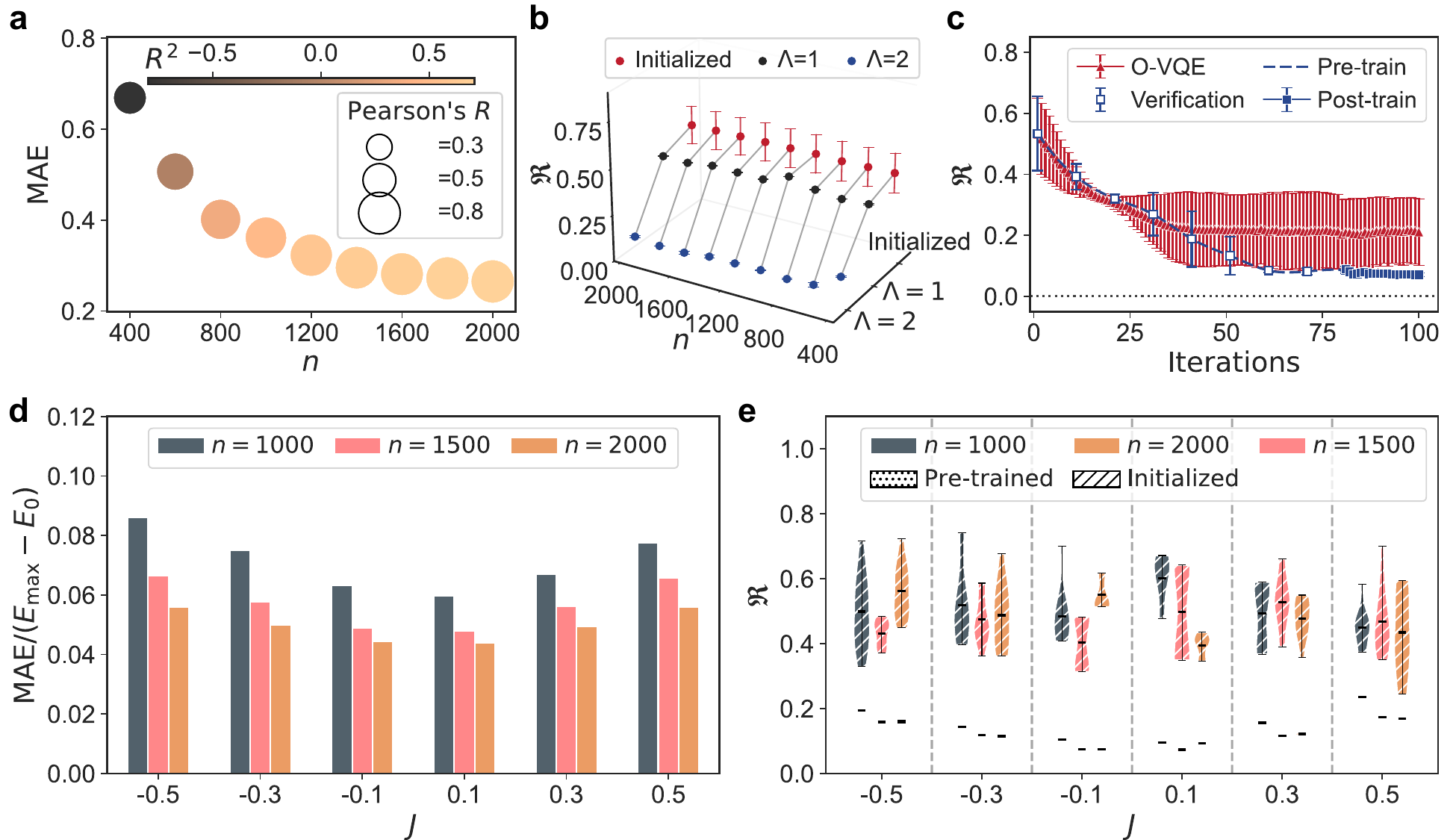}
	\caption{\small{\textbf{Experimental results of using the predictive surrogate $h_{\mathsf{cs}}$ to pre-train VQE for a family of 1D Transverse Field Ising Models.}} \textbf{a.} Quantifying the difference between the outputs of the employed quantum processor and the predictive surrogate in the measures of mean absolute error (MAE), coefficient of determination ($R^2$), and Pearson correlation coefficient ($R$). \textbf{b.} Normalized deviation $\mathfrak{R}(\bx)$ of the estimated energy $f(\tilderho(\bx), \mathsf{H}_{\mathsf{TFIM}})$ from the ground state energy $E_0(\mathsf{H}_{\mathsf{TFIM}})$, evaluated on both randomly chosen parameters (highlighted by red dots with the label `Initialized') and parameters optimized by the predictive surrogate $h_{\mathsf{cs}}$ with varying the size of the training dataset $n$ and the truncation threshold $\Lambda$ (i.e., $\Lambda=1,2$ are highlighted by the black and blue dots, respectively). \textbf{c.} Optimization dynamics of the normalized deviation $\mathfrak{R}(\bx)$ over 100 iterations for the predictive surrogate $h_{\mathsf{cs}}$ compared to the original VQE, corresponding to the labels of `Pre-train' and `O-VQE'. The blue solid box, labeled  `Post-train', refers to the optimization dynamic of the additional fine-tuning on the quantum processor. The blue hollow box, labeled `Verification', indicates the mean values measured on the quantum processor at the parameters optimized by the predictive surrogate.   \textbf{d.} Normalized MAE between quantum processor outputs and predictions from three optimized $h_{\mathsf{cs}}$, each trained on datasets of varied sizes (i.e., $n=1000, 1500, 2000$), evaluated for a set of  $\mathsf{H}_{\mathsf{TFIM}}$ with varying $J\in \{\pm 0.1,\pm 0.3, \pm 0.5\}$ and fixed $h=-0.5$. \textbf{e.} Normalized deviation $\mathfrak{R}(\bx)$ for randomly chosen parameters (highlighted by the hatched box with the label `Initialized') and parameters optimized by the predictive surrogate (highlighted by the dotted box with the label `Pre-trained') when $J$ in $\mathsf{H}_{\mathsf{TFIM}}$ ranges from $-0.5$ to $0.5$ and fixed $h=-0.5$.  All the results presented in the subplot are statistical data collected from $20$ independent experiments, with the standard deviations highlighted by shaded regions and error bars.}
	\label{fig:fig2}
\end{figure*}
 
The predictive surrogate can also facilitate the identification of phase transitions in quantum many-body systems. A representative example is identifying Floquet symmetry-protected topological (FSPT) phases~\cite{zhang_digital_2022}, which are exotic non-equilibrium states of matter arising in periodically driven systems~\cite{ponte2015periodically}. Define a 1D spin-$1/2$ chain governed by the time-periodic Hamiltonian as
\begin{equation}
	\mathsf{H}_{\mathsf{tp}}(\mathrm{t}) = \begin{cases}
	\mathsf{H}_{\mathsf{tp},1} = (\frac{\pi}{2}-\delta)\sum_{i}X_i, & 0 \leq \mathrm{t} < \mathrm{T}_1, \\
	\mathsf{H}_{\mathsf{tp},2} = -\sum_i J_iZ_{i-1}X_i Z_{i+1}, & \mathrm{T}_1 \leq t < 2\mathrm{T}_1.
			\end{cases} \nonumber
\end{equation}
Here $P_i\in \{X, Y, Z\}$ refers to the Pauli operator on the $i$-th qubit. Through varying the drive perturbation $\delta$, two dynamical phases emerge:  FSPT phase and thermal phase. These two phases can be characterized by the local magnetization $\braket{Z_i(\bx)}=\bra{0^N}U(\bx)^{\dagger}Z_iU(\bx) \ket{0^N}$, where the gate layout of $U(\bx)$ depends on $\mathsf{H}_{\mathsf{tp}}(\mathrm{t})$, and $\bx$ is determined by the specified disordered coupling parameters $\{J_i\}$, $t$, and $\delta$. Refer to SM~E for more details about the circuit implementation. In this regard, identifying FSPT phases typically requires a large number of measurements to estimate $\braket{Z_i(\bx)}$ with varied $\bx$. The proposed $h_{\mathsf{qs}}$ can effectively address this bottleneck.

\section{Experimental results} 
We experimentally validate the efficacy of the proposed predictive surrogates in advancing digital quantum simulation. All experiments are carried out on a contiguous $20$-qubit chain within a 66-qubit programmable superconducting processor, which comprises a $6\times 11$ array of frequency-tunable transmon qubits with tunable couplers. See SM~G for a full device characterization and the omitted experimental results.

The first task is to explore the potential of $h_{\mathsf{cs}}$ to enhance VQE optimization in finding the ground state energy of 1D transverse field Ising models (TFIMs). The formal expression is $\mathsf{H}_{\mathsf{TFIM}}=-J\sum_iZ_iZ_{i+1}-h\sum_iX_i$. The initial state is $\rho_0=(\ket{+}\bra{+})^{\otimes N}$ and the  adopted ansatz in VQE follows the Trotter decomposition, i.e., \[U(\bx)=\prod_{l=1}^L \Big(\prod_{i=1}^N  \RX_i(\bx_{l,i}) \prod_{i=1}^{N-1}\RZZ_{i,i+1}(\bx_{l,N+i})\Big),\] 
where the input dimension is $d=L(2N-1)$. 

Our first subtask is to quantify the performance of $h_{\mathsf{cs}}$ in emulating the mean-value behavior of the VQE under study. The TFIM considered here is set to $N=6$, $J=-0.1$, and $h=-0.5$. Moreover, the dataset $\mathcal{T}_{\mathsf{cs}}$ is constructed by uniformly sampling  $\bx$, where the number of snapshots for each state $\rho(\bx)$ is fixed to be $T=10$. The threshold of truncation is set to $\Lambda=2$. The performance of $h_{\mathsf{cs}}(\bx)$ over $200$ test examples is shown in FIG.~\ref{fig:fig2}\text{a}. In particular, under the measures of the mean absolute error (MAE), Pearson correlation coefficient, and $R^2$ score, the optimized $h_{\mathsf{cs}}(\bx)$ exhibits a strong agreement with  $f(\tilde{\rho}(\bx), \mathsf{H}_{\mathsf{TFIM}})$ given $n \geq  1600$. These results not only validate the effectiveness of $h_{\mathsf{cs}}(\bx)$, but also serve as a crucial prerequisite for successfully pre-training VQE.

We now proceed to examine the ability of the optimized $h_{\mathsf{cs}}$ in estimating the ground state energy of $\mathsf{H}_{\mathsf{TFIM}}$. To this end, we use the normalized deviation as a metric to quantify its performance, i.e., 
\[\mathfrak{R}(\bx) = \frac{|f(\tilderho(\bx),\mathsf{H}_{\mathsf{TFIM}})-E_0(\mathsf{H}_{\mathsf{TFIM}})|}{(E_{\max}(\mathsf{H}_{\mathsf{TFIM}})-E_0(\mathsf{H}_{\mathsf{TFIM}}))},\] 
where $\bx$ refers either to the initialized inputs or the optimized parameters returned by $h_{\mathsf{cs}}$. Here  $E_0(\mathsf{H}_{\mathsf{TFIM}})$ and the dominator refer to the ground state energy and the spectrum width of $\mathsf{H}_{\mathsf{TFIM}}$, respectively. We repeat the evaluation $20$ times, each time using the same set of initialized parameters, followed by the parameter optimization described in Eq.~(\ref{eqn:pre-train-VQE}). As depicted in FIG.~\ref{fig:fig2}\text{b}, in every case, the average $\mathfrak{R}(\bx)$ before pre-training remains near $0.482$, regardless of $n$ used to implement $h_{\mathsf{cs}}$. In addition, when $\Lambda=1$, the estimated  $f(\tilde{\rho}(\bx), \mathsf{H}_{\mathsf{TFIM}})$ with $\hat{\bx}$ being return by pre-training $h_{\mathsf{cs}}$ still suffer from a high $\mathfrak{R}(\hat{\bx})$ with all above $0.42$. In contrast, when $\Lambda=2$ and $n>600$,  $\mathfrak{R}(\hat{\bx})$ is less than $0.09$ after pre-training. These results hint that when  $\Lambda$ surpasses a critical value, its performance improves with a large sampling size $n$, aligning with theoretical results in Theorem~\ref{thm:learning-non-trunc}. 

To further comprehend the potential of pre-training VQE, we compare the optimization dynamics of the predictive surrogate with those of the original VQE. Specifically, the explored $\mathsf{H}_{\mathsf{TFIM}}$ in this analysis is identical to that employed in the previous subtask, and the hyperparameter settings of $h_{\mathsf{cs}}$ are $n=2000$, $T=10$, and $\Lambda=2$. The learning rate is set as 0.1, and the total number of iterations employed in pre-training VQE and the original VQE is fixed to be $100$. We repeat each setting 5 times to acquire the statistical results. As shown in FIG.~\ref{fig:fig2}\text{c}, after the pre-training phase, the VQE already outperforms the original VQE, yielding $\mathfrak{R} = 0.09$ compared to $0.21$. Remarkably, this improvement is attained while using only 0.023\% of the quantum measurement resources required by the original VQE (Refer to SM G for details). More interestingly, the obtained $\mathfrak{R}$ during the pre-training stage is near optimal, while only marginal gains are observed when further fine-tuning $\hat{\bx}$ on the quantum processor, improving $\mathfrak{R}$ from $0.09$ to $0.07$.
 
\begin{figure}[t]
\centering\includegraphics[width=0.4\textwidth]{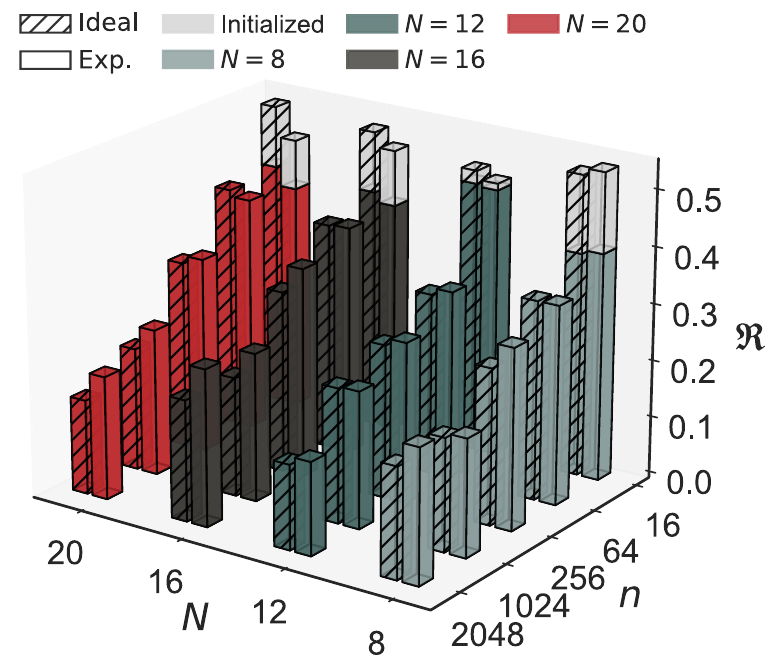}
\caption{\small{\textbf{Scaling behavior of $h_{\mathsf{cs}}$ when applied to pre-train VQE.}} Normalized deviation $\mathfrak{R}(\bx)$ evaluated on the initial parameters before pre-training (highlighted by the gray transparent bars with the label `Initialized') and after pre-training (highlighted by the colored bars) across varying qubit counts $N\in\{8,12,16,20\}$ and training dataset size $n\in\{2^4,2^6,2^8,2^{10},2^{11}\}$, with results aggregated over $20$ random instances per experimental setting. The unhatched bars with label `Exp.' and hatched bars with label `Ideal' denote quantum processor outputs and ideal expectation value obtained from the noiseless numerical simulation, respectively.}
\label{fig:fig3}
\end{figure}

To assess the performance of the optimized $h_{\mathsf{cs}}$ in estimating the ground state energies of many $\mathsf{H}_{\mathsf{TFIM}}$ without any access to the quantum processor,  we undertake a new subtask. The TFIMs considered here span six different settings with $J\in\{-0.5, -0.3, -0.1, 0.1, 0.3, 0.5\}$, with $h$ being fixed at $-0.1$. The predictive surrogates with $\Lambda=2$ are separately optimized using three different training set sizes, i.e., $n\in\{1000, 1500, 2000\}$. All other hyperparameter settings remain identical to those used in previous subtasks.  For each optimized predictive surrogate, the pre-training stage is terminated upon convergence. Experimental results are visualized in FIG.~\ref{fig:fig2}\text{e}. For all TFIMs, the optimized surrogates achieve a lower $\mathfrak{R}$ after pre-training compared to their initial values.   Furthermore, the surrogates trained with $n=1500$ and $n=2000$ exhibit comparable performance, suggesting diminishing returns beyond a certain training set size. These findings confirm the capability of $h_{\mathsf{cs}}$ in substantially reducing the measurement overhead for VQE optimization, especially as the number of Hamiltonians grows. 

We last explore the scalability of the predictive surrogates for pre-training VQE. The TFIMs considered have varying numbers of qubits ($N=8, 12, 16, 20$), while the coefficients $J$ and $h$ are fixed at $-0.1$ and $-0.5$ for each Hamiltonian. Besides, the number of training examples for each Hamiltonian has five distinct settings, which are $n=2^4, 2^6, 2^8, 2^{10}, 2^{11}$.  All other hyperparameters employed in the pre-training stage are almost identical to the prior subtasks, except for the settings of the snapshot $T$. That is, for $N=8$, $12$, $16$, and $20$, the number of snapshots $T$ per example is set to $20$, $20$, $50$, and $100$, respectively. As depicted in FIG.~\ref{fig:fig3}, the results show that the predictive surrogates achieve comparable performance across different Hamiltonians.  Specifically, the performance of $h_{\mathsf{cs}}$ depends more strongly on the factors $n$ and  $\Lambda$ than on the number of qubits $n$. Besides, for each setting, we compare results from the noisy scenario and the ideal scenario, where for the latter $\mathfrak{R}(\bx)$ is computed by $f(\rho(\bx),\mathsf{H}_{\mathsf{TFIM}})$ instead of $f(\tilderho(\bx),\mathsf{H}_{\mathsf{TFIM}})$. The minimal difference observed between the noisy and ideal cases (up to $N=16$, with $\text{MSE}$ peaking at 0.0013) not only reflects the high fidelity of our quantum processor, but also underscores the practical utility of predictive surrogates in pre-training VQEs.
 
\smallskip

\begin{figure*}[t]
	\centering\includegraphics[width=1\textwidth]{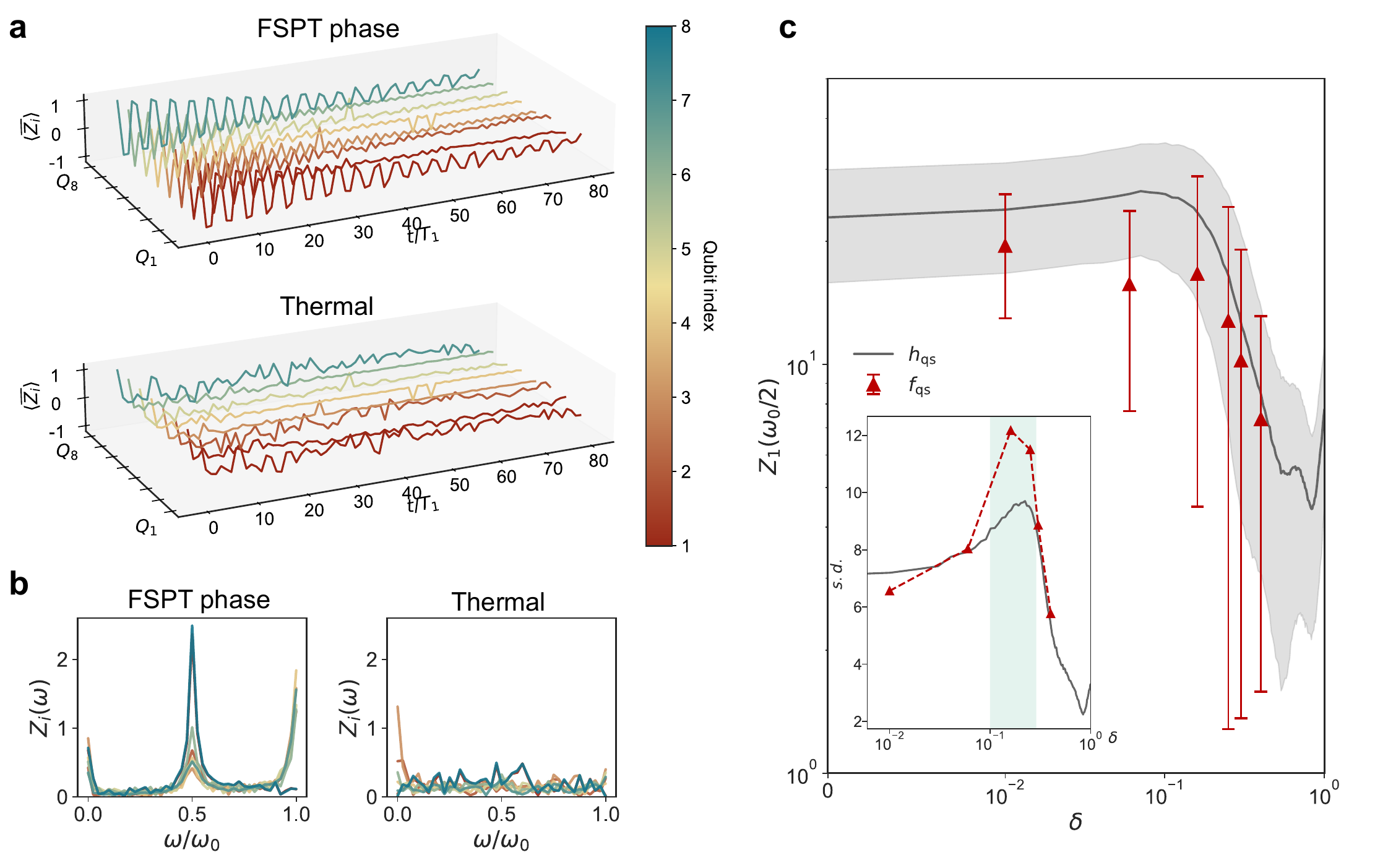}
	\caption{\small{\textbf{Experimental results of using the predictive surrogate $h_{\mathsf{qs}}$ to characterize the phase diagram of a 1D time-periodic Hamiltonian.} \textbf{a.} Time evolution of disorder-averaged local
    magnetizations $\{\braket{\bar{Z}_i}\}_{i=1}^8$ predicted by  $h_{\mathsf{qs}}$ in the FSPT phase ($\delta=0.01$) and the thermal phase ($\delta=0.8$). A stable subharmonic oscillation of local magnetizations that persists over 40 driving periods could only be observed at the bound qubits (i.e., $Q_1$ and $Q_8$) in the FSPT phase, while magnetizations at bulk qubits (i.e., from $Q_2$ to $Q_7$) in the FSPT phase and all magnetizations in the thermal phase decay quickly to zero. \textbf{b.} Fourier transform of local magnetization $\braket{Z_i(\bx)}$ predicted by  $h_{\mathsf{qs}}$ in the FSPT and thermal phases. The label of the $x$-axis,  $\omega/\omega_0$, refers to the normalized frequency with $\omega_0=\pi/\mathrm{T}_1$ being the driven frequency. The boundary qubits (spins) lock to the subharmonic frequency, which is in sharp contrast to the bulk spins. \textbf{c.} The subharmonic peak height, i.e., $Z_1(\omega_0/2)$ with respect to the varied $\delta$ from $0$ to $1$. The results are obtained from the employed quantum processor (labeled by `$f_{\mathsf{qs}}$') and the predictive surrogate (labeled by `$h_{\mathsf{qs}}$'). The data points shown for each $\delta$ are averaged over $5000$ random instances of coupling parameters $\{J_i\}$ for the predictive surrogate and over $50$ random instances for the adopted quantum processor. The shaded regions and error bars represent the standard error over these random instances, respectively. Inset: the standard deviation of the central peak height as a function of $\delta$. The green shaded region refers to the identified region where the phase transition occurs. 
    }}  
	\label{fig:fig4}
\end{figure*}

We next turn to the second task, which is employing the predictive surrogate $h_{\mathsf{qs}}$ to identify FSPT phases and thermalized phases in a class of $8$-qubit time-periodic Hamiltonians. Mathematically, this class of Hamiltonians is defined with 40 driving periods that increase monotonically, i.e., $\{\mathsf{H}_{\mathsf{tp}}(t)\}$  with $t\in [k\mathrm{T}_1]$, $k\in[79]$, and $2\mathrm{T_1}$ being the driving period. This amounts to evaluating the mean value $\braket{Z_i(\bx)}$ with diverse $\bx$.  By convention, the values of $\bx$ are generated by extensive sampling of $\{J_i\}$ and $\delta$ within the ranges $[0, 1]$ and $[0, 2]$, respectively.

To reduce the measurement overhead, we implement $8\times 79$ predictive surrogates $\{h_{\mathsf{qs}}\}$, each corresponding to a circuit configuration $U(\bx)$ with a given driving time $k\mathrm{T}_1$ and a qubit index $i$ to be measured. For each $h_{\mathsf{qs}}$, the size of  $\mathcal{T}_{\mathsf{mb}}$ is set to $n=250$  and the number of shots to obtain $\yi$ is $T=40000$. Refer to SM~G for the satisfactory performance of the optimized predictive surrogates in the measure of $\mathsf{R}(h_{\mathsf{qs}})$ in Eq.~(\ref{eq:prediction_error}).  

Once all predictive surrogates are optimized, we use them to identify the phase transition point entirely on the classical side. In particular, for each surrogate $h_{\mathsf{qs}}$, we consider two values of $\delta$, which are $0.01$ (FSPT phase) and $0.8$ (thermal phase). For each $\delta$,  the remaining entries in $\bx$, i.e., disordered coupling parameters $\{J_i\}$, are randomly sampled from a predefined distribution. This procedure is repeated $20$ times, and $8\times 79$ local magnetizations predicted by each surrogate are averaged over these trials, denoted by $\{\braket{\overline{Z_i(\bx)}}\}$. As shown in FIG.~\ref{fig:fig4}\text{a}, the optimized surrogates successfully capture the signatures of distinct quantum phases. For a small perturbation $\delta=0.01$, the values of $\{\braket{\overline{Z_i(\bx)}}\}$ exhibit stable subharmonic oscillations at boundary spins (i.e., $Q_1$ and $Q_8$) over $40$ driving periods, while decaying quickly to zero in the bulk of the chain (i.e., the qubit indices are from $Q_2$ to $Q_7$). In contrast, for a large perturbation $\delta=0.8$, the magnetizations at both the edge and bulk spins decay quickly to zero within a few periods, signaling the thermalized phase.

We take a step further by quantifying the differences between these two phases through their spectral signatures. To this end, we apply the Fourier transform to the average local magnetization predicted by $8\times 79$ predictive surrogates. As visualized in FIG.~\ref{fig:fig4}\text{b}, for FSPT phase with $\delta = 0.01$, the Fourier spectra of $\{\braket{\overline{Z_i(\bx)}}\}$ reveal that the boundary spins lock to the subharmonic driven frequency $\omega/\omega_0 = 0.5$ with $\omega_0=\pi/\mathrm{T}_1$, whereas the bulk spins do not, showing a distinct difference between edge and bulk behavior. In contrast, in the thermal phase with $\delta = 0.8$, neither the edge spins nor the bulk spins exhibit a pronounced subharmonic frequency peak.

We last exploit the ability of the optimized $\{h_{\mathsf{qs}}\}$ to identify the critical value $\delta^*$, at which the phase transition occurs. According to the setting in Ref.~\cite{zhang_digital_2022}, the critical point amounts to the largest variance of the subharmonic spectral peak height. This peak height is defined as the amplitude in the Fourier spectrum of the local magnetization $\braket{Z_1(\bx)}$ measured at $\omega = \omega_0/2$ for the boundary spin. Figure.~\ref{fig:fig4}\text{c} presents predicted subharmonic peak height as a function of $\delta$, obtained from $h_{\mathsf{qs}}$ over $100\times 5000$ random parameters $\{(\delta, J_i)\}$. The inset of FIG.~\ref{fig:fig4}\text{c} shows a variance peak at $\delta=0.202$. Using the $90\%$ peak-height width of the surrogate’s predictions as a criterion, this indicates the phase transition occurs near this value. Experimental validation using quantum processors on a smaller dataset of $6\times 50$ random parameters $\{(\delta, J_i)\}$  confirms the probable transition region. All of these results are consistent with the expected signature of FSPT and thermal phases as experimentally demonstrated in Ref.~\cite{zhang_digital_2022}, confirming the efficacy of $h_{\mathsf{qs}}$ to characterize non-equilibrium phase transitions.

\section{DISCUSSION and Outlook}
In this study, we propose two predictive surrogates designed to emulate the mean-value behavior of large-scale quantum processors. Theoretically, we prove the computational efficiency of these two predictive surrogates in practical scenarios. Moreover, we experimentally demonstrate the capabilities of these predictive surrogates in advancing digital quantum simulation. Experimental results in pre-training VQEs and the identification of the FSPT phase up to $20$ superconducting qubits indicate that our approaches not only substantially reduce measurement overhead compared to conventional quantum algorithms, but may also surpass their performance. 

Our study opens a new research direction by harnessing advanced AI techniques to accelerate the pursuit of practical quantum utility. Several important research avenues remain to be explored. A promising avenue is to investigate how predictive surrogates can advance additional variational quantum algorithms~\cite{cerezo2021variational,du2025quantum} and classical shadow estimation tasks~\cite{elben2022randomized}. Next, it is important to generalize the achieved results to the circuits with intrinsic symmetries. This motivation is twofold: many quantum algorithms leverage the underlying symmetries of the problem Hamiltonian to reduce gate counts and enhance performance~\cite{romero2018strategies,wang2022symmetric,weaving2025contextual}; and it is intriguing to investigate whether incorporating symmetry into quantum circuits can further improve the computational efficiency of predictive surrogates. Another important research direction is to extend the emulation of mean-value behavior to scenarios involving the learning of nonlinear functions, e.g., entanglement entropy, stabilizer entropy, and purity. Any progress in this direction would benefit a range of fields, such as quantum many-body physics~\cite{elben2018renyi,lukin2019probing,vermersch2024many} and quantum system certification~\cite{zhu2022cross,joshi2023exploring,shaw2024benchmarking,wu2025state}. In addition, it is important to investigate the development of predictive surrogates that can emulate output distributions rather than just mean values, as many quantum algorithms with provable advantages rely on sampling from the evolved state instead of estimating expectation values~\cite{shor1999polynomial,grover1996fast, hangleiter2023computational}. Advancements in this direction would further broaden the potential applications of predictive surrogates. Parallel to research on provably efficient predictive surrogates, an important complementary direction is to systematically investigate whether deep learning-based surrogates~\cite{zhu2022flexible,melko2024language,wu2024variational} can further enhance the efficiency in emulating quantum processors' behavior at scale.

\newpage
 
\clearpage

\onecolumngrid

\appendix 
\tableofcontents
\renewcommand{\appendixname}{SM}
 \renewcommand\thefigure{\thesection.\arabic{figure}}   
 
\bigskip

\section{Preliminary}
In this section, we first present the essential mathematical framework for the quantum circuits discussed in this work, as outlined in SM~\ref{append:subsec:ptm_qc}. This includes the Pauli transfer matrix representation and the trigonometric expansion of quantum circuits in both noiseless and noisy scenarios. Moreover, to contextualize our work within the existing literature, we conduct a comprehensive review of
relevant studies in SM~\ref{append:subsec:review}. For convenience, we summarize the frequently used notations throughout this study in the following table.
\begin{longtable}{  p{.16\textwidth}   p{.8\textwidth} }
\caption{The summary of notations frequently used in this work.}\\
\endfirsthead

\hline
 \textbf{Notation} & \textbf{Definition} \\ 
        \hline
        $N$ & Number of qubits \\
        $n$ & Number of training examples \\
        $T$ & Number of snapshots for constructing classical shadow \\
        $\bx_l$ & $l$-th component of the parameter vector $\bx$
        \\
        $d$ & Dimension of parameter vector $\bx$ i.e., $\bx\in \mathbb{R}^d$ \\
        $[d]$ & The set $\{1,2,\cdots, d\}$ \\
        $U(\bx)$, $\tilde{\mathcal{U}}({\bx})$ & Parametrized quantum circuit and its noisy counterpart \\
        $\rho(\bx)$, $\tilderho(\bx)$ & Density matrix of the noiseless and noisy quantum states   \\
        $\tilderho_T(\bx)$ & classical shadow representation of $\tilderho(\bx)$ with $T$ measurement shots \\
        $\mathsf{H}$ & An $N$-qubit Hamiltonian \\
        $f(\tilderho(\bx),O)$ & The expectation value $\Tr[O \tilderho(\bx)]$ for a given observable $O$  \\
        $h_{\mathsf{cs}}, h_{\mathsf{qs}}$ & The predictive surrogates for classical shadow and quantum simulations  \\
        $\mathcal{M}, \mathcal{N}, \mathcal{V}$ & Quantum channel \\
        $\bUnitary(\bx), \mathsf{M}, \mathsf{N}, \mathsf{V}$ & Pauli transfer matrix representation of quantum operators $U(\bx),\mathcal{M}, \mathcal{N}, \mathcal{V}$ \\
        $|\rho \rrangle, |O \rrangle$ & Pauli transfer matrix representation of quantum state $\rho$ and observable $O$ \\
        $I, X, Y, Z$ & Identity, Pauli-X, Pauli-Y, and Pauli-Z matrix \\
        $\PX$, $\PY$, $\PZ$ & Error probability of Pauli noisy channel $\mathcal{N}_P$ \\
        $\Phi_{\bomega}(\bx)$ & Trigonometric monomial basis related to the frequency vector $\bomega \in \{0,1,-1\}^d$ \\
        $\mathfrak{C}(\Lambda)$ & Truncated frequency set $\{\bomega|\bomega\in \{\pm 1, 0\}^d,\|\bomega\|_0 \le \Lambda\}$ with $\Lambda$ being the threshold value \\
        $\mathfrak{S}(\Lambda)$ & Truncated frequency set $\{\bomega|\bomega\in \{\pm 1, 0\}^d,\|\bomega_{\#-1}\|_0 \le \Lambda\}$ with $\|\bomega_{\#-1}\|_0$ being the number of $\bomega_j=-1$. \\
\hline
\end{longtable}

\subsection{Pauli transfer matrix and the trigonometric expansion of quantum circuits}\label{append:subsec:ptm_qc}

\noindent \textbf{Pauli Transfer Matrix.} Here we review how to use Pauli-Liouville representation to formulate the quantum state and the observable. Denote $P_l$ as  the $l$-th normalized Pauli operator with $P_l\in \frac{1}{\sqrt{2^N}}\{\mathbb{I}, X, Y, Z\}^N$ and $\llangle P_l|P_k\rrangle=\Tr(P_lP_k)=\delta_{lk}$. An $N$-qubit state $\rho$ can be represented as a $4^N$-dimensional vector under these normalized Pauli bases, i.e., 
\begin{equation}
    |\rho \rrangle = \left[\Tr(\rho P_1), \cdots, \Tr(\rho P_{4^N})\right]~\mbox{with}~P_l\in \frac{1}{\sqrt{2^N}}\{\mathbb{I}, X, Y, Z\}^{\otimes N}.
\end{equation}
Similarly, the normalized Pauli operator $O$ under the normalized Pauli basis yields $|O \rrangle = [\Tr(O P_1), ...,  \Tr(O P_{4^N})]^{\top}$.

The unitary operator can also be expanded by Pauli bases. Given the circuit $U(\bx)$, its Pauli Transfer Matrix (PTM)~\cite{hantzko2024pauli}, denoted by $\bUnitary(\bx)$, yields
\begin{equation}\label{eqn:append:PTM}
	[\bUnitary(\bx)]_{lk} = \llangle P_l|\bUnitary(\bx)| P_k \rrangle = \Tr(P_lU(\bx)P_k U(\bx)^{\dagger}). 
\end{equation}
For instance, the PTM representation of $\RZ(\bx_l)$ gates takes the form as
\begin{equation}
    \mathsf{R_Z} = \left(\begin{matrix}
		1 & 0 & 0 & 0\\
		0 & \cos(\bx_l) & -\sin(\bx_l) & 0\\
		0 & \sin(\bx_l) & \cos(\bx_l) & 0\\
		0 & 0 & 0 & 1
	\end{matrix}\right) = 
\mathsf{D}_0 + \cos(\bx_l)\mathsf{D}_1 + \sin(\bx_l)\mathsf{D}_{-1},
\end{equation}
where
\begin{equation}\label{append:eq:ptm_basis_noiseless}
	\mathsf{D}_0 = \left(\begin{matrix}
		1 & 0 & 0 & 0\\
		0 & 0 & 0 & 0\\
		0 & 0 & 0 & 0\\
		0 & 0 & 0 & 1
	\end{matrix}\right), \quad 	
        \mathsf{D}_1 = \left(\begin{matrix}
		0 & 0 & 0 & 0\\
		0 & 1 & 0 & 0\\
		0 & 0 & 1 & 0\\
		0 & 0 & 0 & 0
	\end{matrix}\right), 
	\quad 	\mathsf{D}_{-1} = \left(\begin{matrix}
		0 & 0 & 0 & 0\\
		0 & 0 & -1 & 0\\
		0 & 1 & 0 & 0\\
		0 & 0 & 0 & 0
	\end{matrix}\right).
\end{equation}

\medskip

\noindent \textbf{Trigonometric expansion of noiseless quantum circuits.}   Consider an $N$-qubit quantum circuit discussed in the main text, which is
\begin{equation}\label{eqn:append:circuit-ideal}
	U(\bx) = \prod_{l=1}^{d}(\RZ(\bx_l)V_l).
\end{equation}
When applied to an arbitrary $N$-qubit input state $\rho_0$, the generated state can be reformulated by the trigonometric expansion~\cite{fontana2023classical}.  The mathematical expression of the generated state yields
\begin{equation}\label{append:eq:ptm_state_noiseless}
	\rho(\bx) = U(\bx) \rho_0 U(\bx) ^{\dagger} =  \sum_{\bomega}\Phi_{\bomega}(\bx) \llangle \rho_0| \bUnitary_{\bomega}^{\dagger}.
\end{equation}
Here the notation $\Phi_{\bomega}(\bx)$ with $\bomega\in \{0, 1, -1\}^d$ refers to the trigonometric monomial basis with the values 
\begin{equation}
	\Phi_{\bomega}(\bx) = \prod_{l=1}^d \begin{cases}
		 1 ~ & \textnormal{if}~ \bomega_l = 0 \\
		 \cos(\bx_l) & \textnormal{if}~\bomega_l = 1 \\
		 \sin(\bx_l) &  \textnormal{if}~ \bomega_l = -1
	\end{cases}.
\end{equation}
Moreover, $\bUnitary_{\bomega}^{\dagger}$ is the purely-Clifford circuit and each RZ-gate at position $l$ is replaced by one of the operators $\mathsf{D}_0$, $\mathsf{D}_1$, $\mathsf{D}_{-1}$ in Eq.~\eqref{append:eq:ptm_basis_noiseless} depending on the value of $\bomega_l$. 

When the generated state  $\rho(\bx)$ interacts with an observable $O$, the expectation value can be equivalently expressed using a trigonometric expansion, i.e.,
\begin{equation}\label{eqn:append:expectation-idea}
	f(\rho(\bx), O) \equiv \Tr(\rho(\bx)O)= \sum_{\bomega}\Phi_{\bomega}(\bx) \llangle 0| \bUnitary_{\bomega}^{\dagger} | O \rrangle = \sum_{\bomega}\Phi_{\bomega}(\bx) \bm{\alpha}_{\bomega},
\end{equation} 
where  $\bm{\alpha}_{\bomega}=\llangle 0| \bUnitary_{\bomega}^{\dagger} | O \rrangle$.

\medskip

\noindent \textbf{Trigonometric expansion in the noisy scenario.}  We now generalize the above PTM representations to the noisy scenario, a ubiquitous feature of quantum computers without fault tolerance. Supported by the results of twirling operation~\cite{bennett1996mixed}, we model the noisy environment via a Pauli noise channel~\cite{chen2024tight}. For a single qubit, the Pauli channel yields
\begin{equation}\label{append:eq:pauli_channel}
	\mathcal{N}_P(\PX , \PY, \PZ )[\rho] = (1-\PX -\PY-\PZ )\rho + \PX X\rho X + \PY Y\rho Y+\PZ Z\rho Z. 
\end{equation}
Following the same approach, we denote a general Pauli channel applied to multiple qubits as $\mathcal{M}$. With these notations, the explored circuit $U(\bx)$ under Pauli noise channel yields the formalism in Eq.~(\ref{eqn:noisy-unitary-channel}), which is 
\begin{equation}\label{eqn:append:noisy-u(x)}
	\widetilde{\mathcal{U}}({\bx}) = \bigcirc_{l=1}^d \tilde{\mathcal{R}}_{\mathcal{Z}}(\bx_l) \circ \tilde{\mathcal{V}}_{l}.
\end{equation}
Here $\tilde{\mathcal{R}}_{\mathcal{Z}}(\bx_l)= \mathcal{N}_P\circ \mathcal{R_Z}(\bx_l)$ and $\tilde{\mathcal{V}}_{l} = \mathcal{M}_l \circ \mathcal{V}_l$ with $\mathcal{M}_l$ being a multi-qubit Pauli channel applied to the $l$-th Clifford operation $\mathcal{V}_l$.

We next introduce PTM representation of the noisy circuit $\widetilde{\mathcal{U}}({\bx})$. Recall the definition of PTM in Eq.~(\ref{eqn:append:PTM}). A quantum channel $\mathcal{E}$ can also be reformulated as a PTM with the size $4^N\times 4^N$, i.e.,
\begin{equation}
	[\mathsf{E}]_{l,k} = \llangle P_l|\mathsf{E}|P_k\rrangle = \Tr({P}_l \mathcal{E}({P}_k)),
\end{equation}
where $P_l$ denotes the normalized Pauli operator. As such, the PTM of the single-qubit Pauli channel in Eq.~(\ref{append:eq:pauli_channel}) is 
\begin{equation}
    \mathsf{N}=\mathrm{diag}(1,q_{\mathrm{X}} ,q_{\mathrm{Y}},q_{\mathrm{Z}}),
\end{equation}
where $q_{\mathrm{X}}=1-2(\PZ +\PY)$, $q_{\mathrm{Y}}=1-2(\PZ +\PX )$, $q_{\mathrm{Z}}=1-2(\PX +\PY)$ refer to the eigenvalues of the Pauli channel related to the error probabilities $\PX $, $\PY$, $\PZ $, respectively.  Moreover, the PTM related to $\mathcal{N}_P$ applied to the $\RZ(\bx_l)$ gate is given by 
\begin{equation}
	\mathsf{N}\cdot \mathsf{R_Z} = \left(\begin{matrix}
		1 & 0 & 0 & 0\\
		0 & q_{\mathrm{X}} & 0 & 0\\
		0 & 0 & q_{\mathrm{Y}} & 0\\
		0 & 0 & 0 & q_{\mathrm{Z}}
	\end{matrix}\right)\cdot \left(\begin{matrix}
		1 & 0 & 0 & 0\\
		0 & \cos(\bx_l) & -\sin(\bx_l) & 0\\
		0 & \sin(\bx_l) & \cos(\bx_l) & 0\\
		0 & 0 & 0 & 1
	\end{matrix}\right)=\left(\begin{matrix}
		1 & 0 & 0 & 0\\
		0 & q_{\mathrm{X}}\cos(\bx_l) & -q_{\mathrm{X}}\sin(\bx_l) & 0\\
		0 & q_{\mathrm{Y}}\sin(\bx_l) & q_{\mathrm{Y}}\cos(\bx_l) & 0\\
		0 & 0 & 0 & q_{\mathrm{Z}}
	\end{matrix}\right). 
\end{equation} 
After simplification, we have
\begin{equation}\label{append:eq:noisy_rz_ptm}
	\mathsf{N}\cdot \mathsf{R_Z} = \widetilde{\mathsf{D}}_0 + \cos(\bx_l)\widetilde{\mathsf{D}}_1 + \sin(\bx_l)\widetilde{\mathsf{D}}_{-1},
\end{equation}
where
\begin{equation}\label{append:eq:ptm_basis_noise}
	\widetilde{\mathsf{D}}_0 = \left(\begin{matrix}
		1 & 0 & 0 & 0\\
		0 & 0 & 0 & 0\\
		0 & 0 & 0 & 0\\
		0 & 0 & 0 & q_{\mathrm{Z}}
	\end{matrix}\right), \quad 	\widetilde{\mathsf{D}}_1 = \left(\begin{matrix}
		0 & 0 & 0 & 0\\
		0 & q_{\mathrm{X}} & 0 & 0\\
		0 & 0 & q_{\mathrm{Y}} & 0\\
		0 & 0 & 0 & 0
	\end{matrix}\right), 
	\quad 	\widetilde{\mathsf{D}}_{-1} 
    = \left(\begin{matrix}
		0 & 0 & 0 & 0\\
		0 & 0 & -q_{\mathrm{X}} & 0\\
		0 &  q_{\mathrm{Y}} & 0 & 0\\
		0 & 0 & 0 & 0
	\end{matrix}\right). 
\end{equation}
This formalism reduces to the noiseless counterpart in Eq.~\eqref{append:eq:ptm_basis_noiseless} for $\PX =\PY=\PZ =0$ or $q_{\mathrm{X}}=q_{\mathrm{Y}}=q_{\mathrm{Z}}=1$.

 Following the same routine, denote PTM representation of the $l$-th Clifford operation $\mathcal{V}_l$ and the general Pauli channel $\mathcal{M}_l$ acting on this operation as $\mathsf{V}_l$ and $\mathsf{M}_l$, respectively. Then the PTM representation of $\widetilde{\mathcal{U}}({\bx})$ is 
 \begin{equation}\label{eqn:append:PTM_U}
 	 \prod_{l=1}^d  \left(\mathsf{N}\cdot \mathsf{R_Z}(\bx_l) \cdot  \mathsf{M}_l \cdot  \mathsf{V}_l \right).
 \end{equation}

Denote $\tilderho(\bx)$ as the quantum state $\rho_0$ evolved under the noisy circuit $\tilde{\mathcal{U}}({\bx})$. The decomposition in Eqs.~\eqref{append:eq:noisy_rz_ptm} and \eqref{eqn:append:PTM_U} allows us to express the mean value $\Tr(\tilderho(\bx)O)$, as the noisy version of $f(\rho(\bx), O)$ in Eq.~(\ref{eqn:append:expectation-idea}), under the trigonometric expansion. Following the notation in Eq.~\eqref{append:eq:ptm_state_noiseless}, we have
\begin{equation}
    \label{eq:fourier_expansion}
	f(\tilderho(\bx), O) \equiv \Tr(\tilderho(\bx)O)= \sum_{\bomega}\Phi_{\bomega}(\bx) \tilde{\bm{\alpha}}_{\bomega} := \sum_{\bomega}\Phi_{\bomega}(\bx)  \llangle O|\tilde{\bUnitary}_{\bomega}|\rho_0 \rrangle,
\end{equation} 
where $\tilde{\bm{\alpha}}_{\bomega}:=\llangle O|\tilde{\bUnitary}_{\bomega}|\rho_0 \rrangle$ and $\tilde{\bUnitary}_{\bomega}:=\prod_{l}^d \widetilde{\mathsf{D}}_{\bomega_l} \cdot \mathsf{M}_l \cdot  \mathsf{V}_l$  refer to the noisy counterpart of the purely-Clifford circuit $\bUnitary_{\bomega}$.

\smallskip 
The introduced Pauli channel results in the contraction phenomenon of $\tilde{\bm{\alpha}}_{\bomega}$. More specifically, the absolution value of $\{\tilde{\bm{\alpha}}_{\bomega}\}$ decays exponentially with the Hamming weight of the frequency vector $\bomega$. The following lemma establishes the relation between $\tilde{\bm{\alpha}}_{\bomega}$ and $ \bm{\alpha}_{\bomega}$, which will be employed in the proof of Theorem~\ref{thm:learning-non-trunc} and Theorem~\ref{thm:ridge_predict_error}.  

\begin{lemma}[Adapted from Theorem 4, Ref.~\cite{fontana2023classical}]\label{append:lem:noisy_coefficient}
    Following the notation in Eqs.~(\ref{eqn:noisy-unitary-channel}), (\ref{eqn:append:expectation-idea}), and (\ref{eq:fourier_expansion}), define $p=\min\{\PX ,\PY\}$ and suppose $\max\{p,\PZ \}>0$. Then, for any observable $O$ and initial state $\rho_0$, the coefficients of trigonometric expansion over $f(\rho(\bx), O)$ and $f(\tilderho(\bx), O)$ yields
    \begin{equation}
        \left|\widetilde{\bm{\alpha}}_{\bomega} \right| \le \Big(1-2(p+\PZ )\Big)^{\|\bomega\|_0}\left|\bm{\alpha}_{\bomega}\right|,\quad \forall \bomega\in \{-1,0,1\}^d.
    \end{equation}
\end{lemma}
\begin{proof}
    [Proof of Lemma~\ref{append:lem:noisy_coefficient}] We first consider the case of $O$ being a single Pauli string with $P\in\{\mathbb{I},X,Y,Z\}^{\otimes N}$.  Recall that the PTM representation of the Pauli channel ${\mathsf{N}}$ or ${\mathsf{M}}_l$ is a diagonal matrix. Accordingly, any Pauli channel will map a  Pauli operator to itself, up to a proportionality factor that is determined by the eigenvalues of ${\mathsf{M}}_l$ and is less than $1$, namely $\widetilde{\bUnitary}_{\bomega}|P\rrangle=\prod_{l}^d \widetilde{\mathsf{D}}_{\bomega_i} \cdot {\mathsf{M}}_l \cdot\mathsf{V}_l  |P\rrangle\propto \prod_{l} \widetilde{\mathsf{D}}_{\bomega_l} \cdot {\mathsf{V}}_l |P\rrangle$ and hence $\llangle P|\widetilde{\bUnitary}_{\bomega}|\rho_0 \rrangle = r\llangle P|\prod_{l} \widetilde{\mathsf{D}}_{\bomega_l} \cdot \mathsf{V}_l|\rho_0 \rrangle$ with $r<1$.  
    Moreover, following the definition of $\widetilde{\mathsf{D}}_{\bomega_l}$, we have $\llangle P|\prod_{l} \widetilde{\mathsf{D}}_{\bomega_l} \cdot {\mathsf{V}}_l|\rho_0 \rrangle = c_{\bomega} \llangle P|\prod_{l} {\mathsf{D}}_{\bomega_l} \cdot {\mathsf{V}}_l|\rho_0 \rrangle$, where $0<c_{\bomega}\le q^{\|\bomega\|_0}$ with $q:=\max\{q_{\mathrm{X}},q_{\mathrm{Y}}\}$. Combining the above analysis and the expression of $q=1-2(p+\PZ )$ regarding the error probability $p,\PZ$, we have
    \begin{equation}\label{append:eq:lem1_proof_1}
        \widetilde{\bm{\alpha}}_{\bomega} =r \llangle P|\prod_{l}^d \widetilde{\mathsf{D}}_{\bomega_l}\cdot {\mathsf{V}}_l|\rho_0 \rrangle = rc_{\bomega} \llangle P|\prod_{l}^d {\mathsf{D}}_{\bomega_l}\cdot{\mathsf{V}}_l|\rho_0 \rrangle \le q^{\|\bomega\|_0} \left|{\bm{\alpha}}_{\bomega}\right| = (1-2(p+\PZ ))^{\|\bomega\|_0}\left|\bm{\alpha}_{\omega}\right|.
    \end{equation}
    This equality could be directly generalized to the case of a general observable $O=\sum_{j} a_j P_j$ by employing the linearity of $\bm{\alpha}_{\bomega}$ with respect to the Pauli observable $P_j$. In particular, we have
    \begin{equation}
        \widetilde{\bm{\alpha}}_{\bomega} = \llangle O|\tilde{\bUnitary}_{\bomega}|\rho_0 \rrangle = \sum_{j} a_j \llangle P_j|\tilde{\bUnitary}_{\bomega}|\rho_0 \rrangle = rc_{\bomega} \sum_{j} a_j \llangle P_j|{\bUnitary}_{\bomega}|\rho_0 \rrangle = rc_{\bomega} \bm{\alpha}_{\bomega} \le (1-2(p+\PZ ))^{\|\bomega\|_0} \left|{\bm{\alpha}}_{\bomega}\right|, 
    \end{equation}
    where the inequality follows Eq.~\eqref{append:eq:lem1_proof_1}.  
\end{proof}

\medskip
\noindent\textbf{Equivalence between quantum circuits using various rotation gates up to additional Clifford gates.}
We note that the quantum circuits $U(\bx)=\prod_{l=1}^d \RZ(\bx_l)V_l$ consisting of the $\RZ$ gates and Clifford gates considered in this study are equivalent to those consisting of other Pauli rotation gates up to Clifford gates. Formally, consider the quantum circuit $W(\bx)=\prod_{l=1}^d \mathsf{R}_{P_l}(\bx_l)W_l$ with $\mathsf{R}_{P_l}(\bx_l)=e^{-iP_l \bx_l/2}$ being the ratation gates around any Pauli basis $P_l\in \{I,X,Y,Z\}^{\otimes N}$ and $W_l$ being the Clifford gates, there always exist Clifford gates $W_1', \cdots, W_d'$ such that for $l\in[d]$, we have
\begin{equation}
    W_l' \mathsf{R}_{P_l}(\bx_l) (W_l')^{\dagger}=W_l' e^{-iP_l \bx_l/2} (W_l')^{\dagger}=e^{-i W_l'P_l(W_l')^{\dagger} \bx_l/2} = e^{-i  Z  \bx_l/2} = \RZ(\bx_l).
\end{equation}
As a result, we could rewrite the quantum circuit $W(\bx)$ as
\[W(\bx)=\prod_{l=1}^d \mathsf{R}_{P_l}(\bx_l)W_l=\prod_{l=1}^d (W_{l}')^{\dagger} \RZ(\bx_l)W_l' W_l = \prod_{l=0}^d \RZ(\bx_l) V_l',\]
where the last equality follows by setting $\bx_0=0$ such that $\RZ(\bx_0)=I$ and denoting $V_0'=(W_{1}')^{\dagger}$, $V_d'=W_d' W_d$, and $V_l'=W_l' W_l(W_{l+1}')^{\dagger}$ for $l\in [d-1]$. Here, $V_l'$ are still Clifford gates according to the properties of the Clifford group, namely, for any Clifford operators $A,B$, the product $AB$ is also a Clifford operator. In this regard, the quantum circuits $W(\bx)$ could be regarded as equivalent to $U(\bx)$ with various Clifford gates. This relationship of equivalence enables the direct generalization of the achieved results for $U(\bx)$ in this study to any other quantum circuit $W(\bx)$, as we do not impose any assumption on the Clifford gates $V_l$.

 \subsection{Literature review}\label{append:subsec:review}

The most relevant works to this study can be broadly classified into four categories: classical simulators of quantum processors, learning-based surrogates for quantum processors, and methodologies for enhancing variational quantum algorithms (VQAs) at scale, as well as enhancing the characterization of non-equilibrium phases. In the following, we review each of these research directions separately.

\medskip
\noindent \textbf{Classical simulator of quantum processors}. 
Classical simulators of quantum processors involve designing classical simulation algorithms to predict the outcomes of quantum circuits. While simulating general quantum circuits using classical algorithms has been proven to be computationally challenging, specialized classical simulation methods have been proposed to efficiently simulate specific types of quantum systems with particular properties. For example, tensor network-based methods are developed to simulate low-entanglement quantum systems efficiently \cite{shi2006classical,markov2008simulating,singh2010tensor,biamonte2017tensor,hauschild2018efficient,pang2020efficient,causer2023optimal,patra2024efficient}. Various techniques have also been introduced to simulate near-stabilizer circuits with low magic \cite{gottesman1998heisenberg,aaronson2004improved,nest2008classical,bravyi2016improved,lerch2024efficient,beguvsic2025simulating}, including the stabilizer rank, quasiprobability, Pauli-path simulation, and Clifford perturbation methods. More general group-theoretic approaches are proposed for simulating symmetrized systems \cite{somma2005quantum,somma2006efficient,galitski2011quantum,goh2023lie,anschuetz2023efficient}. Additionally, it has been shown that the presence of noise in quantum systems can relax the constraints on classically simulable quantum circuits that are otherwise intractable in the noiseless case \cite{gao2018efficient,takahashi2021classically,shao2024simulating,gonzalez2025pauli,rajakumar2025polynomial,fontana2023classical}. As a result, quantum circuits that are computationally infeasible for classical simulation in the noiseless case may be efficiently simulated under noisy conditions. 

Despite the importance of developing efficient classical simulators for quantum processors, a fundamental limitation is their ability to faithfully emulate the behavior of real quantum devices. Classical simulation methods typically require a fully specified noise model, which is often computationally expensive to obtain. In contrast, predictive surrogates can automatically and effectively learn the noisy behavior of quantum processors directly from collected training data, offering a distinct advantage over conventional classical simulators.

\medskip
\noindent \textbf{Learning-based surrogates for quantum processors}. Unlike classical simulators, which aim to predict the output of individual quantum circuits, learning-based surrogates are designed to emulate the output of families of quantum circuits with varying parameters. The existing literature on the construction of learning-based surrogates can be categorized into two main approaches: deep learning-based surrogates and traditional machine learning-based surrogates. In what follows, we provide a separate review of each subcategory.

\smallskip

\noindent \underline{\textit{Deep learning-based surrogates.}} The deep learning-based surrogates are constructed by optimizing a deep neural network to fit the experimental data collected from quantum processors. Due to the flexibility of neural architectures and their expressive power, a variety of deep learning-based surrogates have been proposed for different tasks to predict specific properties of quantum processors based on various data types \cite{wu2023quantum,qin2024experimental,qian2024multimodal,tang2024towards,wang2022predicting,zhao2025rethink}. For instance, Ref.~\cite{wu2023quantum} constructs a surrogate using a convolutional neural network to predict the similarity between quantum states. In a similar task, Ref.~\cite{qian2024multimodal} utilizes a multimodal neural network to construct the surrogate, enabling it to fully leverage classical information about quantum processors, such as the gate layout of quantum circuits. Ref.~\cite{wang2022predicting} develops a generative surrogate capable of generating measurement outputs given a classical description of a quantum system. 

Although these approaches have demonstrated promising empirical results, they are mostly heuristic, and there is a lack of theoretical guarantees regarding their efficiency. Our work addresses this knowledge gap by providing rigorous evidence for the feasibility and effectiveness of learning-based surrogates for noisy quantum processors.

\smallskip

\noindent \underline{\textit{Machine learning-based surrogates.}} 
The machine learning-based surrogates aim to construct an explicit learning model by dequantizing the quantum circuits. Existing literature explores dequantization mainly by expressing quantum circuits in two different kinds of basis expansions: the Fourier expansion \cite{sweke2023potential} and the trigonometric expansion. For the first case, the learning-based surrogate is formalized as a linear model in the Fourier basis related to the Fourier expansion of quantum circuits, and the training data obtained from quantum processors is used to estimate the corresponding Fourier coefficients~\cite{sweke2023potential, fontana2022efficient, schreiber2023classical, nemkov2023fourier, landman2022classically, gan2024concept}. A  limitation of this kind of surrogate is that all Fourier coefficients need to be estimated without considering the truncation of the less important Fourier basis contributing to the expectation value. This leads to an exponentially increasing sample complexity with the number of Pauli rotation gates, denoted as $d$. As a result, this method is efficient only for quantum circuits with $d=\log(N)$, where $N$ is the number of qubits.

The second approach involves expressing the quantum circuits in a trigonometric expansion and utilizing the truncated trigonometric feature map to construct the predictive surrogate. This helps mitigate the exponential scaling of sample complexity with respect to the number of rotational gates. Specifically, Ref.~\cite{du2024efficient} demonstrates that for certain quantum circuits with a small gradient norm, truncating the trigonometric features with a large order only minimally impacts the quantum processor’s output. Therefore, they can be neglected when constructing the classical surrogate model. 
It is worthy to note that the results in Ref.~\cite{du2024efficient} are established in noiseless scenarios, leaving several open questions unresolved: the provable efficiency of these surrogates in the presence of noise remains unknown; it is unclear whether the noisy case is inherently more challenging than the ideal case; and the generalization of these methods beyond uniform input sampling has yet to be demonstrated. In this work, we address all of these open issues.

\medskip
\noindent \textbf{Enhancement of VQAs at scale}. Optimizing VQAs becomes progressively challenging as the number of qubits continuously increases \cite{cerezo2022challenges}. This difficulty stems from their non-convex loss landscapes \cite{bittel2021training} and non-community among trainable gates~\cite{schuld2019evaluating}, necessitating numerous shots to navigate the trainability issues like barren plateaus \cite{mcclean2018barren,zhang2020toward} and optimization divergence \cite{du2021learnability,sweke2020stochastic}. In this regard, the conflict between the resource-intensive nature of VQAs and the scarcity of quantum resources in the coming years poses a substantial obstacle to their wider applications. To improve the practicality of VQAs, huge efforts have been devoted to decreasing the quantum resources required in the optimization while preserving the quality of the output solutions \cite{cerezo2023does}. 

Smart initialization techniques have become leading methods \cite{sutskever2013importance,grant2019initialization,mitarai2022quadratic} in this context. The philosophy of this research line is to identify a better starting point in the loss landscape rather than relying on random initialization, where the starting points are closer to the good local minima or global minimum \cite{boyd2004convex}. As such, the high measurement overhead can be mitigated significantly. Existing initialization methods can be classified into two categories: heuristic initializers and informative initializers, with the key distinction being whether they require access to quantum processors.
Heuristic initialization strategies are typically classically implemented and often involve small-angle initialization  \cite{zhang2022gaussian, park2024hamiltonian, wang2024trainability, park2024hardware, shi2024avoiding} and are implemented classically, where the quantum circuit is initialized in a small region, typically around the identity matrix. This is known as ``identity initialization" \cite{grant2019initialization} and has been shown to exhibit gradients that do not vanish exponentially.

In contrast, informative initializers aim to identify high-quality initialization parameters using experimental data collected from quantum processors. These methods include warm-start techniques \cite{puig2025variational, mhiri2025unifying}, parameter transfer techniques \cite{liu2023mitigating, mele2022avoiding}, pre-training VQAs \cite{dborin2022matrix,mitarai2022quadratic,goh2023lie,shaffer2023surrogate,rudolph2023synergistic}, and neural network-based initializers \cite{verdon2019learning, cervera2021meta, luo2024quack}. Among these strategies, pre-training VQAs has emerged as a promising approach for reducing the measurement cost involved in optimizing VQAs. This method involves constructing a classical surrogate model to emulate the output of a quantum processor, assisting in the pre-training process. For instance, Ref.~\cite{goh2023lie} proposes using a Lie-algebra surrogate to pre-train VQAs, while Ref.~\cite{dborin2022matrix} introduces a circuit pre-training method based on matrix product state machine learning techniques.

The proposed predictive surrogates based on the triangular expansion complement the existing initialization techniques in advancing VQAs and, moreover, offer a new avenue for pre-training VQAs.

\medskip
\noindent \textbf{Enhancement of characterization of non-equilibrium phases}. Existing literature on phase characterization using learning-based methods predominantly focuses on equilibrium quantum many-body systems \cite{van2017learning, carrasquilla2017machine, ch2017machine, huang2022provably, kasatkin2024detecting}, where the phase is determined by the ground state of a time-independent Hamiltonian. However, there is a lack of research that applies machine learning methods to characterize the phases of non-equilibrium quantum many-body systems. This task is more challenging than its equilibrium counterpart due to the absence of well-defined steady states and the difficulty of describing the system's properties without equilibrium conditions.

\section{Learnability of the predictive surrogate $h_{\textsf{cs}}$ (Proof of Theorem 1)}\label{append:sec:proof-thm2}

This section presents the proof of Theorem 1, which analyzes the prediction error of the proposed kernel-based machine learning model defined in Eq.~(\ref{eqn:generic-learner}). The results extend our previous theoretical work~\cite{du2024efficient} to the noisy setting, providing provable guarantees for our proposal when applied to quantum circuits subject to Pauli noise. In particular, Theorem 1 indicates how the prediction error depends on the number of training examples $n$, the size of the quantum system $N$, the dimension of classical inputs $d$, and the strength of noise in Pauli Channels. Recall the explicit form of the proposed learning model in Eqs.~(\ref{eqn:generic-learner-state}) and (\ref{eqn:generic-learner}), which is 
\begin{equation}\label{append:eqn:TriGeo-non-trunc-form}\hatsigma(\bx)=\frac{1}{n}\sum_{i=1}^n\hatkappa\left(\bx, \bxi\right)\tilderho_T(\bxi) ~\text{with}~\hatkappa\left(\bx, \bxi\right) = \sum_{\bomega, \|\bomega\|_0 \leq \Lambda} 2^{\|\bomega\|_0}\Phi_{\bomega}(\bx)\Phi_{\bomega}(\bxi) \in \mathbb{R}.
\end{equation} 
Here $\tilderho_T(\bxi)$ represents the classical shadow of $\tilde{\rho}(\bxi)$ with $T$ snapshots, while $\tilde{\rho}(\bxi)$ is the quantum state prepared by the noisy circuit $\tilde{\mathcal{U}}(\bx)$ in E.~\eqref{eqn:append:noisy-u(x)}. The goal of this proof is to analyze the average discrepancy between the prediction $\Tr(\hat{\sigma}(\bx) O)$ and the expectation value  $\Tr(\tilderho(\bx) O)$, where $\bx$ is randomly and uniformly sampled from the domain $[-\pi, \pi]^d$. More precisely, we aim to bound the expected risk defined in Eq.~(\ref{eq:prediction_error}), i.e., 
\begin{equation}
  \mathsf{R} = \mathbb{E}_{\bx\sim [-\pi, \pi]^d} | \Tr(O\hatsigma(\bx)) - \Tr(O\tilderho(\bx)) |^2. 
\end{equation}

\medskip
In the remainder of this section, we first present the formal statement of Theorem 1. We then introduce two necessary lemmas used in its proof, i.e., Lemmas~\ref{lem:truncation-error-geo-kernel} and \ref{lem:estimation-error-geo-kernel}, and subsequently provide a detailed proof of Theorem~1. The proof of Lemma~\ref{lem:truncation-error-geo-kernel} is given in SM~\ref{append:sec:cs_truncate_error}.

Before proceeding with the further analysis, let us first present the formal statement of Theorem 1.  
\begin{theorem-non}[Formal statement of Theorem 1]\label{lem:Lowesa-no-trunc-prediction-error}
Following the notation in the main text, let $\tilde{\mathcal{U}}({\bx})$ in Eq.~(\ref{eqn:append:noisy-u(x)}) be the explored noisy quantum circuit. That is, each $\RZ$ gate in $\tilde{\mathcal{U}}({\bx})$ is affected by Pauli channel $\mathcal{N}_P(\PX ,\PY,\PZ )$ in Eq.~\eqref{append:eq:pauli_channel} and each Clifford operation is affected by General Pauli channel $\mathcal{M}$. Suppose that a given $N$-qubit state $\rho_0$ is evolved by this circuit, followed by operating with an observable $O$ with $O=\sum_{i=1} O_i$, $\sum_i\|O_i\|_{\infty}\leq B$, and the maximum locality of $\{O_i\}$ being $K$. 

Assume $\mathbb{E}_{\bx\sim [-\pi, \pi]^d}\|\nabla_{\bx} \Tr(\tilderho(\bx)O)\|_2^2\leq C$. Define the dataset as $\mathcal{T}_{\mathsf{s}}=\{\bxi \rightarrow \tilderho_T(\bxi)\}_{i=1}^n$, where $\bxi \sim  \textnormal{Unif}[-\pi, \pi]^d $  and $\tilderho_T(\bxi)$ refers to the Pauli-based classical shadow with $T$ snapshots. When the threshold of the frequency truncation is set  to $\Lambda=4C/\epsilon$ and  the number of trainable examples satisfies 
\begin{equation}
n = \left|\mathfrak{C}\left(\min \left\{\frac{4C}{\epsilon}, \frac{1}{2(p+\PZ )} \log\left(\frac{2}{\sqrt{\epsilon}}\right)\right\}\right)\right|  \frac{2  B^2 9^K}{\epsilon}  \log \left(\frac{2 \cdot \left|\mathfrak{C}\left(\min \left\{\frac{4C}{\epsilon}, \frac{1}{2(p+\PZ )} \log\left(\frac{2}{\sqrt{\epsilon}}\right)\right\}\right)\right|}{\delta}\right)
\end{equation}
with $\mathfrak{C}(\Lambda) =\{\bomega|\bomega \in \{0, \pm 1\}^d, ~s.t.~\|\bomega\|_0\leq \Lambda\}$ and $p=\min \{ \PX , \PY\}$, the state prediction model  $\hatsigma(\bx)$ in Eq.~(\ref{append:eqn:TriGeo-non-trunc-form}) achieves 
\begin{equation}
	\mathbb{E}_{\mathcal{T}_{\mathsf{s}}}[\hatsigma(\bx)] = \tilderho_{\Lambda}(\bx)
\end{equation}
and with  probability at least $1-\delta$, 
\begin{equation}
	\mathbb{E}_{\bx\sim [-\pi, \pi]^d} \left| \Tr(O\hatsigma(\bx)) - \Tr(O\tilderho(\bx))  \right|^2 \leq \epsilon.
\end{equation}	
\end{theorem-non}

To reach Theorem 1, we first use the triangle inequality to decouple the difference between the prediction $\Tr(O\hatsigma(\bx))$ and the expectation value $\Tr(O\tilderho(\bx))$ into the truncation error and the estimation error, i.e.,
\begin{subequations}
\begin{eqnarray}\label{append:eqn:decouple-dis-model-truth}
	&& \mathbb{E}_{\bx\sim [-\pi ,\pi]^d} \left[ \left| \Tr \left(O \hatsigma(\bx)\right) - \Tr \left(O\tilderho(\bx)\right)  \right|^2 \right] \\
	\leq && \Bigg(\sqrt{\mathbb{E}_{\bx\sim [-\pi ,\pi]^d} \left[ \left| \Tr \left(O \tilderho_{\Lambda}(\bx) \right) - \Tr(O\tilderho(\bx))  \right|^2 \right]} +
	\sqrt{\mathbb{E}_{\bx\sim [-\pi ,\pi]^d} \left[ \left| \Tr \left(O \hatsigma(\bx)\right) - \Tr(O\tilderho_{\Lambda}(\bx))  \right|^2 \right]}
	  \Bigg)^2. \label{eq:predic_error_decomposition}
\end{eqnarray}
\end{subequations} 
After decoupling, we now separately derive the upper bound of these two terms, where the relevant results are encapsulated in the following two lemmas with proofs deferred to SM~\ref{append:sec:cs_truncate_error} and SM~\ref{append:sec:cs_estimate_error}, respectively.

\begin{lemma}[Truncation error of $\tilderho_{\Lambda}(\bx)$]\label{lem:truncation-error-geo-kernel} 
 Following notations in Theorem 1, assume $\mathbb{E}_{\bx\sim [-\pi, \pi]^d}\|\nabla_{\bx} \Tr(\tilderho(\bx)O)\|_2^2\leq C$. The truncation error is upper bounded by
 \begin{equation}\label{append:eqn:noisy_truncation_error}
 	\mathbb{E}_{\bx\sim [-\pi, \pi]^d} \left|\Tr(O\tilderho_{\Lambda}(\bx)- \Tr(O\tilderho(\bx)) \right|^2 \leq \min\left\{\frac{C}{\Lambda}, B\cdot\exp\left(-2(p+\PZ )\Lambda\right)\right\}.
 \end{equation}
\end{lemma}

\begin{lemma}[Estimation error of $\hatsigma(\bx)$, adapted from Lemma~F.2 in Ref.~\cite{du2024efficient}]\label{lem:estimation-error-geo-kernel} 
 Following notations in Theorem 1, with probability at least $1-\delta$, the estimation error of $\hatsigma(\bx)$ induced by finite training examples $\mathcal{T}_{\mathsf{s}}=\{\bxi \rightarrow \tilderho(\bxi)\}_{i=1}^n$ is upper bounded by
 \begin{equation}\label{append:eqn:trigo-estimation-error}
 	\mathbb{E}_{\bx\sim [-\pi, \pi]^d} \left[\left |\Tr(O  \hatsigma(\bx)) - \Tr(O \tilderho_{\Lambda}(\bx))\right|^2 \right] \leq    |\mathfrak{C}(\Lambda)|  \frac{1}{2n} B^2 9^K  \log \left(\frac{2 \cdot |\mathfrak{C}(\Lambda)| }{\delta}\right),
 \end{equation}
 where $\mathfrak{C}(\Lambda) =\{\bomega|\bomega \in \{0, \pm 1\}^d, ~s.t.~\|\bomega\|_0\leq \Lambda\}$ refers to the set of truncated frequencies.
\end{lemma}

Supported by these two lemmas, we are now ready to present the proof of Theorem 1.
\begin{proof}[Proof of Theorem 1]
The difference between the prediction and ground truth can be obtained by integrating Lemma~\ref{lem:truncation-error-geo-kernel} and Lemma~\ref{lem:estimation-error-geo-kernel} into Eq.~(\ref{append:eqn:decouple-dis-model-truth}). Mathematically, with probability at least $1-\delta$, we obtain
\begin{eqnarray}
	&& \mathbb{E}_{\bx\sim [-\pi ,\pi]^d} \left[ \left| \Tr(O \hatsigma(\bx)) - \Tr(O\tilderho(\bx))  \right|^2 \right]  
     \nonumber \\
     \leq && 
     \left(\min\left \{ \sqrt{\frac{C}{\Lambda}}, B\cdot\exp\left(-2(p+\PZ )\Lambda\right)  \right\} 
     +  \sqrt{|\mathfrak{C}(\Lambda)|  \frac{1}{2n} B^2 9^K  \log \left(\frac{2 \cdot |\mathfrak{C}(\Lambda)|}{\delta}\right)} \right)^2.
\end{eqnarray}
To ensure the average prediction error is upper bounded by $\epsilon$, it is sufficient to showcase when the inner two terms are upper bounded by $\sqrt{\epsilon}/2$. For the first term, the condition is satisfied when
\begin{equation}
	\min\left \{ \sqrt{\frac{C}{\Lambda}}, B\cdot\exp\left(-2(p+\PZ )\Lambda\right)  \right\} \leq \frac{\sqrt{\epsilon}}{2} \Leftrightarrow \Lambda \geq \min \left\{\frac{4C}{\epsilon}, \frac{1}{2(p+\PZ )} \log\left(\frac{2B}{\sqrt{\epsilon}}\right)\right\}. 
\end{equation}
For the second term, we have 
\begin{eqnarray}
&&	\sqrt{|\mathfrak{C}(\Lambda)|  \frac{1}{2n} B^2 9^K  \log \left(\frac{2 \cdot |\mathfrak{C}(\Lambda)|}{\delta}\right)} \leq \frac{\sqrt{\epsilon}}{2}   
\Leftrightarrow  n \geq  |\mathfrak{C}(\Lambda)|  \frac{2  B^2 9^K}{\epsilon}  \log \left(\frac{2 \cdot |\mathfrak{C}(\Lambda)|}{\delta}\right). 
\end{eqnarray}  

Taken together, with probability $1-\delta$, the prediction error is upper bounded by $\epsilon$ when the number of training examples satisfies
\begin{equation}
	n \geq \left|\mathfrak{C}\left(\min \left\{\frac{4C}{\epsilon}, \frac{1}{2(p+\PZ )} \log\left(\frac{2B}{\sqrt{\epsilon}}\right)\right\}\right)\right|  \frac{2  B^2 9^K}{\epsilon}  \log \left(\frac{2 \cdot \left|\mathfrak{C}\left(\min \left\{\frac{4C}{\epsilon}, \frac{1}{2(p+\PZ )} \log\left(\frac{2B}{\sqrt{\epsilon}}\right)\right\}\right)\right|}{\delta}\right).
\end{equation}

\end{proof}

\subsection{Truncation error of the classical learning model (Proof of Lemma \ref{lem:truncation-error-geo-kernel})}\label{append:sec:cs_truncate_error}
The proof of Lemma~\ref{lem:truncation-error-geo-kernel} leverages the following two lemmas.

\begin{lemma}[Adapted from Lemma~F.1, Ref.~\cite{du2024efficient}]
\label{lem:truncation_error_idea}
Following the notation in Theorem 1, assuming $\mathbb{E}_{\bx\sim [-\pi, \pi]^d}\|\nabla_{\bx} \Tr(\tilderho(\bx)O)\|_2^2\leq C$, the truncation error induced by removing high-frequency terms of $\tilderho(\bx)$ under the trigonometric  expansion with $\|\bomega\|_0\leq \Lambda$ is upper bounded by
 \begin{equation}\label{append:eqn:truncation_error}
 	\mathbb{E}_{\bx\sim [-\pi, \pi]^d} \left|\Tr(O\tilderho_{\Lambda}(\bx)- \Tr(O\tilderho(\bx)) \right|^2 \leq \frac{C}{\Lambda}.
 \end{equation}
\end{lemma}

The result of Lemma~\ref{lem:truncation_error_idea} was established for the noiseless state $\rho(\bx)$.  We now explain why it can be adapted to the noisy states $\tilderho(\bx)$. Specifically, the target functions $f(\tilde{\rho}\bx, O)=\Tr(\tilderho(\bx)O)$ for both the noiseless state $\rho(\bx)$ and the noisy state $\tilderho(\bx)$ have the same form of Fourier expansions given in Eq.~\eqref{eq:fourier_expansion}, differing only in their Fourier coefficients $\{\tilde{\bm{\alpha}}_{\bomega}\}$ and $\{\bm{\alpha}_{\bomega}\}$. Given the same constraints on these coefficients, namely, $\mathbb{E}_{\bx\sim [-\pi, \pi]^d}\|\nabla_{\bx} \Tr(\tilderho(\bx)O)\|_2^2\leq C$,  the truncation error in Eq.~\eqref{append:eqn:truncation_error} for both the noiseless state $\rho_{\Lambda}(\bx)$ and the noisy state  $\tilderho_{\Lambda}(\bx)$ has the same upper bound.

\begin{lemma}[Theorem 2, Ref.~\cite{fontana2023classical}]
\label{lem:truncation_error_noise}
Following the notation in Theorem 1, the truncation error induced by removing high-frequency terms of $\tilderho(\bx)$ under the trigonometric  expansion with $\|\bomega\|_0\leq \Lambda$ is upper bounded by
 \begin{equation}\label{append:eqn:trigo-expan-traget-state}
 	\mathbb{E}_{\bx\sim [-\pi, \pi]^d} \left|\Tr(O\tilderho_{\Lambda}(\bx)- \Tr(O\tilderho(\bx)) \right|^2 \leq   \exp\left(-4(p+\PZ )\Lambda\right)\cdot \|O\|_{\infty}^2,
 \end{equation}
 where $p:=\min \{\PX , \PY\}$ and $\PZ $ are the noise parameters of the Pauli channel $\mathcal{N}_P$ applied to the rotation gates.
\end{lemma}

\begin{proof}[Proof of Lemma \ref{lem:truncation-error-geo-kernel}]
    The results of Lemma \ref{lem:truncation-error-geo-kernel} can be immediately obtained by combining Lemma~\ref{lem:truncation_error_idea} and Lemma~\ref{lem:truncation_error_noise}.
\end{proof}

\subsection{Estimation error of the classical learning model (Proof of Lemma~\ref{lem:estimation-error-geo-kernel}) }  \label{append:sec:cs_estimate_error}

The results presented in Ref.~\cite[Lemma~F.2]{du2024efficient} were derived under the assumption of a noiseless state $\rho(\bx)$, where the prediction model follows the same form as that for the noisy state $\tilderho(\bx)$ defined in Eq.~(\ref{append:eqn:TriGeo-non-trunc-form}), namely,
\[\sigma(\bx)=\frac{1}{n}\sum_{i=1}^n\hatkappa\left(\bx, \bxi\right)\rho_T(\bxi).\] 
The only difference is that the classical shadow representation is obtained from the noiseless state $\rho(\bx)$.
To extend this analysis to the noisy case, we first review the key steps of the derivation for the noiseless scenario and explain how these steps adapt to the noisy setting. 

The first step is to show that the state prediction model $\sigma(\bx)$ is equal to the trigonometric expansion of the truncated target quantum state $\rho_{\Lambda}(\bx)$ when taking the expectation over the training data, namely $\mathbb{E}_{\mathcal{T}}[\sigma(\bx)]=\rho_{\Lambda}(\bx)$, which includes the randomness from the sampled inputs $\bx$ and classical shadow. The only condition for this equality to hold is that the input $\bx$ is uniformly distributed over the interval $[-\pi, \pi]^d$, which is independent of the specific form of the quantum state. Therefore, this equality also holds for the truncated noisy state $\tilderho_{\Lambda}(\bx)$.

The second step involves employing the classical concentration inequality to upper bound the estimation error of the proposed state prediction model $\hatsigma(\bx)$ when considering the expectation value of an unseen state under the specified observable $O$,  i.e.,  $\mathbb{E}_{\bx\sim [-\pi,\pi]^d} [  | \Tr(O \sigma(\bx)) - \Tr(O\rho_{\Lambda}(\bx)) |^2 ]$. In this step, the quantum-related term needed for the derivation only involves calculating an upper bound for $\Tr(O\rho_{1}(\bx))$, where $\rho_{1}(\bx)$ is the classical shadow representation of $\rho(\bx)$ with snapshot being $T=1$. In particular, this is expressed as
\begin{equation}
    |\Tr(O\rho_{\Lambda}(\bx))| = \left|\sum_{\bomega \in \mathfrak{C}(\Lambda)}|\Phi_{\bomega}(\bx) \Tr(\rho_1(\bx)O) \right| 
 	\leq    |  \Tr(\rho_1(\bx)O)| 
 	\leq    \|O\|_{\infty} \|\rho_1(\bx)\|_1 
 	= 3^K B.
\end{equation}
Notably, this inequality depends only on the locality parameter $K$ and the norm bound $B$ of the observable $O$, and it holds for all quantum states.

Taken together, we can conclude that the results derived in Lemma~F.2 of Ref.~\cite{du2024efficient} for the noiseless state also hold for the noisy state.

\section{Learnability of the predictive surrogate $h_{\mathsf{qs}}$ (Proof of Theorem~2)}
In this section, we elucidate the construction of the predictive surrogate from a ridge regression model for quantum simulation and provide provable guarantees for the constructed learning model. To achieve this, we first derive the feature map for the quantum circuits used to simulate the quantum many-body systems in SM~\ref{append:sec:feature_map_mb}, where the parameters are correlated. Moreover, we introduce the implementation of the ridge regression model under the derived feature map in SM~\ref{append:sec:implement_regression} and present the theoretical analysis of the learnability of the proposed ML model in the subsequent subsections.
 
\subsection{Trigonometric monomial feature maps for quantum circuits with correlated parameters}\label{append:sec:feature_map_mb}
Let us first derive the trigonometric monomial feature maps for quantum circuits with \textit{correlated parameters}. Recall that the quantum circuit with correlated parameters is given by \[U(\bx)=\prod_{l=1}^{L}\left(\prod_{j=1}^{d/L}\RZ(\bx_j)V_j\right),\] 
where the inputs $\bx_j$ vary within each layer but remain the same across different layers, $L$ represents the number of correlated parameters. According to the explanation in SM~A, the PTM representation of a quantum state evolved by $U(\bx)$ is 
\begin{equation}\label{append:eq:quantum_state_param_share}
    \rho(\bx) \equiv \ket{\psi(\bx)} \bra{\psi(\bx)}  = U(\bx)(\ket{0^{N}
    }\bra{0^{N}} U(\bx)^{\dagger} =  \sum_{\bomega}\Phi_{\bomega}(\bx) \llangle 0| \bUnitary_{\bomega}^{\dagger}.
\end{equation}
For quantum circuits with uncorrelated parameters, $\Phi_{\bomega}(\bx)$ with $\bomega\in \{0, 1, -1\}^{d}$ refers to the trigonometric basis with the values 
\[
	\Phi_{\bomega}(\bx) = \prod_{i=1}^{d} \Psi(\bx_i;\bomega_i) = \prod_{i=1}^{d} \begin{cases}
		 1 ~ & \textnormal{if}~ \bomega_i = 0 \\
		 \cos(\bx_i) & \textnormal{if}~\bomega_i = 1 \\
		 \sin(\bx_i) &  \textnormal{if}~ \bomega_i = -1
	\end{cases},
\]
where each trigonometric monomial $\cos(\bx_i)$ or $\sin(\bx_i)$ only appears once in $\Phi_{\bomega}(\bx)$.

\smallskip
We now consider the scenario with correlated parameters. In this context, it is often helpful to re-index the subscript $i$ of parameters $\bx_i$ by a pair of indices  $(j,l)$ with 
\begin{equation}
	j = 1, ..., d/L, \quad \ell =0,...,L-1, \quad i = j + \ell d. 
\end{equation}
An intuition behind this reformulation is considering the layer-wise structure of the explored $U(\bx)$ with $\bx$ exhibiting correlations across layers. Under these notations, the variable $\bx$ and the frequency $\bw$ have the relation
\begin{equation}
	\bx_i = \bx_{j + \ell d} \quad \text{and} \quad \bomega_i = \bomega_{j + \ell d}.
\end{equation}
Accordingly, the explicit form of  $\Phi_{\bomega}(\bx)$ becomes
\begin{equation}
\Phi_{\bomega}(\bx) = \prod_{j=1}^{d/L} 
\prod_{\ell=0}^{L-1}
\Psi\Bigl(\bx_{\,j + \ell\,d}\,;\; \bomega_{\,j + \ell\,d}\Bigr)  = \prod_{j=1}^{d/L} 
\prod_{\ell=0}^{L-1}
\Psi\bigl(\bx_j;\; \bomega_{j,\ell}\bigr).
\end{equation}

Observe that for a fixed $j$ (that is, for one of the $d/L$ base coordinates $\bx_j$), we multiply together
\begin{equation}
	\prod_{\ell=0}^{L-1} \Psi\bigl(\bx_j;\, \bomega_{j,\ell}\bigr)= \Bigl[\cos\bigl(\bx_j\bigr)^{\,N^+_j(\bomega)}
\;\sin\bigl(\bx_j\bigr)^{\,N^-_j(\bomega)}\Bigr].
\end{equation}
Here we define 
\[N^+_j(\bomega) := \sum_{\ell=0}^{L-1} \mathbbm{1}\{\bomega_{j,\ell}=1\},\quad \text{and} \quad N^-_j(\bomega) := \sum_{\ell=0}^{L-1} \mathbbm{1}\{\bomega_{j,\ell}=-1\},\]
where $N^+_j$ and $N^-_j$ quantify how many times we see $+1$ and $-1$ for fixed $j$, respectively, and $\mathbbm{1}\{\cdot\}$ is the indicator function.   Putting everything together, the entire feature map is
\begin{equation}\label{append:eq:feature_map_para_share}
\Phi_{\bomega}(\bx) = \prod_{j=1}^{d/L}
\Bigl[\cos\bigl(\bx_j\bigr)^{\,N^+_j(\bomega)}
\;\sin\bigl(\bx_j\bigr)^{\,N^-_j(\bomega)}\Bigr].
\end{equation}
In words, we only need to track how many times each of the $d/L$  angles $\bx_j$ appears with a cosine or a sine label across the $L$ repetitions.
 
\subsection{Implementation of the regression model $h_{\mathsf{qs}}$ for quantum many-body dynamics}\label{append:sec:implement_regression}
We next explain how to construct the predictive surrogate $h_{\mathsf{qs}}$ for quantum circuits $U(\bx)$ with correlated parameters by solving a regression model. Recall that the goal of $h_{\mathsf{qs}}$ is to predict the mean-value behavior of the noisy quantum state $\tilde{\rho}(\bx)$ for a specific observable $O$. In this regard, instead of using the classical shadow representation $\tilderho_T(\bx)$ as the label of the input $\bx$, we consider the training dataset consisting of $\mathcal{T}_{\mathsf{mb}} = \{(\bxi, \yi)\}_{i=1}^n$, where $\yi$ approximates the expectation value of the quantum circuit with a maximal estimation error $\epsilon_l$, namely $|\yi - \Tr(\tilderho(\bxi)O)| < \epsilon_{l}$. The observable, which is known a prior, has a form $O = \sum_{P \in \{I, X, Y, Z\}^{\otimes N}} \alpha_P P$  and a bounded norm $\|O\|_{\infty}\le B$. 

To construct  $h_{\mathsf{qs}}(\bx)$ that emulates the expectation value $\Tr(\tilderho(\bx)O)$ output by quantum processors, the first step is to select an appropriate hypothesis space $\mathcal{H}$ that satisfies two conditions: (i) there exists a function $h_{\mathsf{qs}}(\bx)\in \mathcal{H}$ with a small prediction error $\mathbb{E}_{\bx\sim \mathcal{X}} \left| h_{\mathsf{qs}}(\bx) - \Tr(O\tilderho(\bx))  \right|^2 \leq \epsilon$, and (ii) the construction of $h_{\mathsf{qs}}(\bx)$ is computationally efficient. To achieve this, we reformulate the target function $\Tr(\tilderho(\bx)O)$ in terms of PTM as described in Eq.~\eqref{append:eq:quantum_state_param_share}. Specifically, we have
\begin{equation}\label{append:eq:linear_func_para_share}               
   f(\tilderho(\bx), O)=\Tr(\tilderho(\bx)O) = \sum_{\bomega\in \Omega} \Phi_{\bomega}(\bx) \llangle 0| \tilde{\bUnitary}_{\bomega}^{\dagger} |O \rrangle = \langle \bm{\Phi}(\bx) , \tilde{\bm{\alpha}}\rangle, 
\end{equation} 
where $\bm{\Phi}(\bx)$ is the feature vector composed of ${\Phi}_{\bomega}(\bx)$, $\tilde{\bm{\alpha}}$ is the coefficient vector, with elements indexed by $\bomega$, and $\tilde{\bm{\alpha}}_{\bomega} = \llangle 0 | \tilde{\bUnitary}_{\bomega}^{\dagger} |O \rrangle$. The expression in Eq.~\eqref{append:eq:linear_func_para_share} shows that the target function $\Tr(\tilderho(\bx)O)$ belongs to the hypothesis class \[\mathcal{H}=\Big\{ h_{\bm{\mathrm{w}}}(\bx) = \braket{\bm{\Phi}(\bx),\bm{\mathrm{w}}} \big | \bm{\mathrm{w}} \in \mathbb{R}^{\dim(\bm{\Phi})} \Big\}.\] 
While this hypothesis class contains the surrogate $h_{\mathsf{qs}}$ with a small prediction error, the dimension of the feature vector $\bm{\mathrm{w}}$ is $\dim(\bm{\Phi}) =3^{d}$, which grows exponentially with the number of rotation gates $d$, making the construction of $h_{\mathsf{qs}}$ inefficient for large gate counts $d=\mathcal{O}(N)$.

To balance prediction error with the computational complexity of $h_{\mathsf{qs}}$, we consider the hypothesis class with \textit{truncated features}, i.e., \[\mathcal{H}_{\mathfrak{C}(\Lambda)}=\{h_{\bm{\mathrm{w}}}(\bx)=\braket{\bm{\Phi}_{\mathfrak{C}(\Lambda)}(\bx), \bm{\mathrm{w}}}|\bm{\mathrm{w}}\in \mathbb{R}^{|\mathfrak{C}(\Lambda)|}\},\] where the features $\Phi_{\bomega}(\bx)$ in Eq.~(\ref{append:eq:feature_map_para_share}) with high frequency $\|\bomega\|_0>\Lambda$ are removed. In this regard, the predictive surrogate $h_{\mathsf{qs}}(\bx)=\braket{\bm{\Phi}_{\mathfrak{C}(\Lambda)}(\bx), \hat{\bm{\mathrm{w}}}}$ is constructed by solving the following optimization problem
\begin{equation}
    \label{append:eq:ridge_reg_model_truncated}
    h_{\mathsf{qs}}(\bx)=\arg\min_{h_{\bm{\mathrm{w}}}\in \mathcal{H}_{\mathfrak{C}(\Lambda)} } \frac{1}{n}\sum_{i=1}^n \left(h_{\bm{\mathrm{w}}}(\bxi)-\yi\right)^2 + \lambda \|\bm{\mathrm{w}}\|_2,
\end{equation}
where $\lambda$ is the regularization parameter. 
Denote $f_{\mathfrak{C}(\Lambda)}(\bx)=\langle \bm{\Phi}_{\mathfrak{C}({\Lambda})}(\bx) , \tilde{\bm{\alpha}}_{\mathfrak{C}(\Lambda)}\rangle$ as the truncated representation of the target function $f(\bx, O)= \langle \bm{\Phi(x)},  \tilde{\bm{\alpha}}\rangle$ in Eq.~(\ref{append:eq:linear_func_para_share}), where $\tilde{\bm{\alpha}}_{\mathfrak{C}(\Lambda)}$ is the sub-vector of $\tilde{\bm{\alpha}}$ consisting of the element $\tilde{\bm{\alpha}}_{\bomega}$ with the index $\bomega \in \mathfrak{C}(\Lambda)$. When taking $\lambda \ge \|\tilde{\bm{\alpha}}_{\mathfrak{C}(\Lambda)}\|_2$, the truncated function $f_{\mathfrak{C}(\Lambda)}(\bx)$ is contained in the solution space, ensuring the existence of $h_{\mathsf{qs}}\in \mathcal{H}_{\mathfrak{C}(\lambda)}$ with a small prediction error for a large truncation threshold $\Lambda$.  

Under the framework of ridge regression, we establish the theoretical guarantee for the efficiency of the constructed predictive surrogate $h_{\mathsf{qs}}(\bx)$, as stated in the following theorem. 

\begin{theorem-non}[Restatement of Theorem 2]\label{append:thm:ridge_predict_error}
    Consider the training dataset $\mathcal{T}_{\mathsf{mb}}=\{\bxi,\yi\}_{i=1}^n$, where $\bxi\in \mathcal{X}\subset [-R, R]^{d}$ is sampled from an arbitrary distribution $\mathbb{D}$, with the maximal value $R>0$, and $y_i$ is the estimated expectation value of observable $O$ on the parameterized state $\tilderho(\bxi)$ in Eq.~\eqref{append:eq:quantum_state_param_share}. This state is affected by the Pauli noisy channel $\mathcal{N}_P(\PX ,\PY,\PZ )$ defined in Eq.~\eqref{append:eq:pauli_channel}. Let $q=1-2(p+\PZ )$ with $p:=\min\{\PX ,\PY\}$ and assume $q(1+R)<1/e$. Let $\Lambda$ be the threshold of the truncated frequency set $\mathfrak{C}(\Lambda)$, and $\epsilon_l$ the maximal estimation error of $\{\yi\}$ in  $\mathcal{T}_{\mathsf{mb}}$, namely, $\max_{i\in[n]}|\yi-\Tr(O\tilderho(\bxi)|\le \epsilon_l$. Define $\epsilon=16B^2(deq(1+R)/\Lambda)^{2\Lambda}$.  
    
When the following conditions are satisfied: (i) $\epsilon_l\le\sqrt{\epsilon}/4$, (ii) the frequency is truncated to $\Lambda>deq(1+R)$, (iii) the number of training examples satisfies \[n=\left(\frac{1}{q(1+R)}\right)^{4\Lambda}\cdot \frac{\log(1/\delta)}{9},\] then the predictive surrogate $h_{\mathsf{qs}}(\bx)$ constructed by solving the ridge regression model in Eq.~\eqref{append:eq:ridge_reg_model_truncated}, achieves with probability at least $1-\delta$,
    \begin{equation}\label{append:eq:ridge_predict_error_bound}
	\mathbb{E}_{\bx\sim \mathbb{D}} \left| h_{\mathsf{qs}}(\bx) - \Tr(O\tilderho(\bx))  \right|^2 \leq \epsilon.
\end{equation}
\end{theorem-non}

\subsection{Prediction error bound for ridge regression model (Proof of Theorem~2)}
The proof of Theorem~2 leverages the results from statistical learning theory regarding the prediction error of kernel ridge regression algorithms \cite{saunders1998ridge, mehyar2018foundations}.

\begin{lemma}
    [Theorem 11.11,  Ref.~\cite{mehyar2018foundations}]\label{append:lem:ridge_reg_prediction_error}
    Let $K:\mathcal{X}\times \mathcal{X}\to \mathbb{R}$ be a positive definite symmetric kernel, $\Phi:\mathcal{X}\to\mathcal{H}$ be a feature mapping associated to $K$, and $\mathcal{H}=\{\bx\to \braket{\Phi(\bx),\bm{\mathrm{w}}}: \|\bm{\mathrm{w}}\|_2\le \lambda\}$. Let $\mathcal{S}=\{(\bxi,\yi)\}_{i=1}^n\in (\mathcal{X}\times \mathcal{Y})^{n}$. Assume that there exists $r>0$ such that $K(\bx,\bx)\le r^2$ and $M>0$ such that $|h(\bx)-y|<M$ for all $(\bx,\by)\in\mathcal{X}\times \mathcal{Y}$. Then, for $\delta>0$, with probability at least $1-\delta$, the following inequality holds for all $h\in \mathcal{H}$:
    \begin{equation}\label{append:eq:cl_ridge_predict_error}
        \mathbb{E}_{\bx\sim \mathcal{X}} \left| h(\bx) - \Tr(O\tilderho(\bx))  \right|^2 \le \frac{1}{n}\sum_{i=1}^n \left(h(\bxi)-\yi\right)^2 + 4M\sqrt{\frac{r^2\lambda^2}{n}} + M^2\sqrt{\frac{\log(1/\delta)}{2n}},
    \end{equation}
    where the first term on the right-hand side $\sum_{i=1}^n \left(h(\bxi)-\yi\right)^2/n$ represents the training error.
\end{lemma}

This theorem provides a general upper bound on the prediction error for the ridge regression model. In the context of our learning setting, it reduces to analyzing the terms in Eq.~\eqref{append:eq:cl_ridge_predict_error} for the specific classical function $h_{\mathsf{qs}}(\bx) $ in Eq.~\eqref{append:eq:ridge_reg_model_truncated}, including the terms $M, r, \lambda$, and the training error $\sum_{i=1}^n \left(h_{\mathsf{qs}}(\bxi)-\yi\right)^2/n$. The derivation of these terms employs the following lemmas about the binomial inequalities and the worst-case truncation error of the learning model $f_{\mathfrak{C}(\Lambda)}(\bx)= \langle \bm{\Phi}_{\mathfrak{C}(\Lambda)}(\bx), \tilde{\bm{\alpha}}_{\mathfrak{C}(\Lambda)}\rangle$ established on the truncated feature map $\bm{\Phi}_{\mathfrak{C}(\Lambda)}(\bx)$.

\begin{lemma}[Adapted from Lemma~4 in Ref.~\cite{lerch2024efficient}]
\label{append:lem:binomial_inequality_1}
    Let $a, b, x, y>0$ be positive constants satisfying $a>b$ and $x \le b/a$. Then we have
    \begin{equation}\label{append:eq:binomial_inequality_1}
        \sum_{i=b}^{a} \binom{a}{i} x^i\le \left( \frac{eax}{b}\right)^b, \quad \mbox{and} \quad \sum_{i=0}^{b} \binom{a}{i} y^i\le \left(\frac{ae^{y} }{b} \right)^b ,
    \end{equation}
    where $e$ is the Euler number.
\end{lemma}

\begin{lemma}\label{append:lem:binomial_inequality_2}
    Let $a, b, x, q>0$ be positive constants  satisfying $a>b$ and $q(1+x) \le b/a$. Then we have
    \begin{equation}
        \sum_{i=b}^{a}\sum_{j=0}^i \binom{a}{i}\binom{i}{j} x^{j}q^i\le \left( \frac{eaq(1+x)}{b}\right)^b
    \end{equation}
\end{lemma}

\begin{lemma}\label{append:lem:ridge_training_error} 
    Following the notation in Theorem~2, let $f_{\mathfrak{C}(\Lambda)}(\bx)= \langle \bm{\Phi}_{\mathfrak{C}(\Lambda)}(\bx) ,  \tilde{\bm{\alpha}}_{\mathfrak{C}(\Lambda)} \rangle=\sum_{\bomega\in \mathfrak{C}(\Lambda)}\Phi_{\bomega}(\bx) \cdot \tilde{\bm{\alpha}}_{\bomega}$ be the truncated function of $\Tr(\tilderho(\bx)O)$ with $\Lambda$ being the truncation threshold of frequency and $\mathfrak{C}(\Lambda) =\{\bomega|\bomega \in \{0, \pm 1\}^d, \|\bomega\|_0\leq \Lambda\}$ being the set of truncated frequencies, $\mathcal{X}\subset [-R,R]^d$ the value space of $\bx$ with $R>0$ being any positive number. Then the worst-case truncation error of $f_{\mathfrak{C}(\Lambda)}(\bx)$ yields
    \begin{equation}
        \max_{\bx\in \mathcal{X}} \left|f_{\mathfrak{C}(\Lambda)}(\bx)-\Tr(\tilderho(\bx)O)\right| \le \left(\frac{edq(1+R)}{\Lambda+1} \right)^{\Lambda+1} ,
    \end{equation}
    where $e$ is the Euler number, $q=1-2(p+\PZ )$ with $p=\min\{\PX ,\PY\}$, and $\PX , \PY, \PZ \ge 0$ are the noise parameters of the Pauli noise channel $\mathcal{N}_P$.
\end{lemma}

We are now ready to present the proof of Theorem~2.

\begin{proof}
    [Proof of Theorem~2] 
We first reframe the theorem in our setting. Consider the input space $\mathcal{X}$ to be the parameter space $[-R,R]^d$. Our feature mapping is $\bm{\Phi}_{\mathfrak{C}(\Lambda)}(\bx)$, which consists of the triangular monomial $\Phi_{\bomega}$ of order less than $\Lambda$, namely $\|\bomega\|_0\le \Lambda$. Then, the kernel associated with this feature map is \[K(\bx,\bx')= \langle \bm{\Phi}_{\mathfrak{C}(\Lambda)}(\bx) ,  \bm{\Phi}_{\mathfrak{C}(\Lambda)}(\bx')\rangle.\] One can easily check that $K(\bx,\bx')$ is positive definite and symmetric. The hypothesis class is defined as $\mathcal{H}_{\mathfrak{C}(\Lambda)}=\{h(\bx)=\braket{\bm{\Phi}_{\mathfrak{C}(\Lambda)}(\bx), \bm{\mathrm{w}}}:\|\bm{\mathrm{w}}\|_2\le \lambda \}$. Moreover, we are given the training data $\mathcal{T}_{\mathsf{mb}}=\{(\bxi,\yi)\}_{i=1}^n$ for $\bxi\in [-R,R]^d$ and $\yi\in\mathbb{R}$, where $|\yi-\Tr(O\tilderho(\bxi))|\le \epsilon_l$.

We set $\lambda \ge \|\tilde{\bm{\alpha}}_{\mathfrak{C}(\Lambda)}\|_2$ such that the truncated function $f_{\mathfrak{C}(\Lambda)}(\bx)=\langle \bm{\Phi}_{\mathfrak{C}(\Lambda)}(\bx) , \tilde{\bm{\alpha}}_{\mathfrak{C}(\Lambda)}\rangle=\sum_{\bomega \in \mathfrak{C}(\Lambda)}\Phi_{\bomega}(\bx) \cdot \tilde{\bm{\alpha}}_{\bomega}$ belongs to the optimization space $\mathcal{H}_{\mathfrak{C}(\Lambda)}$ of the ridge regression model,  i.e., $\mathcal{H}_{\mathfrak{C}(\Lambda)}=\{\sum_{\bomega \in \mathfrak{C}(\Lambda)}\Phi_{\bomega}(\bx) \cdot \bm{\mathrm{w}}| \|\bm{\mathrm{w}}\|_2\le \lambda\le \|\tilde{\bm{\alpha}}_{\mathfrak{C}(\Lambda)}\|_2 \}$. In this regard, the training error of the optimal solution $h_{\mathsf{qs}}(\bx)$ for the ridge regression model is less than the training error of the truncated function $f_{\mathfrak{C}(\Lambda)}(\bx)$, namely
    \begin{align}\label{append:eq:ridge_tr_err_truncat_err}
        \frac{1}{n}\sum_{i=1}^n \left(h_{\mathsf{qs}}(\bxi)-\yi\right)^2 
        \le   \frac{1}{n}\sum_{i=1}^n \left(f_{\mathfrak{C}(\Lambda)}(\bxi)-\yi\right)^2. 
    \end{align}
    In the following, we separately derive the upper bounds of the training error $\sum_{i=1}^n (h_{\mathsf{qs}}(\bxi)-\yi)^2/n$, and the terms $M,r,\lambda$ in Eq.~\eqref{append:eq:cl_ridge_predict_error}.

    \noindent \underline{\textit{The upper bound of training error $\sum_{i=1}^n (h_{\mathsf{qs}}(\bxi)-\yi)^2/n$.}} As indicated by Eq.~\eqref{append:eq:ridge_tr_err_truncat_err}, we only need to bound the training error of the truncated function $f_{\mathfrak{C}(\Lambda)}(\bxi)$. Let $i^*$ be defined as the index of training data $(\bx^{(i*)}, y^{(i*)})$ with maximal error for $f_{\mathfrak{C}(\Lambda)}(\bxi)$, namely,
    \begin{equation}
        i^* = \arg\min_{i\in[n]} \Big(f_{\mathfrak{C}(\Lambda)}(\bxi)-\yi\Big)^2,
    \end{equation}
    In this regard, we only need to analyze the upper bound of the training error on the single data point $(\bx^{(i^*)}, y^{(i^*)})$. More specifically, the training error in  Eq.~(\ref{append:eq:ridge_tr_err_truncat_err}) yields 
    \begin{align}\label{append:eq:ridge_reg_training_error_bound}
    \frac{1}{n}\sum_{i=1}^n \left(h_{\mathsf{qs}}(\bxi)-\yi\right)^2 
        \le &  \frac{1}{n}\sum_{i=1}^n \left(f_{\mathfrak{C}(\Lambda)}(\bxi)-\yi\right)^2
        \nonumber \\
        \le & \left|f_{\mathfrak{C}(\Lambda)}(\bx^{(i^*)})-y^{(i^*)}\right|^2
        \nonumber \\
        \le
        & \left(\left|f_{\mathfrak{C}(\Lambda)}(\bx^{(i^*)})-\Tr(\tilde{\rho}(\bx^{(i^*)})O)\right| +\left|\Tr(\tilde{\rho}(\bx^{(i^*)})O)-y^{(i^*)}\right|\right)^2
        \nonumber  \\
        \le   & \left(B\left(\frac{edq(1+R)}{\Lambda+1} \right)^{\Lambda+1}+\epsilon_l\right)^2,
    \end{align}
    where the third inequality employs the triangle inequality,  the last inequality follows Lemma~\ref{append:lem:ridge_training_error} and employs the condition that estimation error on the training data is bounded by $\epsilon_l$.

\smallskip
    \noindent \underline{\textit{The upper bound of regularization parameter $\|\tilde{\bm{\alpha}}_{\mathfrak{C}(\Lambda)}\|_2\le \lambda$.}} 
    Utilizing the result of Lemma~\ref{append:lem:noisy_coefficient} such that $|\tilde{\bm{\alpha}}_{\bomega}| \le q^{\|\bomega\|_0}|{\bm{\alpha}}_{\bomega}|$ with $\bm{\alpha}_{\bomega}=\llangle 0| (\bUnitary_{\bomega})^{\dagger} |O \rrangle$ being the zero-noise coefficient for $\PX =\PY=\PZ =0$, we have
    \begin{equation}\label{append:eq:ridge_reg_reg_para_bound}
        \|\tilde{\bm{\alpha}}_{\mathfrak{C}(\Lambda)}\|_2^2  =  \sum_{\bomega\in\mathfrak{C}(\Lambda)}\tilde{\bm{\alpha}}_{\bomega}^2  \le  \sum_{\bomega\in\mathfrak{C}(\Lambda)}q^{2\|\bomega\|_0}(\bm{\alpha}_{\bomega})^2 \le  \sum_{\bomega\in\mathfrak{C}(\Lambda)}q^{2\|\bomega\|_0} B^2
    \end{equation}
    where $q=1-2(p+\PZ )$ and $p:=\min\{\PX ,\PY\}$, the first inequality employs Lemma~\ref{append:lem:noisy_coefficient}, and the second inequality follows that $|\bm{\alpha}_{\bomega}| \le \|O\|_{\infty}\le B$.
    To further simplify the term on the right-hand side of this equation, reformulating the summation over $\bomega\in\mathfrak{C}(\Lambda)$ into the summation over the norm of $\|\bomega\|_0\in [\Lambda]$ yields
    \begin{align}\label{append:eq:ridge_reg_reg_para_bound_2}
        \sum_{\bomega\in\mathfrak{C}(\Lambda)}q^{2\|\bomega\|_0} B^2
        &  
        = B^2\sum_{\bomega:\|\bomega\|_0=0}^{\Lambda}\sum_{\|\bomega_{\#-1}\|_0=0}^{\|\bomega\|_0} \binom{d}{\|\bomega\|_0}\binom{\|\bomega\|_0}{\|\bomega_{\#-1}\|_0}q^{2\|\bomega\|_0}
        \nonumber \\
        &  
        = B^2\sum_{\bomega:\|\bomega\|_0=0}^{\Lambda} \binom{d}{\|\bomega\|_0}2^{\|\bomega\|_0} q^{2\|\bomega\|_0} 
        \nonumber \\
        & \le B^2\left(\frac{de^{2q^2}}{\Lambda}\right)^{\Lambda },  
    \end{align}
    where $\|\bomega_{\#-1}\|_0=\sum_{j=1}^{d/L} N_j^-(\bomega)$ represents the number of $\bomega_j=-1$, the first equality follows from the fact that the number of frequency vectors  $\{\bomega\}\in \mathfrak{C}(\Lambda)$ whose nonzero entries satisfy $\|\bomega\|_0=j$ is $\sum_{\|\bomega_{\#-1}\|_0=0}^{j} \tbinom{d}{j}\tbinom{j}{\|\bomega_{\#-1}\|_0}$, the last inequality employs the binomial inequality established in Lemma~\ref{append:lem:binomial_inequality_1}.
 
\smallskip
    \noindent \underline{\textit{The upper bound of $K(\bx,\bx)<r^2$.}} We now analyze the upper bound $r^2$ of the kernel function, namely
    \begin{align}\label{append:eq:ridge_reg_kernel_bound}
        K(\bx,\bx) & = \langle \bm{\Phi}_{\mathfrak{C}(\Lambda)}(\bx) , \bm{\Phi}_{\mathfrak{C}(\Lambda)}(\bx)\rangle 
        \nonumber \\ 
        & = \sum_{\bomega \in \mathfrak{C}(\Lambda)} \Phi_{\bomega}(\bx)^2 
        \nonumber \\
        & \le |\mathfrak{C}(\Lambda)|
        \nonumber \\
        & = \sum_{\bomega:\|\bomega\|_0=0}^{\Lambda}\sum_{\|\bomega_{\#-1}\|_0=0}^{\|\bomega\|_0} \binom{d}{\|\bomega\|_0}\binom{\|\bomega\|_0}{\|\bomega_{\#-1}\|_0} 
        \nonumber \\
        & = \sum_{\bomega:\|\bomega\|_0=0}^{\Lambda}  \binom{d}{\|\bomega\|_0}  2^{\|\bomega\|_0}
        \nonumber \\
        & \le \left(\frac{de^2}{\Lambda}\right)^{\Lambda}, 
    \end{align}
    where the second equality follows the definition of the feature map $\bm{\Phi}_{\mathfrak{C}(\Lambda)}$, the first inequality follows that $\Phi_{\bomega}(\bx)\le 1$, the last equality follows $\sum_{j=0}^i \tbinom{i}{j}=2^i$, and the last inequality employs Lemma~\ref{append:lem:binomial_inequality_1}.

\smallskip
    \noindent \underline{\textit{The upper bound of $|h(\bx)-y|<M$.}} For any $(\bm{x},y)\in \mathcal{X}\times \mathcal{Y}$ and $h\in \mathcal{H}_{\mathfrak{C}(\Lambda)}$, we have
    \begin{align}\label{append:eq:ridge_reg_hypothesis_bound}
        |h(\bx)-y| \le \left| \langle \bm{\Phi}_{\mathfrak{C}(\Lambda)} , \bm{\mathrm{w}} \rangle \right| + |y| \le \left\|\bm{\mathrm{w}} \right\|_2 \left\| \bm{\Phi}_{\mathfrak{C}(\Lambda)} \right\|_2 + B + \epsilon_l \le B \left(\frac{de^{q^2+1}}{\Lambda}\right)^{\Lambda} + B+\epsilon_l \le 2B \left(\frac{de^{q^2+1}}{\Lambda}\right)^{\Lambda} ,
    \end{align}
    where the first inequality follows the triangular inequality, the second inequality employs the Cauchy-Schwartz inequality and $|y|\le B+\epsilon_l$, and the last inequality employs the results in Eq.~\eqref{append:eq:ridge_reg_reg_para_bound} and Eq.~\eqref{append:eq:ridge_reg_kernel_bound}.

    In conjunction with Lemma~\ref{append:lem:ridge_reg_prediction_error} and the derived upper bound of the training error (Eq.~\eqref{append:eq:ridge_reg_training_error_bound}),  the term $\lambda$ (Eq.~\eqref{append:eq:ridge_reg_reg_para_bound}), the term $r^2$ (Eq.~\eqref{append:eq:ridge_reg_kernel_bound}), and the term $M$ (Eq.~\eqref{append:eq:ridge_reg_hypothesis_bound}), with probability at least $1-\delta$, we have
    \begin{align}\label{append:eq:ridge_predict_error_bound_proof}
	\mathbb{E}_{\bx\sim \mathcal{X}} \left| h_{\mathsf{qs}}(\bx) - \Tr(O\tilderho(\bx))  \right|^2 \leq \left(B\left(\frac{deq(1+R)}{\Lambda} \right)^{\Lambda} + \epsilon_l \right)^2 + 4B^2\left(\frac{de^{q^2+1}}{\Lambda}\right)^{2\Lambda}\sqrt{\frac{\log(1/\delta)}{n}}.
    \end{align}
    To make the upper bound nontrivial,   the noise parameters $q$ and the maximal value $R$ of the input $\bx$ should satisfy $q(1+R)<1/e$, and the frequency should be truncated to $\Lambda> deq(1+R)$. 
    
    To ensure the average prediction error is upper bounded by $\epsilon$, it is sufficient to showcase when the two terms are separately bounded by $\epsilon/4$ and $3\epsilon/4$. For the first term, employing the definition that $\epsilon=16B^2(deq(1+R)/\Lambda )^{2\Lambda}$ and the condition $\epsilon_l \le \sqrt{\epsilon}/4$, hence we have
    \begin{equation}
        \left(B\left(\frac{deq(1+R)}{\Lambda} \right)^{\Lambda} + \epsilon_l \right)^2 \le \frac{\epsilon}{4}.
    \end{equation}
    For the second term, we have
    \begin{equation}
        4B^2\left(\frac{de^{q^2+1}}{\Lambda}\right)^{2\Lambda}\sqrt{\frac{\log(1/\delta)}{n}}\le \frac{3}{4}\epsilon \Leftrightarrow n \ge B^4\left(\frac{de^{q^2+1}}{\Lambda}\right)^{4\Lambda} \frac{256\log(1/\delta)}{9\epsilon^2}.
    \end{equation}
    Employing the condition $\epsilon=16B^2(deq(1+R)/\Lambda )^{2\Lambda}$ could further simplify this lower bound, namely
    \begin{align}
        n \ge B^4\left(\frac{de^{q^2+1}}{\Lambda}\right)^{4\Lambda} \frac{256\log(1/\delta)}{9\epsilon^2} \ge \left(\frac{de}{\Lambda}\right)^{4\Lambda} \cdot \left(\frac{\Lambda}{deq(1+R)}\right)^{4\Lambda} \cdot \frac{\log(1/\delta)}{9}=\left(\frac{1}{q(1+R)}\right)^{4\Lambda}\cdot \frac{\log(1/\delta)}{9}.
    \end{align}
    Taken together, with probability at least $1-\delta$, the prediction error is upper bounded by $\epsilon$ when the number of training examples satisfies
    \begin{equation}
        n \ge \left(\frac{1}{q(1+R)}\right)^{4\Lambda}\cdot \frac{\log(1/\delta)}{9}.
    \end{equation}
\end{proof}

\subsection{Binomial inequality (Proof of Lemma~\ref{append:lem:binomial_inequality_1} and Lemma~\ref{append:lem:binomial_inequality_2})}
\begin{proof}
    [Proof of Lemma~\ref{append:lem:binomial_inequality_1}]
    We first observe that 
    \begin{equation}\label{append:eq:binomial_inequality_1_2}
        \left(\frac{a}{b}\right)^b \le \binom{a}{b} \le \frac{a^b}{b!},~~\mbox{and}~~\sum_{i=0}^{\infty} \frac{b^i}{i!}  = e^b,
    \end{equation}
    where the second equality follows the Taylor series expansion of the exponential function $e^b$.
    Hence we have
    \begin{equation}\label{append:eq:binomial_inequality_1_3}
        \sum_{i=b}^a \binom{a}{i} x^i \le \sum_{i=b}^a  \frac{a^i}{i!} x^i = \sum_{i=b}^a \frac{b^i}{i!} \cdot \left(\frac{ax}{b} \right)^i \le \left(\frac{ax}{b} \right)^b \sum_{i=b}^a \frac{b^i}{i!}  \le \left( \frac{eax}{b}\right)^b,
    \end{equation}
    where the second inequality employs the condition $x\le b/a$. The last inequality is obtained by adding terms in the sum (from $i = 0$ to $i = \infty$) and recognizing the exponential series. Moreover, for any $y>0$, we have
    \begin{equation}
        \sum_{i=0}^b \binom{a}{i} y^i \le \sum_{i=0}^a  \frac{a^i}{i!} y^i = \sum_{i=0}^b \frac{(by)^i}{i!} \cdot \left(\frac{a}{b} \right)^i \le \left(\frac{a}{b} \right)^b \sum_{i=0}^b \frac{(by)^i}{i!}  \le \left(\frac{a}{b} \right)^b e^{by} =\left(\frac{ae^{y} }{b} \right)^b, 
    \end{equation}
    where the second inequality employs $a/b>1$, and the last inequality is obtained by adding terms in the sum (from $i = 0$ to $i = \infty$) and recognizing the exponential series.
\end{proof}

\begin{proof}
    [Proof of Lemma~\ref{append:lem:binomial_inequality_2}]
    We first observe that for fixed $i$ the inner sum can be simplified to
    \begin{equation}
    \sum_{j=0}^i \binom{a}{i}\binom{i}{j} x^{j} = \binom{a}{i} \sum_{j=0}^i \binom{i}{j} x^{j} = \binom{a}{i} (1+x)^{i}.
    \end{equation}
    Therefore, employing Eq.~\eqref{append:eq:binomial_inequality_1} established in Lemma~\ref{append:lem:binomial_inequality_1} yields
    \begin{equation}
        \sum_{i=b}^{a}\sum_{j=0}^i \binom{a}{j}\binom{a-j}{i-j} x^{j}q^{i} = \sum_{i=b}^{a} \binom{a}{i} (1+x)^{i}q^{i} \le \left( \frac{eaq(1+x)}{b}\right)^b,
    \end{equation}
    where the last equality follows that $(1+x)q\le b/a$.
\end{proof}

\subsection{Worst-case truncation error for ridge regression models (Proof of Lemma~\ref{append:lem:ridge_training_error})}
\begin{proof}
    [Proof of Lemma~\ref{append:lem:ridge_training_error}]
    Denote the set of truncated frequencies as $\mathfrak{C}(\Lambda) =\{\bomega|\bomega \in \{0, \pm 1\}^d, ~s.t.~\|\bomega\|_0\leq \Lambda\}$. We consider the worst-case truncation error of the truncated function $f_{\mathfrak{C}(\Lambda)}(\bx)= \langle \bm{\Phi}_{\mathfrak{C}(\Lambda)}(\bx),  \tilde{\bm{\alpha}}_{\mathfrak{C}(\Lambda)} \rangle=\sum_{\|\bomega\|_0\le \Lambda}\Phi_{\bomega}(\bx) \cdot \tilde{\bm{\alpha}}_{\bomega}$, where $\bm{\Phi}_{\mathfrak{C}(\Lambda)}(\bx)=[\Phi_{\bomega}(\bx)]_{\bomega\in\mathfrak{C}(\Lambda)}$ and $\tilde{\bm{\alpha}}_{\mathfrak{C}(\Lambda)}=[\tilde{\bm{\alpha}}_{\bomega}]_{\bomega\in\mathfrak{C}(\Lambda)}$, namely 
    \begin{align}\label{append:eq:worst_case_train_error}
    \max_{\bx\in \mathcal{X}}\left|f_{\mathfrak{C}(\Lambda)}(\bx)-\Tr(\tilde{\rho}(\bx)O)\right| = &
        \max_{\bx\in \mathcal{X}}\left|\sum_{\bomega:\|\bomega\|_0>\Lambda}\Phi_{\bomega}(\bx)\tilde{\bm{\alpha}}_{\bomega}\right|
        \nonumber \\
        \le & \max_{\bx\in \mathcal{X}} \sum_{\bomega:\|\bomega\|_0>\Lambda}q^{\|\bomega\|_0}\left|\Phi_{\bomega}(\bx)\right| \left|\bm{\alpha}_{\bomega}\right| 
        \nonumber \\
        \le & \max_{\bx\in \mathcal{X}} \sum_{\bomega:\|\bomega\|_0>\Lambda}q^{\|\bomega\|_0}\cdot \left|\prod_{j=1}^{d/L}\sin^{\|\bomega_{\#-1}\|_0}(\bx_j)\cos^{\|\bomega\|_0-\|\bomega_{\#-1}\|_0}(\bx_j)\right| \cdot B
        \nonumber \\
        \le &  B \sum_{\bomega:\|\bomega\|_0>\Lambda}q^{\|\bomega\|_0} \cdot \prod_{j=1}^{d/L}R^{\|\bomega_{\#-1}\|_0}
        \nonumber \\
        = &  B \sum_{\bomega:\|\bomega\|_0=\Lambda+1}^{d} \sum_{\|\bomega_{\#-1}\|_0=0}^{\|\bomega\|_0}\binom{d}{\|\bomega\|_0}\binom{\|\bomega\|_0}{\|\bomega_{\#-1}\|_0}q^{\|\bomega\|_0}R^{\|\bomega_{\#-1}\|_0}
        \nonumber \\
        \le & B\left(\frac{edq(1+R)}{\Lambda+1} \right)^{\Lambda+1},
    \end{align}
    where $\|\bomega_{\#-1}\|_0=\sum_{j=1}^{d/L} N_j^-(\bomega)$ represents the number of $\bomega_j=-1$. The first equality follows the definition of $\Tr(\tilderho(\bx)O)=\sum_{\bomega:\|\bomega\|_0\ge0}\Phi_{\bomega}(\bx)\tilde{\bm{\alpha}}_{\bomega}$. The first inequality employs Lemma~\ref{append:lem:noisy_coefficient} about the relation between $\tilde{\bm{\alpha}}_{\bomega}$ and zero-noise coefficients $\bm{\alpha}_{\bomega}$ such that  $|\tilde{\bm{\alpha}}_{\bomega}| \le q^{\|\bomega\|_0}\bm{\alpha}_{\bomega}$. The second inequality employs $|\bm{\alpha}_{\bomega}|\le \|O\|_{\infty}\le B$ and the definition of $\Phi_{\bomega}(\bx)=\prod_{j=1}^{d/L} \sin^{N_j^-(\bomega)}(\bx_j)\cos^{N_j^+(\bomega)}(\bx_j)$ in Eq.~\eqref{append:eq:feature_map_para_share}. The third inequality follows that $|\sin(\bx_j)|\le |\bx_j| \le R$ and $|\cos(\bx_j)|\le 1$. The second equality follows from the fact that the number of frequency vectors  $\bomega\in \mathfrak{C}(\Lambda)$ satisfying $\|\bomega\|_0=a$  and $\|\bomega_{\#-1}\|_0=b$ is $\tbinom{d}{a}\tbinom{a}{b}$. The last inequality employs Lemma~\ref{append:lem:binomial_inequality_2} and the condition $q(1+R)\le (\Lambda+1)/d$.
\end{proof}

\section{Learnability for the proposed ridge regression model in a small range}\label{append:sec:ridge_samll_range}
In this section, we consider the efficiency of the predictive surrogate for the quantum circuits with correlated parameters, where the range of the parameters $\bx$ is restricted to a small region $[-R,R]^d$. In particular, we follow the same manner as discussed in SM~\ref{append:sec:implement_regression} to construct the predictive surrogate $h_{\mathsf{qs}}$ from the ridge regression model as defined in Eq.~\eqref{append:eq:ridge_reg_model_truncated} but with a different truncated feature map $\bm{\Phi}_{\mathfrak{C}(\Lambda)}(\bx)$. 

Recall that the truncated feature map $\bm{\Phi}_{\mathfrak{C}(\Lambda)}(\bx)$ considered in SM~\ref{append:sec:implement_regression} contains the trigonometric monomial $\Phi_{\bomega}(\bx)$ with the frequency index $\bomega$ satisfying $\|\bomega\|_0 \le \Lambda$, where the product terms $\sin(\bx_j)$ and $\cos(\bx_j)$ are treated equally in the truncation rule. For a small range of the input parameters $\bx_j \in [-R,R]$ with $R\to 0$, this truncation rule is not appropriate as the term $\cos(\bx_j)$ is much larger than $\sin(\bx_j)$ and approaches $1$ such that the truncation of the trigonometric monomial $\Phi_{\bomega}(\bx)$ containing $\cos(\bx_j)$ leads to a large truncation error. In this regard, we consider the truncated feature map 
\begin{equation}
\bm{\Phi}_{\mathfrak{S}(\Lambda)}(\bx)=[\Phi_{\bomega}(\bx)]_{\bomega\in\mathfrak{S}(\Lambda)},~ \mbox{with} ~\mathfrak{S}(\Lambda)=\{\bomega|\bomega\in\{0,\pm 1\},~s.t.~\|\bomega_{\#-1}\|_0\le \Lambda\},
\end{equation}
where $\|\bomega_{\#-1}\|_0$ refers to the number of $\bomega_j=-1$. This indicates that the truncation rule is only based on the number of $\sin(\bx_j)$ contained in $\Phi_{\bomega}(\bx)$. For this feature map, we construct the predictive surrogate $h_{\mathsf{qs}}=\braket{\bm{\Phi}_{\mathfrak{S}(\Lambda)}(\bx), \hat{\bm{\mathrm{w}}}}$ by solving the optimization problem defined in Eq.~\eqref{append:eq:ridge_reg_model_truncated}. This predictive surrogate could achieve an improved efficiency compared to that constructed on the feature map $\bm{\Phi}_{\mathfrak{C}(\Lambda)}(\bx)$, as encapsulated in the following theorem.

\begin{theorem}\label{append:thm:h_mb_small_range}
    Following the notation in Theorem~2, let the feature map be $\bm{\Phi}_{\mathfrak{S}(\Lambda)}$. Define $\epsilon=16B^2(deqR)/\Lambda)^{2\Lambda}$. When the following conditions are satisfied: (i) $\epsilon_l\le\sqrt{\epsilon}/4$, (ii) the frequency is truncated to $\Lambda>deqR$, (iii) the number of training examples satisfies \[n=\Big(\frac{1}{eqR}\Big)^{4\Lambda}\cdot \frac{\log(1/\delta)}{9},\] then the predictive surrogate $h_{\mathsf{qs}}(\bx)$ constructed by the ridge regression model in Eq.~\eqref{append:eq:ridge_reg_model_truncated} with the feature map $\bm{\Phi}_{\mathfrak{S}(\Lambda)}$, achieves with probability at least $1-\delta$,
    \begin{equation}\label{append:eq:ridge_predict_error_bound}
	\mathbb{E}_{\bx\sim \mathbb{D}} \left| h_{\mathsf{qs}}(\bx) - \Tr(O\tilderho(\bx))  \right|^2 \leq \epsilon.
\end{equation}
\end{theorem}

\begin{proof}
    [Proof of Theorem~\ref{append:thm:h_mb_small_range}] 
    We follow a similar strategy with the proof of Theorem~2 to derive the upper bound of the prediction error. Denote $f_{\mathfrak{S}(\Lambda)}(\bx)=\langle \bm{\Phi}_{\mathfrak{S}(\Lambda)}(\bx) , \tilde{\bm{\alpha}}_{\mathfrak{S}(\Lambda)} \rangle=\sum_{\bomega \in \mathfrak{S}(\Lambda)}\Phi_{\bomega}(\bx) \cdot \tilde{\bm{\alpha}}_{\bomega}$ as the truncated target function regarding the feature map $\bm{\Phi}_{\mathfrak{S}(\Lambda)}(\bx)$, $K_{\mathfrak{S}(\Lambda)}(\bx,\bx')=\langle \bm{\Phi}_{\mathfrak{S}(\Lambda)}(\bx),  \bm{\Phi}_{\mathfrak{S}(\Lambda)}(\bx')\rangle$ as the kernel induced by $\bm{\Phi}_{\mathfrak{S}(\Lambda)}(\bx)$, and \[\mathcal{H}_{\mathfrak{S}(\Lambda)}=\{h(\bx)=\braket{\bm{\Phi}_{\mathfrak{S}(\Lambda)}(\bx), \bm{\mathrm{w}}}|\bm{\mathrm{w}}\in \mathbb{R}^{|\mathfrak{S}(\Lambda)|}\}\] as the hypothesis class. 
    
    Recall that deriving the upper bound of prediction error for surrogate $h_{\mathsf{qs}}$ based on the feature map $\bm{\Phi}_{\mathfrak{S}(\Lambda)}(\bx)$ could be reduced to separately derive the upper bound of training error $\sum_{i=1}^n (h_{\mathsf{qs}}(\bxi)-\yi)^2/n$, regularization parameter $\|\tilde{\bm{\alpha}}_{\mathfrak{S}(\Lambda)}\|_2\le \lambda$, kernel function $K_{\mathfrak{S}(\Lambda)}(\bx,\bx)\le r^2$, and $|h(\bx)-y|\le M$ for any $h\in \mathcal{H}_{\mathfrak{S}(\Lambda)}$.

\smallskip
    \noindent \underline{\textit{The upper bound of training error $\sum_{i=1}^n (h_{\mathsf{qs}}(\bxi)-\yi)^2/n$.}} Let $i^*$ be the index of training data $(\bx^{(i*)}, y^{(i*)})$ with the maximal error for $f_{\mathfrak{S}(\Lambda)}(\bxi)$, i.e., $
        i^* = \arg\min_{i\in[n]} (f_{\mathfrak{S}(\Lambda)}(\bxi)-\yi)^2$. Following the derivation of Eq.~\eqref{append:eq:ridge_reg_training_error_bound}, we have
    \begin{align}\label{append:eq:smallR_ridge_reg_training_error_bound}
    \frac{1}{n}\sum_{i=1}^n \left(h_{\mathsf{qs}}(\bxi)-\yi\right)^2 \le
        & \left(\left|f_{\mathfrak{S}(\Lambda)}(\bx^{(i^*)})-\Tr(\tilde{\rho}(\bx^{(i^*)})O)\right| +\left|\Tr(\tilde{\rho}(\bx^{(i^*)})O)-y^{(i^*)}\right|\right)^2
        \nonumber \\
        \le   & \left(B\left(\frac{edqR}{\Lambda+1} \right)^{\Lambda+1}+\epsilon_l\right)^2,
    \end{align}
    where the second term in the second inequality employs the condition that estimation error on the training data is bounded by $\epsilon_l$, and the first term in the second inequality follows the derivation in Eq.~\eqref{append:eq:worst_case_train_error}, namely,
      \begin{align}
    \left|f_{\mathfrak{S}(\Lambda)}(\bx^{(i^*)})-\Tr(\tilde{\rho}(\bx^{(i^*)})O)\right| = &
        \max_{\bx\in \mathcal{X}}\left|\sum_{\bomega:\|\bomega_{\#-1}\|_0>\Lambda}\Phi_{\bomega}(\bx)\tilde{\bm{\alpha}}_{\bomega}\right|
        \nonumber \\
        \le &  \sum_{\bomega:\|\bomega_{\#-1}\|_0>\Lambda}q^{\|\bomega_{\#-1}\|_0} \cdot R^{\|\bomega_{\#-1}\|_0} \cdot B
        \nonumber \\
        = &  B \sum_{\bomega:\|\bomega_{\#-1}\|_0=\Lambda+1}^{d} \binom{d}{\|\bomega_{\#-1}\|_0}(qR)^{\|\bomega_{\#-1}\|_0}
        \nonumber \\
        \le & B\left(\frac{edqR}{\Lambda+1} \right)^{\Lambda+1},
    \end{align}
    where the second equality follows follows from the fact that the number of frequency vectors  $\bomega\in \mathfrak{S}(\Lambda)$ satisfying $\|\bomega_{\#-1}\|_0=a$ is $\tbinom{d}{a}$, the last inequality employs Lemma~\ref{append:lem:binomial_inequality_1}.

\smallskip
    \noindent \underline{\textit{The upper bound of regularization parameter $\|\tilde{\bm{\alpha}}_{\mathfrak{S}(\Lambda)}\|_2\le \lambda$.}} 
    Following the derivation of Eq.~\eqref{append:eq:ridge_reg_reg_para_bound}, we have
    \begin{align}\label{append:eq:smallR_ridge_reg_reg_para_bound}
        \|\tilde{\bm{\alpha}}_{\mathfrak{S}(\Lambda)}\|_2^2    \le  \sum_{\bomega\in\mathfrak{S}(\Lambda)}q^{2\|\bomega\|_0} B^2
        = B^2\sum_{\bomega:\|\bomega_{\#-1}\|_0=0}^{\Lambda} \binom{d}{\|\bomega_{\#-1}\|_0} q^{2\|\bomega_{\#-1}\|_0}
         \le B^2\left(\frac{de^{q^2}}{\Lambda}\right)^{\Lambda }
    \end{align}
    where the first equality follows that the number of frequency vectors  $\{\bomega\}\in \mathfrak{S}(\Lambda)$ satisfying $\|\bomega_{\#-1}\|_0=a$ is $\tbinom{d}{a}$, the last inequality employs Lemma~\ref{append:lem:binomial_inequality_1}.

\smallskip 
    \noindent \underline{\textit{The upper bound of $K_{\mathfrak{S}(\Lambda)}(\bx,\bx)<r^2$.}} Following the derivation of Eq.~\eqref{append:eq:ridge_reg_kernel_bound}, we have
    \begin{equation}\label{append:eq:smallR_ridge_reg_kernel_bound}
        K_{\mathfrak{S}(\Lambda)}(\bx,\bx)  = \sum_{\bomega \in \mathfrak{S}(\Lambda)} \Phi_{\bomega}(\bx)^2  = \sum_{\bomega:\|\bomega_{\#-1}\|_0=0}^{\Lambda} \binom{d}{\|\bomega_{\#-1}\|_0} \Phi_{\bomega}(\bx)^2.
    \end{equation}
    Employing the fact of $\sin(\bx_j)\le \bx_j\le R$ and $\cos(\bx_j)\le 1$ for a small range of $\bx_j\in [-R,R]$, we have
    \begin{align}\label{append:eq:smallR_ridge_reg_kernel_bound_2}
        K_{\mathfrak{S}(\Lambda)}(\bx,\bx)
        & \le \sum_{\bomega:\|\bomega_{\#-1}\|_0=0}^{\Lambda}  \binom{d}{\|\bomega_{\#-1}\|_0}  R^{2\|\bomega_{\#-1}\|_0}
        \nonumber \\
        & \le \left(\frac{de^{R^2}}{\Lambda}\right)^{\Lambda}, 
    \end{align}
    where the first inequality follows $\Phi_{\bomega}(\bx)\le R^{\|\bomega_{\#-1}\|_0}$, and the last inequality employs Lemma~\ref{append:lem:binomial_inequality_1}.

\smallskip
    \noindent \underline{\textit{The upper bound of $|h(\bx)-y|<M$.}} Following the derivation of Eq.~\eqref{append:eq:ridge_reg_hypothesis_bound}, for any $h\in \mathcal{H}_{\mathfrak{S}(\Lambda)}$, we have
    \begin{align}\label{append:eq:smallR_ridge_reg_hypothesis_bound}
        |h(\bx)-y| \le \left| \Phi_{\mathfrak{S}(\Lambda)} \cdot \bm{\mathrm{w}} \right| + |y| \le \left\|\bm{\mathrm{w}}\right\|_2 \left\| \Phi_{\mathfrak{C}(\Lambda)} \right\|_2 + B+\epsilon_l \le 2B \left(\frac{de^{(q^2+R^2)/2}}{\Lambda}\right)^{\Lambda} ,
    \end{align}
    where the last inequality employs the results in Eq.~\eqref{append:eq:smallR_ridge_reg_reg_para_bound} and Eq.~\eqref{append:eq:smallR_ridge_reg_kernel_bound}. 
    
    In conjunction with Lemma~\ref{append:lem:ridge_reg_prediction_error} and the derived upper bound of the training error (Eq.~\eqref{append:eq:smallR_ridge_reg_training_error_bound}),  the term $\lambda$ (Eq.~\eqref{append:eq:smallR_ridge_reg_reg_para_bound}), the term $r^2$ (Eq.~\eqref{append:eq:smallR_ridge_reg_kernel_bound}), and the term $M$ (Eq.~\eqref{append:eq:smallR_ridge_reg_hypothesis_bound}), with probability at least $1-\delta$, we have
    \begin{align}\label{append:eq:ridge_predict_error_bound_proof}
	\mathbb{E}_{\bx\sim \mathcal{X}} \left| h_{\mathsf{qs}}(\bx) - \Tr(O\tilderho(\bx))  \right|^2 \leq \left(B\left(\frac{deqR}{\Lambda} \right)^{\Lambda} + \epsilon_l \right)^2 + 4B^2\left(\frac{d^2e^{q^2+R^2}}{\Lambda^2}\right)^{\Lambda}\sqrt{\frac{\log(1/\delta)}{n}}.
    \end{align}
    To make the upper bound nontrivial,  the noise parameters $q$ and the maximal value $R$ of the input $\bx$ should satisfy $qR<1/e$, and the frequency should be truncated to $\Lambda> deqR$. 
    
    To ensure the average prediction error is upper bounded by $\epsilon$, it is sufficient to showcase when the two terms are separately bounded by $\epsilon/4$ and $3\epsilon/4$. For the first term, employing the definition that $\epsilon=16B^2(deqR/\Lambda )^{2\Lambda}$ and the condition $\epsilon_l \le \sqrt{\epsilon}/4$, hence we have
    \begin{equation}
        \left(B\left(\frac{deqR}{\Lambda} \right)^{\Lambda} + \epsilon_l \right)^2 \le \frac{\epsilon}{4}.
    \end{equation}
    For the second term, we have
    \begin{equation}
        4B^2\left(\frac{d^2e^{q^2+R^2}}{\Lambda^2}\right)^{\Lambda^2}\sqrt{\frac{\log(1/\delta)}{n}}\le \frac{3}{4}\epsilon \Leftrightarrow n \ge B^4\left(\frac{d^2e^{q^2+R^2}}{\Lambda^2}\right)^{2\Lambda} \frac{256\log(1/\delta)}{9\epsilon^2}.
    \end{equation}
    Employing the condition $\epsilon=16B^2(deqR/\Lambda )^{2\Lambda}$ could further simplify this lower bound, namely
    \begin{align}
        n \ge B^4\left(\frac{d^2e^{q^2+R^2}}{\Lambda^2}\right)^{2\Lambda} \frac{256\log(1/\delta)}{9\epsilon^2} \ge \left(\frac{de}{\Lambda}\right)^{4\Lambda} \cdot \left(\frac{\Lambda}{deqR}\right)^{4\Lambda} \cdot \frac{\log(1/\delta)}{9}=\left(\frac{1}{qR}\right)^{4\Lambda}\cdot \frac{\log(1/\delta)}{9}.
    \end{align}
    Taken together, with probability at least $1-\delta$, the prediction error is upper bounded by $\epsilon$ when the number of training examples satisfies
    \begin{equation}
        n \ge \left(\frac{1}{eqR}\right)^{4\Lambda}\cdot \frac{\log(1/\delta)}{9}.
    \end{equation}
\end{proof}

\section{Computational time for training and prediction}\label{append:sec:complexity-analysis}
In this section, we analyze the computational efficiency of the proposed predictive surrogates, namely $h_{\mathsf{cs}}$ for classical shadow and $h_{\mathsf{qs}}$ for quantum simulation, as described in Eq.~\eqref{eqn:generic-learner} and Eq.~\eqref{eqn:surrogate_correlated_para}, respectively. Recall that our proposal consists of two stages: training dataset collection and model inference. As a result, to demonstrate the efficiency of our proposal, it is equivalent to exhibiting the computational efficiency at each stage. In the following, we separately analyze the computational cost of $h_{\mathsf{cs}}$  and $h_{\mathsf{qs}}$ in SM~\ref{appendix:subsec:time_cs} and SM~\ref{appendix:subsec:time_mb}.

\subsection{Computational time of $h_{\mathsf{cs}}$}\label{appendix:subsec:time_cs}

The analysis of the computational time of $h_{\mathsf{cs}}$ is similar to the noiseless setting \cite{du2024efficient}. This is because their training and prediction procedures are completely identical. For self-consistency, here we provide a brief review of the analysis provided in Ref.~\cite{du2024efficient}.

\noindent\underline{\textit{Training time}}. Recall that the training procedure of $h_{\mathsf{cs}}$ amounts to loading the collected training dataset $\mathcal{T}_{\mathsf{s}}=\{\bxi, \tilde{\rho}_T(\bxi)\}_{i=1}^n$ to the classical memory, where the required computation cost to store and load the $N$-qubit classical shadow representation $\tilde{\rho}_T(\bxi)$ with $T$ snapshots is $\mathcal{O}(NT)$. Simultaneously, the computation cost to store and load the classical input $\bxi$ is $\mathcal{O}(d)$. Combining these facts with the result of Theorem~1, where the required total number of training examples $n$ of $h_{\mathsf{cs}}$, the computation cost to load the dataset $\mathcal{T}_{\mathsf{s}}$ is 
 \begin{eqnarray}
 	\mathcal{O}(nNT)= \mathcal{O}\left(  \frac{B^2 9^KNT|\mathfrak{C}(\Lambda)|}{\epsilon}   \log \left(\frac{|\mathfrak{C}(\Lambda)|}{\delta}\right)\right) = \widetilde{\mathcal{O}}\left(\frac{ NTB^2 9^K}{\epsilon} \left|\mathfrak{C}\left(\min \left\{\frac{4C}{\epsilon}, \frac{1}{2(p+\PZ )} \log\left(\frac{2B}{\sqrt{\epsilon}}\right)\right\}\right)\right|    \right).
 \end{eqnarray}
 Analogous to the noiseless case, when $C$ is bounded (i.e., $C\sim \mathcal{O}(1/\text{poly}(N))$ or $C\sim \mathcal{O}(1/\text{exp}(N))$) or $d$ is a small constant, the cardinality $|\mathfrak{C}(4C/\epsilon)|$ polynomially scales with $N$ and $d$, and thus the predictive surrogate $h_{\mathsf{cs}}$ is computationally efficient to complete the training stage. Moreover, higher levels of Pauli noise (i.e., a smaller $1/2(p+p_Z)$) also lead to the efficient training of $h_{\mathsf{cs}}$.
 
 \smallskip
  \noindent\underline{\textit{Inference (prediction) time}}. Suppose that the observable $O$ is constituted by multiple local observables with a bounded norm, i.e., $O=\sum_{i=1}^q O_i$ and $\sum_l \|O_i\|\leq B$, and the maximum locality of $\{O_i\}$ is $K$. Following the definition in Eq.~(\ref{append:eqn:TriGeo-non-trunc-form}), given a new input $\bx$ and the observable $O$, the prediction yields 
  \begin{equation}
  h_{\mathsf{cs}} \equiv\Tr(O \hatsigma(\bx))=\frac{1}{n}\sum_{i=1}^n\hatkappa\left(\bx, \bxi\right)	\Tr\left(O \tilderho_T(\bxi)\right).
  	\end{equation}  
  In this regard, the evaluation involves summing over the assessment of each training example  $(\bxi, \tilderho_T(\bxi))$ for $\forall i \in [n]$, and the evaluation of each training example can be further decomposed into two components. That is, the first component is classically computing the shadow estimation $\Tr(O\tilde{\rho}_T(\bxi))$; the second component is calculating the kernel function $\kappa_{\Lambda}(\bx, \bxi)$ for $\forall i \in [n]$. The computation of each $\Tr(O_i\tilde{\rho}_T(\bxi))$ for $K\sim \mathcal{O}(1)$ can be completed in $\mathcal{O}(T)$ time after storing the classical shadow in the classical memory. Therefore, the computation cost to complete the first part is $\mathcal{O}(Tq)$. Moreover, the computation cost of evaluating the kernel $\kappa_{\Lambda}(\bx, \bxi)$ is determined by the dimension of the truncated feature map, namely $\mathcal{O}(|\mathfrak{C}(4C/\epsilon)|)$.
 In conjunction with the computation cost of each example and the total number of training examples $n$, the required predicting time for $h_{\mathsf{cs}}$ is 
    \begin{equation}
  	\mathcal{O}\Big(n(Tq+|\mathfrak{C}(4C/\epsilon)|)\Big) \leq \widetilde{\mathcal{O}}\left( \frac{TqB^2 9^K}{\epsilon} \left|\mathfrak{C}\left(\min \left\{\frac{4C}{\epsilon}, \frac{1}{2(p+\PZ )} \log\left(\frac{2B}{\sqrt{\epsilon}}\right)\right\}\right)\right|^2  \right) \le \mathcal{O}(n^2Tq).
  \end{equation}
    
\smallskip
  \noindent\underline{\textit{Total computation cost.}}  When $C$ is bounded or a large noise level, the whole procedure of our proposal (encompassing the training and predicting) is both computation and memory efficient, which is upper bounded by 
\begin{equation}
\widetilde{\mathcal{O}}\left( \frac{TNqB^2 9^K |\mathfrak{C}(4C/\epsilon)|^2}{\epsilon} \right)\le \widetilde{\mathcal{O}}(n^2NTq).
\end{equation}

\subsection{Computational time of $h_{\mathsf{qs}}$}\label{appendix:subsec:time_mb}
\noindent\underline{\textit{Training time}}. Recall that the predictive surrogate for quantum simulation $h_{\mathsf{qs}}=\braket{\bm{\Phi}_{\mathfrak{C}(\Lambda)}(\bx), \widehat{\bm{\mathrm{w}}}}$ is constructed by solving the optimization problem
\begin{equation}\label{append:eq:ridge}
    \min_{\bm{\mathrm{w}}} \frac{1}{n}\sum_{i=1}^n \left(\yi- h_{\mathsf{qs}}(\bxi,\bm{\mathrm{w}})\right)^2 + \lambda \|\bm{\mathrm{w}}\|_2,
\end{equation}
where the training dataset is $\mathcal{T}_{\mathsf{mb}}=\{(\bxi,\yi)\}_{i=1}^n$, and the response $\yi$ is the estimated mean value of the observable $O$ with the maximal estimation error  $\max_{i\in [n]} |\yi-\Tr(\tilderho(\bxi)O)| \le \epsilon_l$. In this regard, the computational cost during the training stage is dominated by the time required for
kernel ridge regression over the feature space defined by the feature vector $\bm{\Phi}_{\mathfrak{C}(\Lambda)}(\bx)$. In particular, one can show that the optimization problem in Eq.~\eqref{append:eq:ridge} is a convex problem so that we can solve its equivalent dual problem instead, which is given by
\begin{equation}
    \max_{\bm{\beta}\in \mathbb{R}^n} - \bm{\beta}^{\top} (K+\lambda I) \bm{\beta} + 2 \bm{\beta}\cdot \bm{Y},
\end{equation}
where $K=\bm{\Phi}_{\mathfrak{C}(\Lambda)}(\bx)^{\top} \bm{\Phi}_{\mathfrak{C}(\Lambda)}(\bx)$ is the kernel matrix and $\bm{Y}=(y^{(1)}, \cdots, y^{(n)})$ refers to the response vector. The time complexity for the kernel-based model has been well established in the statistical learning theory \cite{mehyar2018foundations}. In particular,
assume the time complexity for computing the kernel entry $K(\bx,\bx')=\braket{\bm{\Phi}_{\mathfrak{C}(\Lambda)}(\bx), \bm{\Phi}_{\mathfrak{C}(\Lambda)}(\bx')}$ is $\kappa$, then the complexity for solving this dual problem is $\mathcal{O}(\kappa n^2+n^3)$ and the time complexity for prediction is $\mathcal{O}(\kappa n)$. 

In our case, the time complexity for computing $K(\bx,\bx')$ is determined by the feature dimension of $\bm{\Phi}_{\mathfrak{C}(\Lambda)}(\bx)$, namely $\kappa =\mathfrak{C}(\Lambda)$. Combining these facts with the result of Theorem~2 such that the required total number of training examples $n$ of our model, the computation cost for solving the ridge regression model in Eq.~\eqref{append:eq:ridge} is 
\begin{equation}
    \mathcal{O}(|\mathfrak{C}(\Lambda)| \cdot n^2+n^3) = \mathcal{O} \left(|\mathfrak{C}(\Lambda)| \cdot  \left(\frac{1}{q(1+R)}\right)^{8\Lambda}\right).
\end{equation}

\noindent\underline{\textit{Inference (prediction) time}}.   Similarly, as explained above for the training time, the time complexity for prediction in the regression model is $\mathcal{O}(\kappa n)$, where $\kappa$ represents the time complexity for computing the kernel entry $K(\bx, \bx')$, which is determined by the cardinality $|\mathfrak{C}(\Lambda)|$. As a result, the prediction time is given by
\begin{equation}
    \mathcal{O}(|\mathfrak{C}(\Lambda)| \cdot n) = \mathcal{O}\left(|\mathfrak{C}(\Lambda)| \cdot \left(\frac{1}{q(1+R)}\right)^{4\Lambda}\right).
\end{equation}

\section{More implementation details of the proposed predictive surrogates}

In this section, we provide a detailed explanation of how to implement the proposed predictive surrogates, $h_{\mathsf{cs}}$ and $h_{\mathsf{qs}}$, in two practical tasks: pre-training variational quantum eigensolvers and Hamiltonian simulation for identifying Floquet symmetry-protected topological phases.

\subsection{Implementation details in pre-training VQEs with $h_{\textsf{cs}}$}\label{apd:hyper}
In this section, we first review the workflow of VQEs and highlight the extensive measurement costs when implementing the VQEs entirely on quantum processors. We then elucidate how to exploit the predictive surrogate $h_{\textsf{cs}}$ to pre-train the VQEs.

\subsubsection{Expensive measurement overheads of variational quantum Eigensovers}
 Variational quantum algorithms  (VQA)~\cite{bharti2021noisy,cerezo2021variational}  are a class of hybrid algorithms that combine variational ans\"atze with classical optimization. Thanks to their flexibility in adapting to hardware constraints, e.g.,  restrictive qubit connectivity, limited circuit depth, and the presence of noise, VQAs have emerged as leading candidates for pursuing the practical utility of near-term quantum computers. Two typical prototypes of VQAs are quantum neural networks~\cite{mitarai2018quantum,havlicek2018supervised}  and variational quantum Eigensolvers (VQEs)~\cite{peruzzo2014variational,tilly2022variational}. Over the past years, a plethora of theoretical and experimental studies have explored the potential of VQAs in solving real-world tasks, including image classification \cite{huang2021power,liu2021rigorous,du2023problem,wang2023transition}, solving combinatorial optimization problems \cite{farhi2014quantum,zhou2020quantum}, and ground state energy estimation~\cite{kandala2017hardware,kim2023evidence}.

Let us briefly recap the mechanism of VQEs, which constitutes one of the main focuses of this work. Without loss of generality, let $U(\bx)$ in Eq.~(\ref{eqn:append:circuit-ideal}) be the employed $N$-qubit ansatz. Given an initial state $\rho_0$ and the explored Hamiltonian $\mathsf{H}$, the expectation value is $f(\rho(\bx), \mathsf{H})=\Tr(\rho(\bx)\mathsf{H})$, as defined in Eq.~(\ref{eqn:append:expectation-idea}). The objective of VQE is to find the ground state energy of $\mathsf{H}$. To achieve this, a common strategy is to use a gradient descent optimizer to iteratively update the tunable parameters $\bx$ to minimize $f(\rho(\bx), \mathrm{H})$. Specifically, at the $t$-th iteration, the updating rule yields
\begin{equation}
	\bx^{(t+1)} = \bx^{(t)} - \eta \nabla_{\bx} f(\rho(\bx^{(t)}), \mathsf{H}),
\end{equation}
where $\eta$ refers to the learning rate. This process continues until the loss function converges to a predefined condition.

However, implementing VQEs on quantum processors incurs significant measurement overhead, primarily due to the need for estimating expectation values and gradients during the optimization process.

\smallskip
\noindent\underline{\textit{Expectation value estimation}}.
The core of any optimization strategy requires estimating the expectation value of the Hamiltonian $\mathsf{H}$ on the quantum state $\rho(\bx)$. Consider the representation of $\mathsf{H}$ in the Pauli basis as  $\mathsf{H}=\sum_{j=1}^q a_j P_j$. A common approach to estimating $\Tr(\rho(\bx)\mathrm{H})$ is to separately obtain the expectation value estimation of the individual Pauli terms, $\Tr(\rho(\bx)P_j)$. As a result, the measurement cost scales linearly with both the number of Pauli terms $q$ and the inverse of the desired precision $\varepsilon$. For the Hamiltonian studied in quantum chemistry \cite{cao2019quantum}, the quantity of $q$ typically scales as $\mathcal{O}(N^4)$, resulting in a measurement cost that follows the same scaling.

\smallskip
\noindent\underline{\textit{Gradient estimation}}. Beyond the measurement cost required for estimating $f(\rho(\bx), \mathsf{H})$,  the evaluation of gradients $\nabla f(\rho(\bx, \mathsf{H})$ likewise encounters a considerable measurement cost. Due to the no-cloning theorem and the non-commutability of trainable gates in $U(\bx)$, obtaining gradients often requires an intensive number of quantum measurements~\cite{abbas2023quantum}. For instance, when the parameter shift rule \cite{schuld2019evaluating} is adopted, the total number of measurements at each iteration linearly scales with the dimension of $\bx$, and the number of snapshots used to estimate $\Tr(\rho(\bx)\mathsf{H})$. In addition, the measurement cost also scales linearly with the total number of iteration steps. It has been shown that an under-parameterized quantum circuit requires exponentially many iterations to reach convergence \cite{you2022convergence,liu2023analytic}. 

More importantly, the measurement requirement of VQEs can increase substantially when the gradient becomes small, a phenomenon that often arises when the quantum circuit encounters barren plateaus  \cite{mcclean2018barren, larocca2025barren} or when the parameter approaches the convergence point during optimization~\cite{du2021learnability}. In the case of a barren plateau, the gradients are concentrated to zero, which results in an exponential increase in the measurement cost needed to obtain informative gradients in each iteration. 

Taken together, the optimization process of VQEs requires a significant number of quantum measurements, a challenge that could potentially be alleviated by constructing a predictive surrogate.

\subsubsection{Workflow of pre-training VQEs via the predictive surrogate}
The employment of the predictive surrogate $h_{\mathsf{cs}}$ to pre-training VQEs involves two main steps: (i) Pre-training VQE with $h_{\mathsf{cs}}$ on classical processors, and (ii) optional fine-tuning with $f(\tilderho(\bx),\mathsf{H})$ on the quantum processors.  In the following, we describe the detailed training process for each step, assuming the predictive surrogate $h_{\mathsf{cs}}$ is properly constructed, and discuss the reduction in measurements achieved by pre-training VQEs compared to the original VQEs.

\smallskip

\noindent \underline{\textit{Step I: Pre-training VQE with $h_{\mathsf{cs}}$.}} Recall that the aim of VQEs is to minimize $f(\tilderho(\bx)\mathsf{H})$ to find the optimal parameter $\bx^*$. Given access to the \textit{optimized} predictive surrogate $h_{\mathsf{cs}}$, the target parameter $\bx^*$ can be approximated by minimizing $h_{\mathsf{cs}}(\bx,\mathsf{H})$. The optimized parameter obtained from this step is denoted by $\hat{\bx}$, which satisfies \[\hat{\bx}=\arg\min_{\bx} h_{\mathsf{cs}}(\bx,\mathsf{H}).\]  Notably, this step is performed exclusively on the classical side. As such, the optimization process can be facilitated by automatic differentiation combined with an efficient classical optimizer, distinguishing it from conventional VQEs that rely on the parameter-shift rule to perform optimization directly on quantum processors.

\smallskip

\noindent \underline{\textit{Step II: Additional fine-tuning with $f(\tilderho(\bx),\mathsf{H})$.}} Given the access to the optimized parameter $\hat{\bx}$, the quantum circuit $U(\hat{\bx})$ can be further \textit{optionally} trained on the employed quantum processor by minimizing $f(\tilderho(\hat{\bx}),\mathsf{H})$. Specifically, when the estimated expectation value $f(\tilderho(\hat{\bx}),\mathsf{H})$ obtained from the optimized parameter $\hat{\bx}$ fails to achieve a satisfactory precision relative to the target value $f(\tilderho(\bx^*),\mathsf{H})$, this additional fine-tuning process is necessary to improve the solution further. During the fine-tuning stage, the parameters are updated using the parameter-shift rule.

\smallskip

\noindent \underline{\textit{Reduction in measurement achieved by pre-training VQE.}} The total number of measurements of the pre-training VQE consists of two components: the collection of the training dataset for constructing the surrogate $h_{\mathsf{cs}}$, and the additional fine-tuning on quantum processors. The offline optimization enabled by the predictive surrogate in the pre-training phase significantly reduces the demand for quantum resources. This is particularly evident when the quantum circuit involves a small number of rotation gates or when solving VQEs for a family of problem Hamiltonians. Specifically, quantum circuits with fewer parameters often exhibit an under-parameterized loss landscape, requiring an exponential number of iterations to achieve convergence. As a result, optimization on quantum processors incurs an exponential measurement cost. However, this issue can be mitigated by optimizing the surrogate classically. Additionally, the predictive surrogate, constructed on the classical shadow, can calculate the gradients of a family of Hamiltonians simultaneously, leading to a substantial reduction in measurement costs when using quantum circuits for VQEs involving an increased number of Hamiltonians. In conclusion, these measurement reductions provide critical advantages given the scarcity of available quantum processors in the current landscape.

\subsection{Implementation details in  identifying FSPT phases with $h_{\mathsf{qs}}$} \label{appendix:sec-F:subsec:FSPT}

Here we first recap the task of identifying Floquet symmetry-protected topological (FSPT) phases through Hamiltonian simulation. After that, we introduce how to employ the predictive surrogate to enhance the identification of FSPT phases.

\subsubsection{Quantum simulation of Floquet
symmetry-protected topological phases}\label{append:subsec:background_FSPT}
Understanding and categorizing the phases of quantum many-body systems is essential for gaining insights into the fundamental nature of matter at a microscopic level. A quantum many-body system is governed by a Hamiltonian, $\mathsf{H}(\bx)$, where $\bx$ represents external physical parameters that influence the system’s properties. To study the phases of these systems at different values of $\bx$, various Hamiltonian simulation algorithms have been developed. These algorithms simulate the time evolution of the quantum system to reach a specific state that contains information about its phases. For equilibrium quantum many-body systems, the phase is entirely determined by the ground state of a time-independent Hamiltonian $\mathsf{H}(\bx)$. Numerous quantum simulation algorithms \cite{latorre2004adiabatic,cao2025unveiling} and learning-based algorithms \cite{van2017learning,carrasquilla2017machine,ch2017machine,huang2022provably} have been proposed to characterize the phase of $\mathsf{H}(\bx)$. In contrast, characterizing the \textit{non-equilibrium phases} of a time-dependent Hamiltonian $\mathsf{H}(\bx,t)$ presents a more challenging task due to the absence of well-defined steady states and the difficulty of describing the system's properties in the absence of equilibrium conditions.

One particularly intriguing type of non-equilibrium phase is the Floquet symmetry-protected topological (FSPT) phase, a unique phase of matter that only exists out of equilibrium. We follow the conventions of Ref.~\cite{zhang_digital_2022}
to recap the approach of identifying the FSPT phase and the phase transition point with digital Hamiltonian simulation. Recall the 1D spin-1/2 chain governed by the time-periodic Hamiltonian defined in the main text, i.e.,
\begin{equation}
	\mathsf{H}_{\mathsf{tp}}(\mathrm{t}) = \begin{cases}
	\mathsf{H}_{\mathsf{tp},1} = (\frac{\pi}{2}-\delta)\sum_{i}X_i, & 2k_1\mathrm{T}_1 \leq \mathrm{t} < (2k_1+1)\mathrm{T}_1, \\
	\mathsf{H}_{\mathsf{tp},2} = -\sum_i J_iZ_{i-1}X_i Z_{i+1}, & (2k_1+1)\mathrm{T}_1 \leq t < (2k_1+2)\mathrm{T}_1,
			\end{cases} \nonumber
\end{equation}
where $P_i\in \{X, Y, Z\}$ refers to the Pauli operator acting on the $i$-th qubit, $2\mathrm{T}_1$ refer to a driven period, $k_1\in \mathbb{N}$ is an arbitrary non-negative integer, $\delta$ denotes the driven perturbation.

For clarity, we denote $\bm{J}=(J_2,\cdots,J_{N-1})$ as the vector of coupling parameters, and $\bx=(\delta,\bm{J})$ as the parameter vector of  $\mathsf{H}_{\mathsf{tp}}(\mathrm{t})$. Through varying the drive perturbation $\delta$, two dynamical phases emerge: the FSPT phase and thermal phase. The FSPT phase is fully characterized by the Floquet unitary $U_{F}=U_2(\bm{J})U_1(\delta)$, where $U_1(\delta)=\exp(-\imath \mathsf{H}_{\textsf{tp,1}})$ and $U_2(\bm{J})=\exp(-\imath \mathsf{H}_{\textsf{tp,2}})$ are the unitary operators generated by the Hamiltonians $\mathsf{H}_{\textsf{tp,1}}$ and $\mathsf{H}_{\textsf{tp,2}}$, respectively. Here, the Floquet unitary $U_F$ describes the evolution of $\mathsf{H}_{\mathsf{tp}}(\mathrm{t})$ over one full period $2\mathrm{T}_1$.   One could use the idea of digital quantum simulation to evolve $\mathsf{H}_{\mathsf{tp}}(\mathrm{t})$ using the quantum circuit $U(\bx)$. More precisely, for the evolution time $t=k\mathrm{T}_1$ of $k/2$ periods, the quantum circuit has the form of \[U(\bx)=(U_2(\bm{J})U_1(\delta))^{t/2\mathrm{T}_1},\] where the detailed implementations of $U_1(\delta)$ and $U_2(\bm{J})$ are deferred to SM~\ref{apd:FSPT_circuit}. As such, the evolution time $t$ determines the circuit depth and the number of correlated parameters $\bx$. In the following, we alternatively employ $\bx$ and $(\delta,\bm{J},t)$ to refer to the parameters of the quantum circuit.

The FSPT phase could be theoretically signified by the dynamical properties of the evolved state under the quantum circuit $U(\delta,\bm{J},t)$, i.e., the local spin magnetization 
\[\braket{Z_i(\delta,\bm{J},t)}=\braket{\psi_0|U(\delta,\bm{J},t)^{\dagger}Z_iU(\delta,\bm{J},t)|\psi_0}\] 
for all qubits $i\in [N]$. Here, the initial state $\ket{\psi_0}$ could be an arbitrary computational basis state. In general, to characterize the FSPT phases at a specific driven perturbation $\delta$, it is necessary to estimate the averaged local spin magnetization $\braket{\overline{Z_i(\delta,t)}}$ over a set of varying coupling parameters $\{\bm{J}^{(s)}\}_{s=1}^S$ at a range of evolution time $t$, namely
\begin{equation}\label{append:eqn:avg-local-spin}
    \braket{\overline{Z_i(\delta,t)}} = \frac{1}{S}\sum_{s=1}^S  \braket{Z_i(\delta,\bm{J}^{(s)},t)},
\end{equation}
where the evolution time $t$ varies in $\{k\mathrm{T}_1\}$ with $k \in [n_k]$. Here, $n_k/2$ represents the total evolution period. In the FSPT phase, the averaged local spin magnetization $\braket{\overline{Z_i(\delta,t)}}$  is expected to exhibit persistent period-doubled oscillations at the boundary spins (i.e., $i=1,N$) and rapidly decay to zero for the bulk spins  (i.e., $i=2,\cdots,N-1$). In contrast, for the thermal phase, the averaged local spin magnetization  $\braket{\overline{Z_i(\delta,t)}}$ is expected to decay quickly to zero for all spins (i.e., $i=1,\cdots,N$). 

Another challenging task in this context is to identify the critical point where the phase transition occurs.  According to the setting in Ref.~\cite{zhang_digital_2022}, the critical point amounts to the largest variance of the subharmonic spectral peak height, namely the amplitude $A_{\omega}(\delta,\bm{J})$ in the Fourier spectrum of the local magnetization $\braket{Z_1(\delta,\bm{J},t)}$ measured at $\omega = \omega_0/2$ for the boundary spin. To identify the phase transition point, an intensive number of quantum measurements are typically required to estimate $\braket{Z_1(\delta,\bm{J},t)}$ across the parameter space $(\delta,\bm{J})$ with the varied $t$. Specifically, the total number of measurements required to identify the phase transition point up to a precision $\epsilon$ is considerable, which scales linearly with the total evolution time $t$, the dimension of physical parameters $(\delta,\bm{J})$, the inverse of the precision $1/\epsilon$, and the number of shots used to estimate the expectation values for each circuit configuration.

\subsubsection{Workflow of enhancing FSPT phase identification via the predictive surrogate}\label{append:subsec:workflow_FSPT}
Before elaborating on the workflow, we would like first to outline the form of the predictive surrogate $h_{\mathsf{qs}}$ needed for FSPT phase identification. As discussed in SM~\ref{append:subsec:background_FSPT}, characterizing the FSPT phase requires estimating a set of expectation values $\braket{Z_i(\delta,\bm{J},t)}$ for the varied qubit indices with $i\in [N]$ at different evolution times with $t=k\mathrm{T}_1$ and $k\in [n_k]$. This amounts to implementing quantum circuits with varied observables $\{Z_i\}$ and ans\"atze $U(\delta,\bm{J},t)$, each with a circuit depth determined by $t$. Accordingly, $N\cdot n_k$ predictive surrogates $\{h_{\mathsf{qs}}^{(i,t)}\}$ should be employed to separately emulate mean-value behaviors of $\braket{Z_i(\delta,\bm{J},t)}$, whose construction details are deferred to SM~\ref{append:subsec:exp_implement_detail_mb}.

Given access to the optimized predictive surrogates $\{h_{\mathsf{qs}}^{(i,t)}\}$, the workflow of employing them to enhance FSPT phase identification involves two sequential steps: (i) FSPT phase and critical region identification with predictive surrogate $h_{\mathsf{qs}}$; and (ii) validation on quantum processors, as separately detailed below.

\smallskip

\noindent \underline{\textit{Step I: FSPT phase and critical region identification with predictive surrogate $h_{\mathsf{qs}}$.}} 
To characterize the FSPT phase and locate a small region around the critical point $\delta^*$, we first uniformly sample $n_d$ candidate values $\{\delta^{(l)}\}_{l=1}^{n_d}$ over the grid points within the parameter range $[\delta_{\min}, \delta_{\max}]$ with the interval  $\Delta \delta$. Without loss of generality, we assume that $\delta^{(1)}\le \delta^{(2)}\le \cdots \le \delta^{(n_d)}$. For each $\delta^{(l)}$, we generate $S$ disorder parameters $\{\bm{J}^{(l, s)}\}_{s=1}^S$, resulting in a total of $n_d\cdot S$ parameter combinations, denoted as $\mathcal{S}=\{(\delta^{(l)},\bm{J}^{(l,s)})\}_{l, s=1}^{n_d,S}$.  To characterize the FSPT phase at a specific parameter $(\delta^{(l)},\bm{J}^{(l,s)})$, the $N\cdot n_k$ predictive surrogates $h_{\mathsf{qs}}^{(i,t)}$ are employed to independently predict $\braket{Z_i(\delta^{(l)},\bm{J}^{(l,s)},t)}$ for the specified $i$ and the time $t$ over $\mathcal{S}$. As such, following the Eq.~(\ref{append:eqn:avg-local-spin}), the averaged local magnetization $\braket{\overline{Z_i(\delta^{(l)},t)}}$ for a given $\delta^{(l)}$ over the disorder parameters $\{\bm{J}^{(l,s)}\}_{s=1}^S$ is estimated by \[\overline{h_{\mathsf{qs}}^{(i,t)}(\delta^{(l)})}=\frac{1}{S}\sum_{s=1}^S h_{\mathsf{qs}}^{(i,t)}\Big(\delta^{(l)},\bm{J}^{(l, s)}\Big).\]

Recall that the phase transition is indicated by the largest variance of the subharmonic spectral peak height, namely the amplitude $A_{\omega}(\delta^{(l)},\bm{J}^{(l,s)})$ in the Fourier spectrum of the local magnetization $\braket{Z_1(\delta^{(l)},\bm{J}^{(l,s)},t)}$ measured at $\omega = \omega_0/2$ for the boundary spin. We then estimate $A_{\omega}(\delta^{(l)},\bm{J}^{(l,s)})$ by applying the Fourier transformation to $h_{\mathsf{qs}}^{(1,t)}(\delta^{(l)},\bm{J}^{(l,s)})$ as we do for $\braket{Z_1(\delta^{(l)},\bm{J}^{(l,s)},t)}$, and denote the estimation as $\hat{A}_{\omega}(\delta^{(l)},\bm{J}^{(l,s)})$. Hence, the estimated phase transition point is given by
\[\delta^{(l^*)}=\arg\max_{\delta^{(l)}}\Var\left(\hat{A}_{\omega}(\delta^{(l)})\right),~\mbox{with} ~ \Var\left(\hat{A}_{\omega}(\delta^{(l)})\right)= \frac{1}{S}\sum_{s=1}^S \left(\hat{A}_{\omega}\left(\delta^{(l)},\bm{J}^{(l,s)}\right)- \frac{1}{S}\sum_{s=1}^S\hat{A}_{\omega}\left(\delta^{(l)},\bm{J}^{(l,s)}\right)\right)^2.\]
Furthermore, the region $[\delta^{(l^*-1)}, \delta^{(l^*+1)}]$, using the neighboring points $\delta^{(l^*-1)}$ and $\delta^{(l^*+1)}$ as the endpoints, is identified as the critical region where the phase transition occurs. 

\smallskip
\noindent \underline{\textit{Step II: Validation on quantum processors.}} This step aims to verify the effectiveness of the identified critical region $[\delta^{(l^*-1)}, \delta^{(l^*+1)}]$ on the employed quantum processor. To do this, we uniformly sample $n_v$ different $\delta^{(l)}$ values over the grid points within $[\delta_{l^*-1}, \delta_{l^*+1}]$. We then follow the same procedure as we do in Step I, while the local magnetization is estimated by running the quantum processors rather than the predictive surrogate $h_{\mathsf{qs}}$. The identified region of the critical point is confirmed if the identified phase transition point $\delta^{(l^*)}$ maximizes the experimentally observed amplitude variance $\Var(A_{\omega}(\delta))$ on quantum processors.

\smallskip
\noindent\underline{\textit{Reduction in measurement achieved by surrogate-enhanced identification of critical points.}} The measurement cost of this method arises from the construction of predictive surrogates and their validation on quantum processors. Compared to direct quantum simulation on quantum processors, the surrogate-enhanced identification of critical points offers significant measurement reductions, particularly when the physical parameters span a large region and a high precision in identifying the critical region is required.
Specifically, the measurement cost of quantum simulation on quantum processors increases with the number of grid points within the parameter region. In contrast, the measurement cost for constructing surrogates is independent of this factor. Therefore, increasing the parameter range or raising the required precision necessitates more grid points for identifying the critical point. As a result, this leads to a larger discrepancy in measurement costs between the surrogate-enhanced identification method and the conventional Hamiltonian simulation approach.

\section{Additional experimental results} \label{apd:exp_info}

This section provides further details and supporting results for our experimental study. We first describe the device used in our experiments in SM~\ref{apd:device}. Next, we present additional experimental results of pre-training VQEs using the predictive surrogate $h_{\mathrm{cs}}$ in SM~\ref{apd:vqe_add_results}. We then discuss our methodology for FSPT phase transitions, supported by numerical simulations and analysis of measurement overhead in SM~\ref{apd:FSPT_add_results}. Finally, in SM~\ref{apd:vali_noise}, we provide a validation of noise characterization in quantum processor emulation.

\subsection{Device information}\label{apd:device}
\begin{figure*}[!h]
\centering\includegraphics[width=0.65\textwidth]{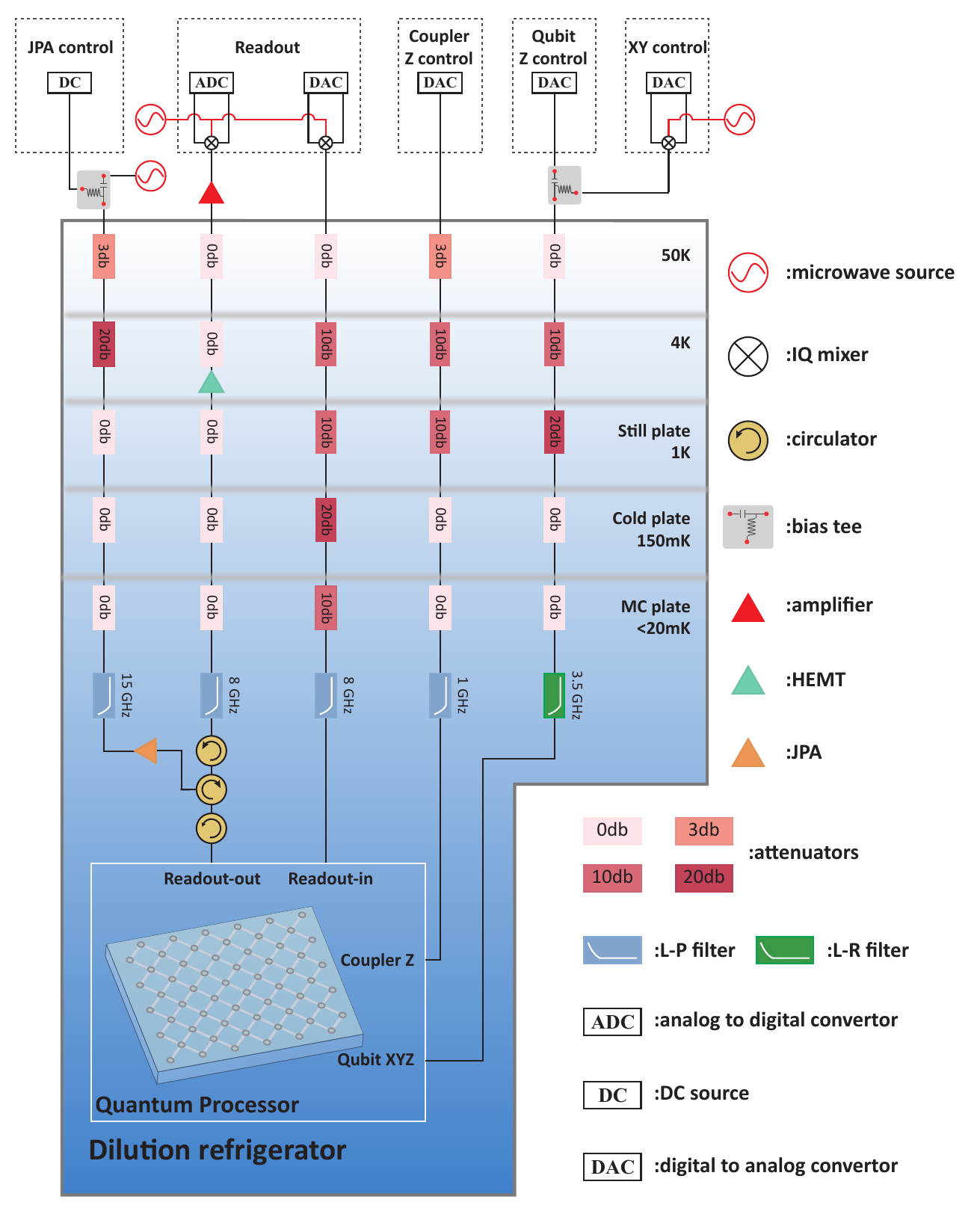}
\caption{\small{\textbf{ Schematic illustration of the superconducting quantum processor.}} Each qubit has independent XY and Z control lines. Room-temperature digital-to-analog converters (DACs) generate qubit XY control pulses and Z control pulses, which are combined via bias tees before entering the quantum processor. Microwave sources and mixers produce readout probing pulses and JPA pump signals.  Attenuators and filters are installed at multiple plates to suppress noise. Readout signals undergo three-stage amplification: first by a Josephson parametric amplifier (JPA) under 20 mK, then a high-electron mobility transistor (HEMT) at 4 K, and finally room-temperature amplifiers. Digitization and demodulation are performed by analog-to-digital converters (ADC) modules. DC sources provide static flux bias for JPA operation. Flux bias lines enable frequency tuning through Z-control pulses combined with microwave signals. }
\label{sfig:device_scheme}
\end{figure*}

\begin{figure*}[ht]
	\centering\includegraphics[width=1\textwidth]{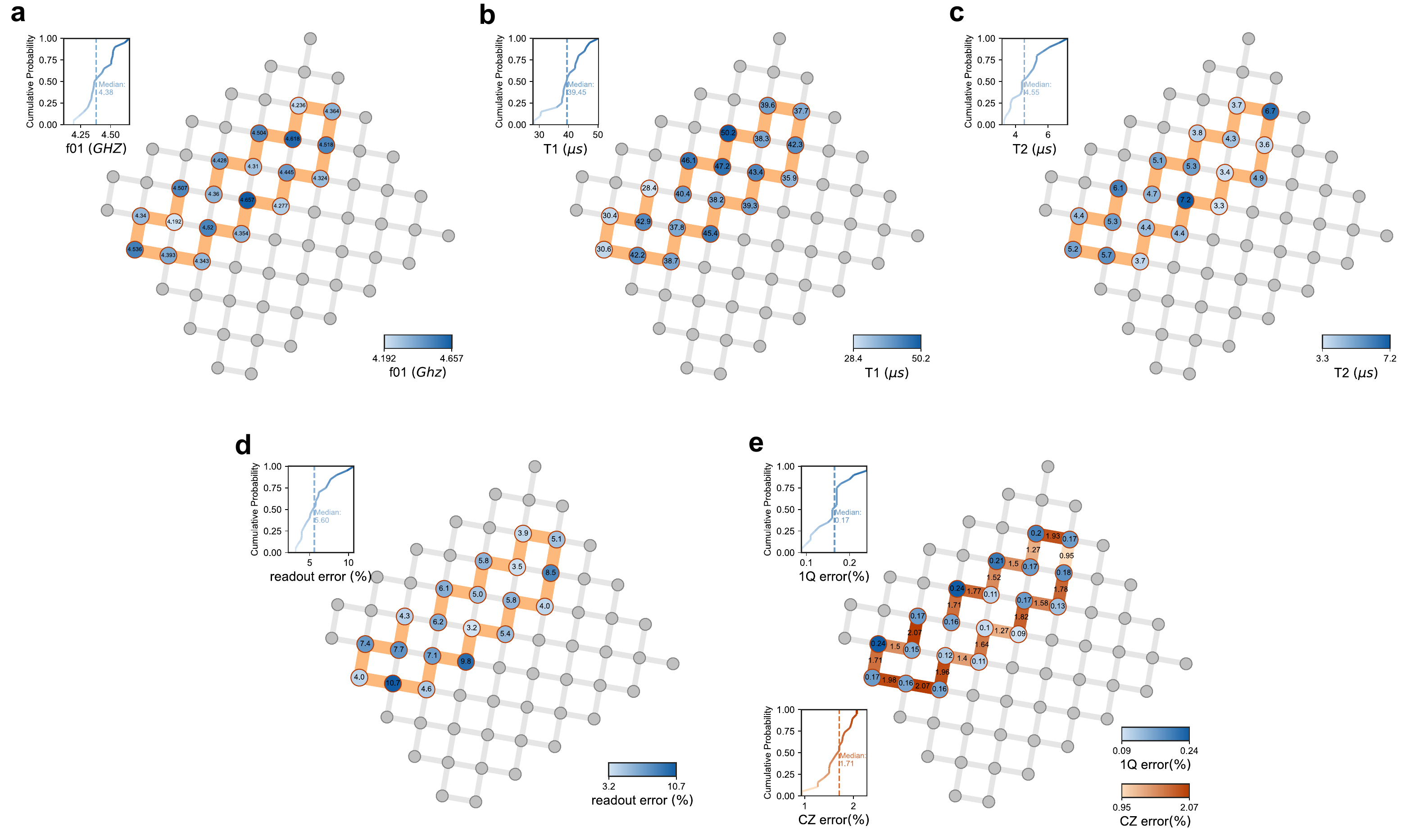}
    \caption{\small{\textbf{Device parameters of the employed superconducting quantum processor.} \textbf{a.} Distribution of Idle frequency for the qubits, with a median value of $4.38$ GHz. \textbf{b.} Distribution of $T_1$ relaxation times, with a median of $39.45 \mu s$. \textbf{c.} Distribution of $T_2$ coherence times, with a median of $4.55 \mu s$. \textbf{d.} Distribution of readout errors, with a median of $5.60\%$. \textbf{e.} Distribution of gate error rates for single-qubit gates, with a median of $0.17\%$, and for CZ gates, with a median of $1.71\%$.}}
	\label{sfig:device}
\end{figure*}

The experiments are conducted on a programmable superconducting quantum processor housed within a dilution refrigerator operated at temperatures below 20 mK. The quantum processor implements a $6\times11$ two-dimensional array of frequency-tunable transmon qubits configured in a cross-shaped architecture. Nearest-neighbor coupling is achieved via dedicated tunable couplers. For this study, a contiguous 20-qubit chain is constructed within the array, with various layouts selected according to specific experimental requirements. The device schematic is illustrated in FIG.~\ref{sfig:device_scheme}. In particular, each qubit is equipped with independent XY and Z control lines. Room-temperature digital-to-analog converters (DACs) generate the qubit XY control pulses and Z control pulses, which are combined via bias tees before entering the quantum processor. Microwave sources and mixers are responsible for producing the readout probing pulses and JPA pump signals. To ensure optimal performance and minimize noise, attenuators and filters are strategically installed at multiple temperature plates throughout the system. The readout signals undergo a three-stage amplification process: initial amplification by a Josephson parametric amplifier (JPA) operating at the base temperature of 20 mK, followed by amplification from a high-electron mobility transistor (HEMT) at the 4 K stage, and finally further amplification from room-temperature amplifiers. Digitization and demodulation of the signals are performed by analog-to-digital converter (ADC) modules. DC sources provide the necessary static flux bias for JPA operation, while flux bias lines enable precise frequency tuning through the combination of Z-control pulses with microwave signals.

After calibration and optimization, the quantum processor is characterized by the metrics depicted in FIG.~\ref{sfig:device}. The median idle frequency of the qubits is measured at 4.38 GHz. The $T_1$ relaxation times exhibit a median value of 39.45 $\mu s$, while the $T_2$ coherence times show a median of 4.55 $\mu s$. The readout error across the qubits has a median of 5.60$\%$. For gate operations, the median error rate for single-qubit gates is 0.17$\%$, and for CZ gates, it is 1.71$\%$. These parameters are determined through rigorous characterization processes to ensure the reliability and accuracy of the experimental results.

\subsection{More experimental results of pre-training VQEs by $h_{\textsf{cs}}$}\label{apd:vqe_add_results}

This subsection provides extended experimental results for VQE pre-training using \(h_{\textsf{cs}}\). We first detail the circuit implementation for TFIM ground state estimation in SM~\ref{append:sec-exp:subsec:circuit-imp}. Next, provide additional implementation details for the experiment in SM~\ref{append:sec:Exp:sub:imp-detail}. Subsequently, we evaluate the prediction performance of \(h_{\textsf{cs}}\) across multiple metrics in SM~\ref{apd:eva_hcs}. We then analyze the discrepancy between experimental and ideal processor outputs in SM~\ref{apd:experimental-ideal discrepancy}, followed by a numerical study of training data quality in SM~\ref{apd: data_quality}. Finally, we quantify the measurement reduction achieved by our pre-training protocol in SM~\ref{vqe_measurement_reduction}.

\begin{figure*}[h]
	\centering\includegraphics[width=0.65\textwidth]{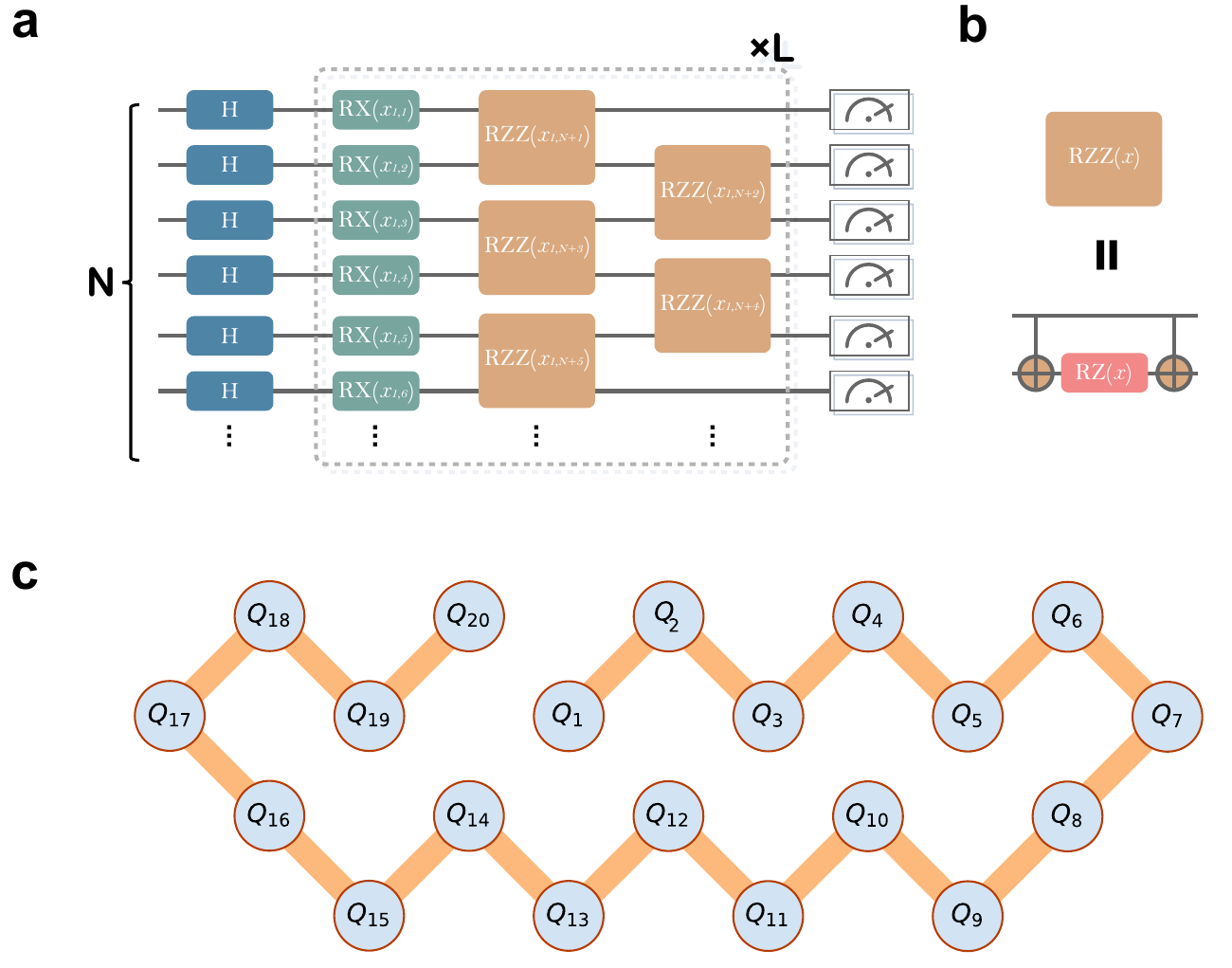}
        \caption{\small{\textbf{Implementations of quantum circuits in pre-training VQEs.} \textbf{a.} The $N$-qubit ansatz $U(\bx)$ used in the task of pre-training VQEs has a layer-wise structure, where the total number of layers is denoted by $L$. Each layer consists of $N$ single-qubit rotation gates $\RX$ and $(N-1)$ two-qubit rotation gates $\RZZ$. All parameters in $U(\bx)$ are independent.  \textbf{b.} Implementation of the two-qubit rotation gate $\RZZ(\bx)$ on the superconducting processor, where it is composed of two CNOT gates and a single-qubit rotation gate $\RZ(\bx)$.  \textbf{c.} The structure of the employed $20$ superconducting qubits in the experiments, where these qubits are ordered from $Q_1$ to $Q_N$ with $Q_i$ referring to the $i$-th qubit.}} 	\label{sfig:cir_layout1}
\end{figure*}

\subsubsection{Circuit implementation details}\label{append:sec-exp:subsec:circuit-imp}

In our pre-training VQE experiments, we apply $h_{\textsf{cs}}$ to estimate the ground state energy of the 1D TFIM Hamiltonian 
\begin{equation}
\mathsf{H}_{\mathsf{TFIM}}=-J\sum_iZ_iZ_{i+1}-h\sum_iX_i.
\end{equation}
Recall that the explicit form of the employed ansatz $U(\bx)$, as introduced in the main text, is  
\begin{equation}
U(\bx)=\prod_{l=1}^L \left(\prod_{i=1}^N  \RX_i(\bx_{l,i}) \prod_{i=1}^{N-1}\RZZ_{i,i+1}(\bx_{l,N+i})\right).
\end{equation}
The circuit diagram of  $U(\bx)$ is visualized in FIG.~\ref{sfig:cir_layout1}\text{a}. Specifically, the structure of $U(\bx)$  originates from the $L$-step Trotter decomposition of the time evolution operator $\exp({-\imath \mathsf{H}_{\mathsf{TFIM}}\tau})$, where the time step $\Delta t = \tau/L$. In the conventional Hamiltonian simulation, each layer would implement rotational quantum gates with fixed angles $2J\Delta t$ for ZZ interactions and $2h\Delta t$ for X fields. In contrast, $U(\bx)$ replaces these fixed couplings with tunable parameters per layer, i.e., $\bx_{l,i}$ replaces $2h\Delta t$ in the single-qubit rotational gate $\RX_i$, while $\bx_{l,N+i}$ replaces $2J\Delta t$ in the two-qubit rotational gate  $\RZZ_{i,i+1}$.

In our experiments,  $\RX_i(\bx_{l,i}) = \exp({-\imath   \bx_{l,i} X_i/2})$ is directly implemented on our quantum processor, while $\RZZ_{i,i+1}(\bx_{l,N+i}) = \exp({-\imath	 \bx_{l,N+i} Z_i Z_{i+1}/2})$ is realized through the standard decomposition as shown in FIG.~\ref{sfig:cir_layout1}\text{b}.  

\subsubsection{Experimental implementation details}\label{append:sec:Exp:sub:imp-detail}
This section presents the details of the construction of the predictive surrogate $h_{\mathsf{cs}}$, and the hyperparameter settings used for pre-training VQE in our experiments.

\smallskip 
\noindent \underline{\textit{Construction of the predictive surrogate $h_{\mathsf{cs}}$.}} The first step of constructing the predictive surrogate is to collect the training dataset $\mathcal{T}_{\mathsf{cs}}$. To achieve this, we first randomly generate classical input $\bx\in [-\pi,\pi]^d$, and feed it to the employed $N$-qubit quantum
processor to prepare the state $U(\bx)\ket{+}^N$. Denote the corresponding noisy state as $\tilderho(\bx)$. Then, we apply Pauli-based classical shadow with $T$ snapshots to constitute a single training example $(\bx,\tilderho_T(\bx))$ with  $\tilderho_T(\bx)$ being the classical shadow representation. This procedure is repeated $n$ times to construct the training dataset $\mathcal{T}_{\mathsf{cs}}$.

Once the training dataset $\mathcal{T}_{\mathsf{cs}}$ is prepared, we move on to construct $h_{\mathsf{cs}}(\bx,\mathsf{H})$. Building on the correspondence between regression techniques and kernel methods, we employ ridge kernel regression to optimize the predictive surrogate, thereby enhancing the efficiency of the optimization process. Mathematically, the predictive surrogate takes the form as $h_{\mathsf{cs}}(\bx,\mathsf{H};\bm{\mathrm{w}})=\braket{\bm{\Phi}_{\mathfrak{C}(\Lambda)}(\bx),\bm{\mathrm{w}}}$. Its optimization amounts to solving the ridge regression problem $\hat{\bm{\mathrm{w}}}=\arg\min_{\bm{\mathrm{w}}}\sum_{i=1}^n \left(h_{\mathsf{cs}}(\bx,\mathsf{H};\bm{\mathrm{w}})-g(\bxi) \right)^2 + \lambda \|\bm{\mathrm{w}}\|_2$, where $\lambda$ is the regularization parameter and $g(\bx)=\Tr(\tilderho_T(\bx)\mathsf{H})$ is the shadow estimation of $\Tr(\tilderho(\bx)\mathsf{H})$.

\smallskip 

\noindent \underline{\textit{Hyperparameter settings.}} The training dataset $\mathcal{T}_{\mathsf{cs}}$ is set with sizes $n \in \{400, 800, 1200, 1600, 2000\}$. For constructing the predictive surrogate $h_{\mathsf{cs}}$, we set the regularization parameter to $\lambda = 1$ and the frequency threshold to $\Lambda \in \{1, 2\}$. The initial parameters for pre-training the VQEs are uniformly sampled from the interval $[-\pi, \pi]^d$. We use the ADAM optimizer with a learning rate of $\eta = 0.1$ for pre-training the VQEs and additional fine-tuning on quantum processors~\cite{Kingma2014Adam}.  
To evaluate the convergence of the training process, we use the exponential moving average (EMA) of weights \cite{morales2024exponential}, with early stopping based on performance on the training data. The convergence threshold is set to $\epsilon = 10^{-4}$, and the decay factor for EMA is set to $\gamma = 0.9$.

Unless otherwise specified, all optimization problems in this study follow the same hyperparameter settings aforementioned. In addition, all numerical simulations were conducted using JAX 0.4.38 with a CUDA 12.1 backend.

\subsubsection{Evaluation of the prediction performance of $h_{\textsf{cs}}$} \label{apd:eva_hcs}
\begin{figure*}[]
	\centering
	\includegraphics[width=0.6\textwidth]{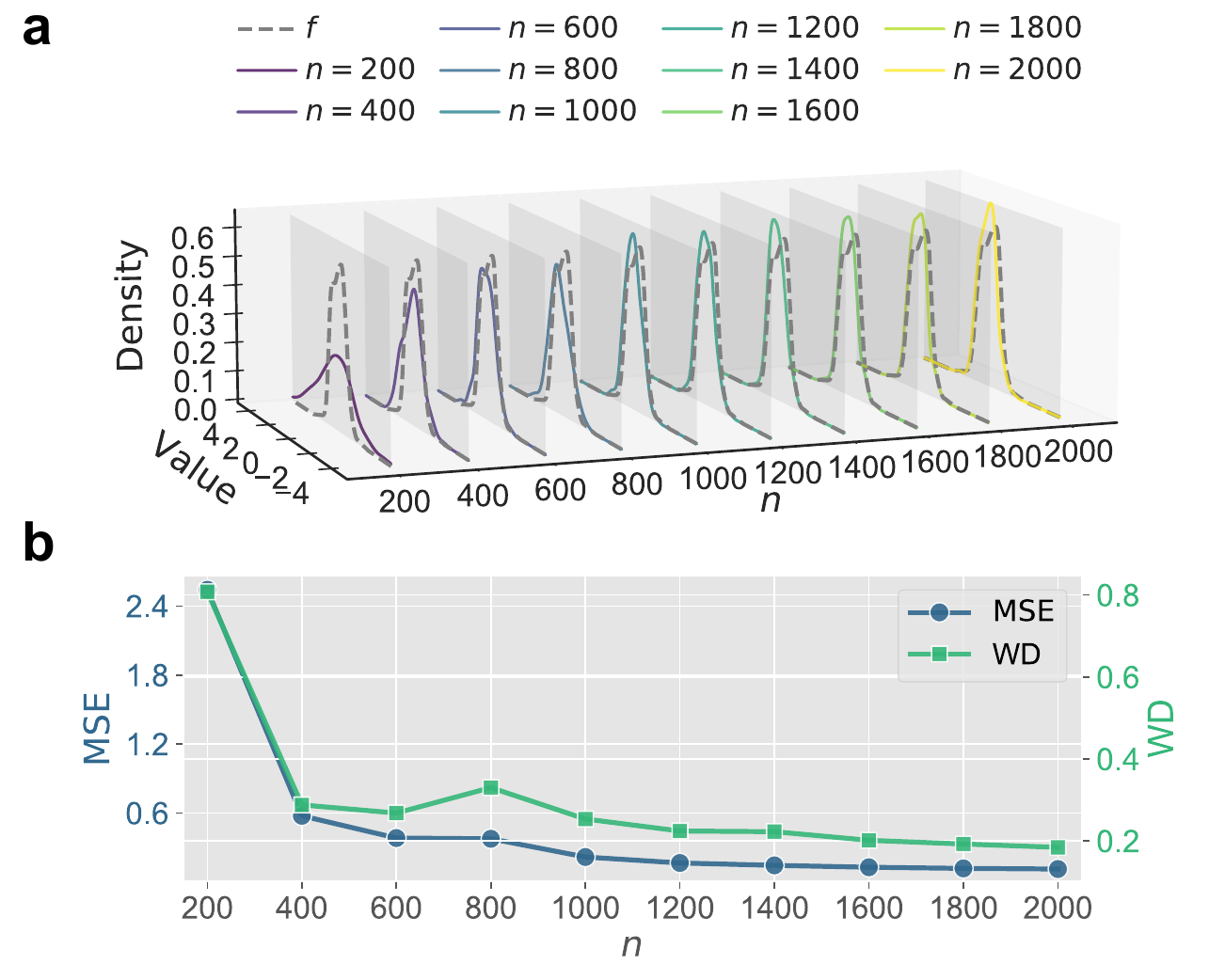}
    \caption{\small{ \textbf{Experimental results for the prediction performance of $h_{\mathsf{cs}}$.} 
	\textbf{a.} Kernel density estimates (KDEs) for the quantum processor outputs $f(\tilde{\rho}(\bx), O)$ (dashed curves) and the classical surrogate predictions $h_{\textsf{cs}}(\bx, O)$ (solid curves). The label `$n=a$' indicates that the number of training data is $a$. The bandwidth $\eta$ in the KDE is set via Scott's rule, i.e., $\eta = m^{-1/5}$ with $m = 200$ test points. 
	\textbf{b.} Mean squared error and Wasserstein distance between the surrogate predictions and the quantum processor outputs versus training set size $n$. The experimental results for these two metrics are plotted with the left $y$-axis labeled `MSE' (blue) and the right $y$-axis labeled `WD' (green), respectively. }}
	\label{sfig:eva_r_cs}
\end{figure*}

We provide a systematic analysis of the prediction performance of $h_{\mathsf{cs}}$ using identical settings to FIG.~2a in the main text, where the parameters for $\mathsf{H}_{\mathsf{TFIM}}$ are   $N=6$, $J=-0.1$, $h=-0.5$, $T=10$ snapshots per training sample, and frequency truncation $\Lambda=2$.

\smallskip
\noindent\textit{\underline{Brief overview of the employed metrics.}} We employ six metrics for the quantitative assessment. Let $\by = \{\yi\}_{i=1}^m$ denote quantum processor outputs $f(\tilde{\rho}(\bxi), O)$ and $\hat{\by} = \{\hat{y}^{(i)}\}_{i=1}^m$ denote the surrogate predictions (i.e., the output of $h_{\textsf{cs}}(\bx_i, O)$) over $m$ test examples. The definitions of the employed metrics are summarized below: 
\begin{itemize}
    \item Mean absolute error: $\text{MAE}=\frac{1}{m} \sum_{i=1}^{m} |\yi - \hat{y}^{(i)}|$;
    \item Coefficient of determination: $R^2 =1 - \frac{ \sum_{i=1}^{m} (\yi - \hat{y}^{(i)})^2 }{ \sum_{i=1}^{m} (\yi - \bar{y})^2 }$;
    \item Pearson's correlation coefficient: $ R = \frac{ \sum_{i=1}^{m} (\yi - \bar{y}) (\hat{y}^{(i)} -\text{mean}({\hat{y}})) }
         { \sqrt{ \sum_{i=1}^{m} (\yi - \bar{y})^2 } \cdot \sqrt{ \sum_{i=1}^{m} (\hat{y}^{(i)} - \text{mean}({\hat{y}}))^2 } }$, 
         where $\bar{y} = \frac{1}{m}\sum_{i=1}^m \yi$ and $\text{mean}({\hat{y}}) = \frac{1}{m}\sum_{i=1}^m \hat{y}^{(i)}$ are sample means of $\by$ and $\hat{\by}$, respectively;
    \item Mean squared error: $\text{MSE} =\frac{1}{m} \sum_{i=1}^{m} (\yi - \hat{y}^{(i)})^2$; 
	\item Kernel Density Estimation (KDE): $\text{Pr}(s;\by) = \frac{1}{m\eta}\sum_{i=1}^m K\left(\frac{s - \yi}{\eta}\right)$, 
		$\text{Pr}_{\hat{y}}(t) = \frac{1}{m\eta}\sum_{i=1}^m K\left(\frac{t - \hat{y}^{(i)}}{\eta}\right)$, 
		where $K(u) = \frac{1}{\sqrt{2\pi}}e^{-u^2/2}$ is the Gaussian kernel and $\eta = m^{-1/5}$ the bandwidth;
    \item Wasserstein distance: $\text{WD} = \inf_{\gamma \in \Gamma(\text{Pr}(\by), \text{Pr}(\hat{\by})} \int_{\mathbb{R}\times\mathbb{R}} |u - v|  d\gamma(u,v)$, 
         where $\Gamma$ denotes the set of all joint distributions with marginals $\text{Pr}(\by)$ and $\text{Pr}(\hat{\by})$, and $u, v \in \mathbb{R}$.
\end{itemize}

These fix metrics provide complementary insights: MAE quantifies mean absolute prediction error with robustness to outliers;  the metric $R$ measures the linear correlation strength between predictions and true values; $R^2$ evaluates the proportion of variance explained by the model; MSE emphasizes large errors through quadratic weighting; KDE enables non-parametric comparison of probability distributions; and WD assesses distributional similarity via optimal transport theory. Optimal performance yields MAE $\to 0$, MSE $\to 0$, WD $\to 0$, $R \to 1$, $R^2 \to 1$, and visual alignment in KDE plots.

\smallskip
\noindent\textit{\underline{Experimental results}}. Here we present extended results across $n \in \{200, 400, \dots, 2000\}$, assessed on $m=200$ test samples under the same hyperparameter settings as used for training.

The convergence of the distributions, as assessed by KDE, is shown in FIG.~\ref{sfig:eva_r_cs}\text{a}. As $n$ increases from 200 to 2000, the surrogate prediction's distribution (solid curves) converges toward the quantum processor output's distribution (dashed curves). This monotonic improvement confirms the efficiency of the predictive surrogate $h_{\mathsf{cs}}$ in learning both point estimates and distributional properties. 

In addition, FIG.~\ref{sfig:eva_r_cs}\text{b} quantifies the reduction in error as n increases. Specifically, MSE decreases from $2.49$ at $n=200$ to $0.11$ at $n=2000$, while WD drops from $0.81$ to $0.18$ over the same range. The MSE and WD trends directly align with the MAE, $R^2$, and $R$ convergence observed in FIG.~2\text{a}, demonstrating consistent performance improvement across all metrics.

\subsubsection{Analysis of experimental-ideal processor discrepancy}\label{apd:experimental-ideal discrepancy}

\begin{figure*}[t]
	\centering
	\includegraphics[width=1\textwidth]{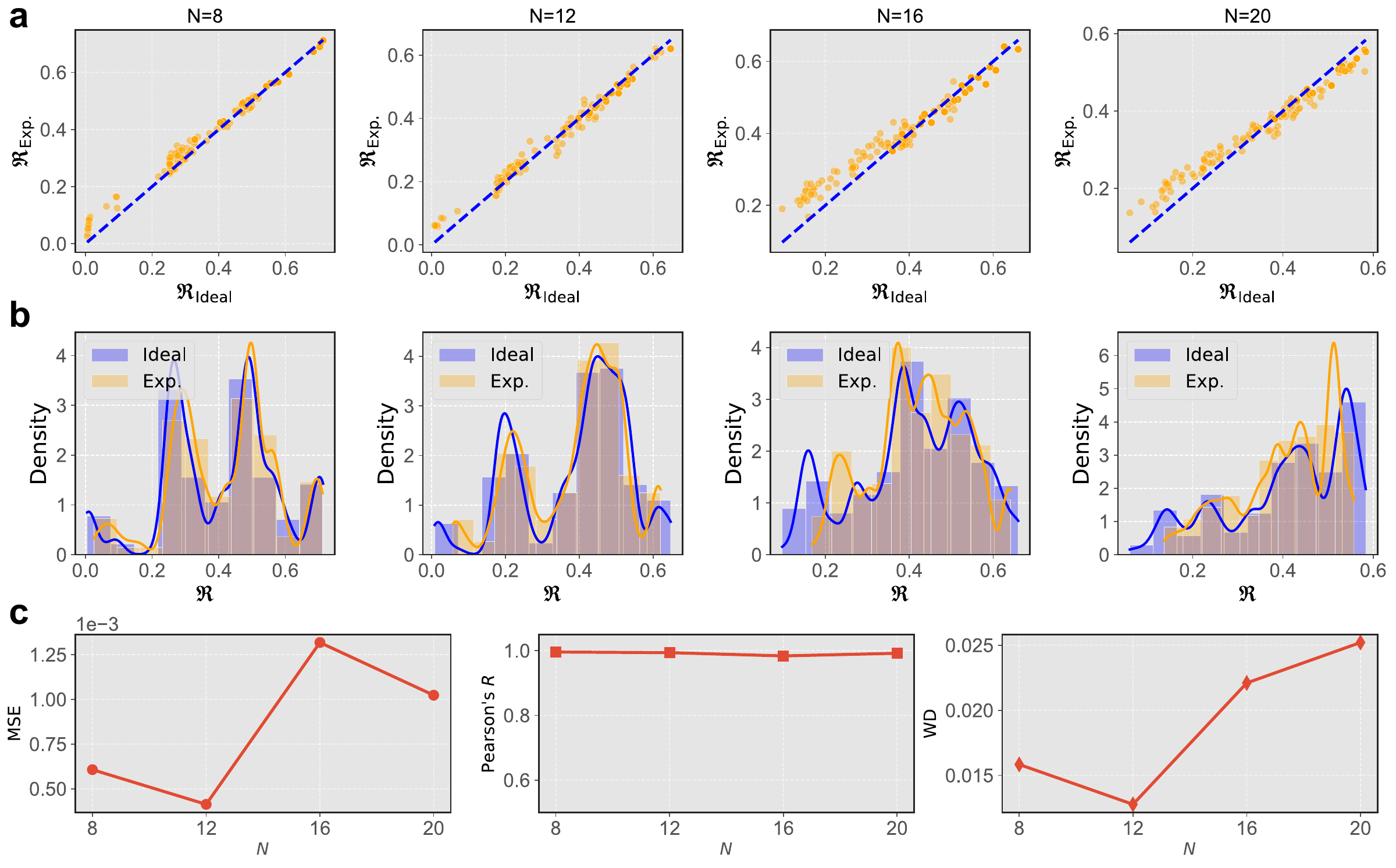}
    \caption{\textbf{Experimental results for the discrepancy between experimental quantum processor outputs and ideal theoretical predictions.} 
	\textbf{a.} The normalized deviation values \(\mathfrak{R}_{\text{Exp.}}\) evaluated on quantum processors (orange points) against the normalized deviation values \(\mathfrak{R}_{\text{Ideal}}\) obtained from noiseless numerical simulation (blue dashed line: \(y = x\)) across four system sizes \(N = \{8,12,16,20\}\). The plotted data points are identical to those used in FIG.~3, where \(\mathfrak{R}_{\text{Exp.}}\) corresponds to unhatched bars labeled 'Exp.' and \(\mathfrak{R}_{\text{Ideal}}\) corresponds to hatched bars labeled 'Ideal'. 
	\textbf{b.} Probability density distribution of the normalized deviation values. The labels `Ideal' and `Exp.' indicate that the results are obtained from noiseless numerical simulation and quantum processor outputs, respectively. The histograms with equal-width bins refer to the statistical results, and the smooth curves refer to Gaussian kernel density estimations with a bandwidth of $\eta \approx 0.35$.
	\textbf{c.} Quantitative analysis  of differences between the quantum processor outputs and the noiseless numerical simulations for varying system sizes $N=\{8,12,16,20\}$. The left, middle, and right panels present the experimental results under the metrics of the mean squared error (MSE), Pearson's correlation coefficient $R$, and Wasserstein distance (WD), respectively.}
	\label{sfig:gap}
\end{figure*}

We now provide an extended analysis of the data presented in FIG.~3, focusing specifically on pointwise discrepancies between experimental quantum processor outputs and ideal theoretical outputs. To this end, we systematically compare the paired data from FIG.~3, i.e., the experimental normalized deviation \(\mathfrak{R}_{\text{Exp.}}\) (indicated by the unhatched bars with the label `Exp.'), which is measured directly on our experimental quantum processor, and the ideal normalized deviation \(\mathfrak{R}_{\text{Ideal}}\) (indicated by the hatched bars with the label `Ideal'), which is computed through noiseless numerical simulation.

The dataset used here comprises all data points from FIG.~3 across four system sizes $N \in \{8,12,16,20\}$. For each system size, we include six training configurations: the initialized VQE (without pre-training), and pre-trained VQEs with five distinct sample sizes $n \in \{2^4, 2^6, 2^8, 2^{10}, 2^{11}\}$. Each configuration contains $20$ random instances, which correspond to different initialized parameters before pre-training. In this way, there are in total $6 \times 20 = 120$ paired results \((\mathfrak{R}_{\text{Exp.}}^{(i)}, \mathfrak{R}_{\text{Ideal}}^{(i)})\) per system size.

The comparison of all experimental outputs against their ideal counterparts is illustrated in FIG.~\ref{sfig:gap}\text{a}. The tight clustering of points along the $y = x$ reference line demonstrates systematic agreement across all system sizes. While minor dispersion appears at larger system sizes with $N=16$ and $20$, indicating expected noise amplification, the fundamental alignment remains uncompromised.

We employ two complementary methods to visualize the distribution of the data. First, we use histograms with ten equal-width bins, where the heights are normalized so that the total area equals one, providing a straightforward view of the distribution’s shape. Second, we apply Gaussian kernel density estimation with a bandwidth of $\eta \approx 0.35$, which produces a smooth, continuous estimate of the underlying probability density function. These two approaches together offer both discrete and continuous perspectives on the data distribution. As indicated in  FIG.~\ref{sfig:gap}\text{b}, experimental results indicate that both metrics reveal substantial distributional overlap between experimental and ideal outputs. 

Quantitative analysis under the measures introduced in SM~\ref{apd:eva_hcs}  reveals consistently strong agreement across all system sizes $N$. As shown in  FIG.~\ref{sfig:gap}\text{c}, MSE remains below $1.32 \times 10^{-3}$, Pearson's $R$ exceeds 0.98, and WD stays below 0.0252. These metrics collectively validate that experimental quantum processor outputs maintain structural correspondence with ideal predictions across diverse system sizes and training configurations, confirming hardware reliability for our pre-training VQE experiments.

\subsubsection{Numerical analysis of training data quality} \label{apd: data_quality}

\begin{figure*}[t]
\centering
\includegraphics[width=0.9\textwidth]{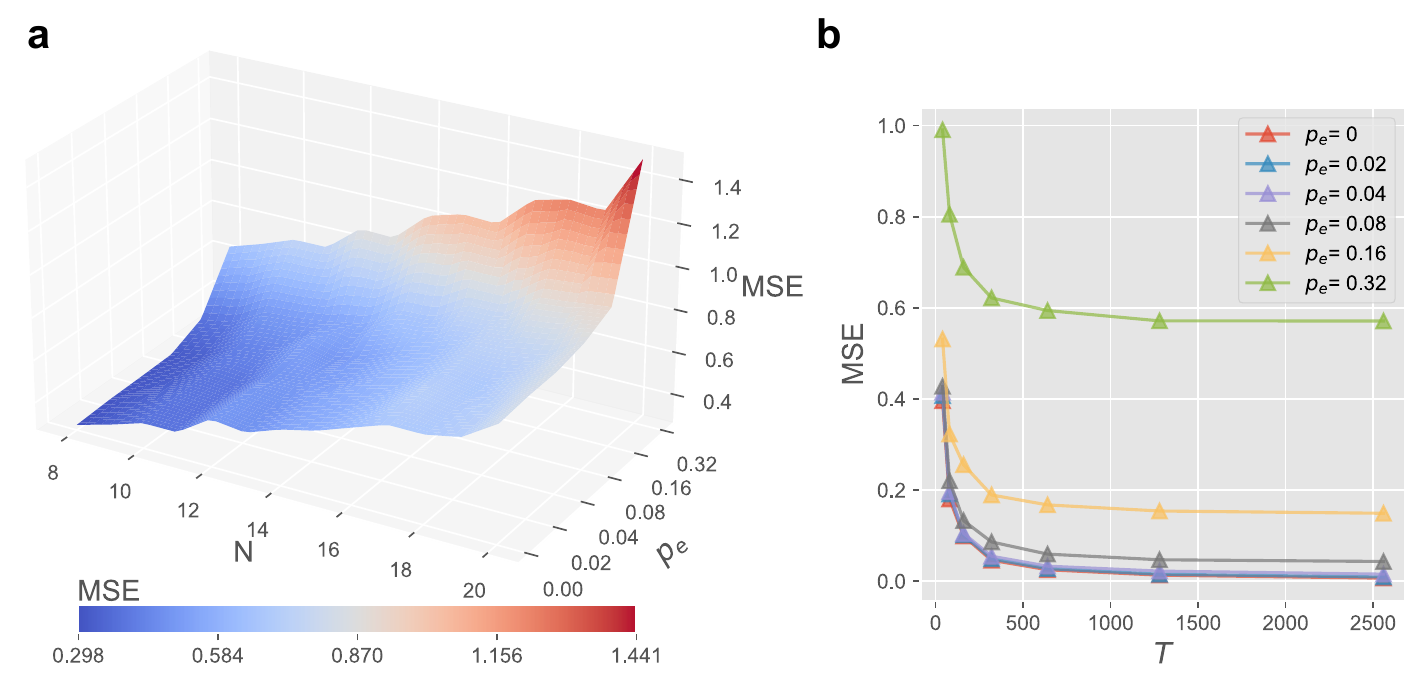}
\caption{\small\textbf{Performance of shadow estimations under varying training data quality.} 
\textbf{a.} Mean squared error (MSE) between the Pauli-based classical shadow estimations $\Tr(\tilderho_{T}(\bxi)\mathsf{H}_{\mathsf{TFIM}})$ and the ground truth expectation values $\Tr(\rho(\bxi)\mathsf{H}_{\mathsf{TFIM}})$ with varying system sizes $N \in \{8,10,\dots,20\}$ and bit-flip error probabilities $p_e \in \{0, 0.02, \cdots, 0.32\}$. The snapshot count $T$ is fixed to be $T=20$. 
\textbf{b.} MSE with varying snapshot counts $T \in \{20,40,\cdots,2560\}$ and the bit-flip error probabilities $p_e \in \{0, 0.02, \cdots, 0.32\}$, while the system size is fixed to be $N=20$.}
\label{sfig:biterror}
\end{figure*}

Recall that in the scalability experiments (i.e., FIG.~3\text{a} in the main text), the number of snapshots $T$ is slightly increased for each training example, which is from $T=20$ for $N=8$ to $T=100$ for $N=20$. The main reason for this adjustment is to suppress the error incurred by shadow estimation. As indicated in Theorem~1 and Lemma~\ref{lem:estimation-error-geo-kernel}, the shadow estimation error dominates the prediction error of $h_{\mathsf{cs}}$, scaling as $3^KB$. 
As the number of qubits increases, the observable norm $B$ also grows, requiring an increased number of snapshots $T$ to suppress the shadow estimation error. Fortunately, this adverse effect can be mitigated by increasing the number of snapshots $T$. However, it is noteworthy that the experimental noise introduces bit-flip errors with probability $p_e$, which can degrade the performance of the predictive surrogate even with an increased number of $T$. In the following,  we conduct numerical simulations to explore how the system size $N$, the bit-flip error $p_e$, and the number of measurements $T$ affect the estimation error. 

To achieve this goal, we employ a family of $N$-qubit TFIM models with $J=-0.1$ and $h=-0.5$. The circuit implementation is the same as those introduced in SM~\ref{append:sec-exp:subsec:circuit-imp}. We systematically explored the following parameter ranges, i.e.,  $N \in \{8,10,\dots,20\}$,  $T \in \{20,40,\dots,2560\}$, and $p_e \in \{0, 0.02, 0.04, 0.08, 0.16, 0.32\}$. For each configuration $(N, T, p_e)$, we performed $n=1000$ independent trials. For each trial, we compute exact expectation values $\Tr(\rho(\bxi)\mathsf{H}_{\mathsf{TFIM}})$ as ground truth, record noisy Pauli-based classical shadows $\tilderho_{T}(\bxi)$ with the bit-flip error $p_e$, and calculate the shadow estimation $\Tr(\tilderho_{T}(\bxi)\mathsf{H}_{\mathsf{TFIM}})$. After that, we exploit MSE to measure the discrepancy between the ground truth and shadow estimations across all $n$ trials.

The scaling behavior of the predictive surrogate with fixed $T=20$ is shown in FIG.~\ref{sfig:biterror}\text{a}. An observation is that the achieved  MSE exhibits a strong positive correlation with both system size $N$ and bit-flip error $p_e$. Specifically, at $N=20$ and $p_e=0.32$, MSE reaches the maximum value $1.49$, which is approximately $5$ times higher than that observed at $N=8$ with $p_e=0$. These results hint that with well-controlled $p_e$, the estimation error can be effectively suppressed with a slightly increased $T$.

We next turn to explore the scaling behavior of the predictive surrogate with the fixed $N=20$. The achieved simulation results are exhibited in FIG.~\ref{sfig:biterror}\text{b}. In particular, MSE decreases with increasing $T$ across all $p_e$, while higher values of $p_e$ tend to elevate the entire curve. For example, when $p_e=0.32$, the MSE plateaus around  $0.571$ even as $T$ increases, while setting $p_e=0$ enables MSE to drop below $0.006$. These observations further confirm the necessity of increasing $T$ in our scalability experiments. Given that the employed quantum processor exhibits relatively low $p_e$,  increasing $T$ slightly is an effective strategy for controlling estimation errors and maintaining consistent data quality across different system sizes.

\subsubsection{Measurement reduction analysis}\label{vqe_measurement_reduction}
\begin{figure*}[t]
	\centering\includegraphics[width=0.7\textwidth]{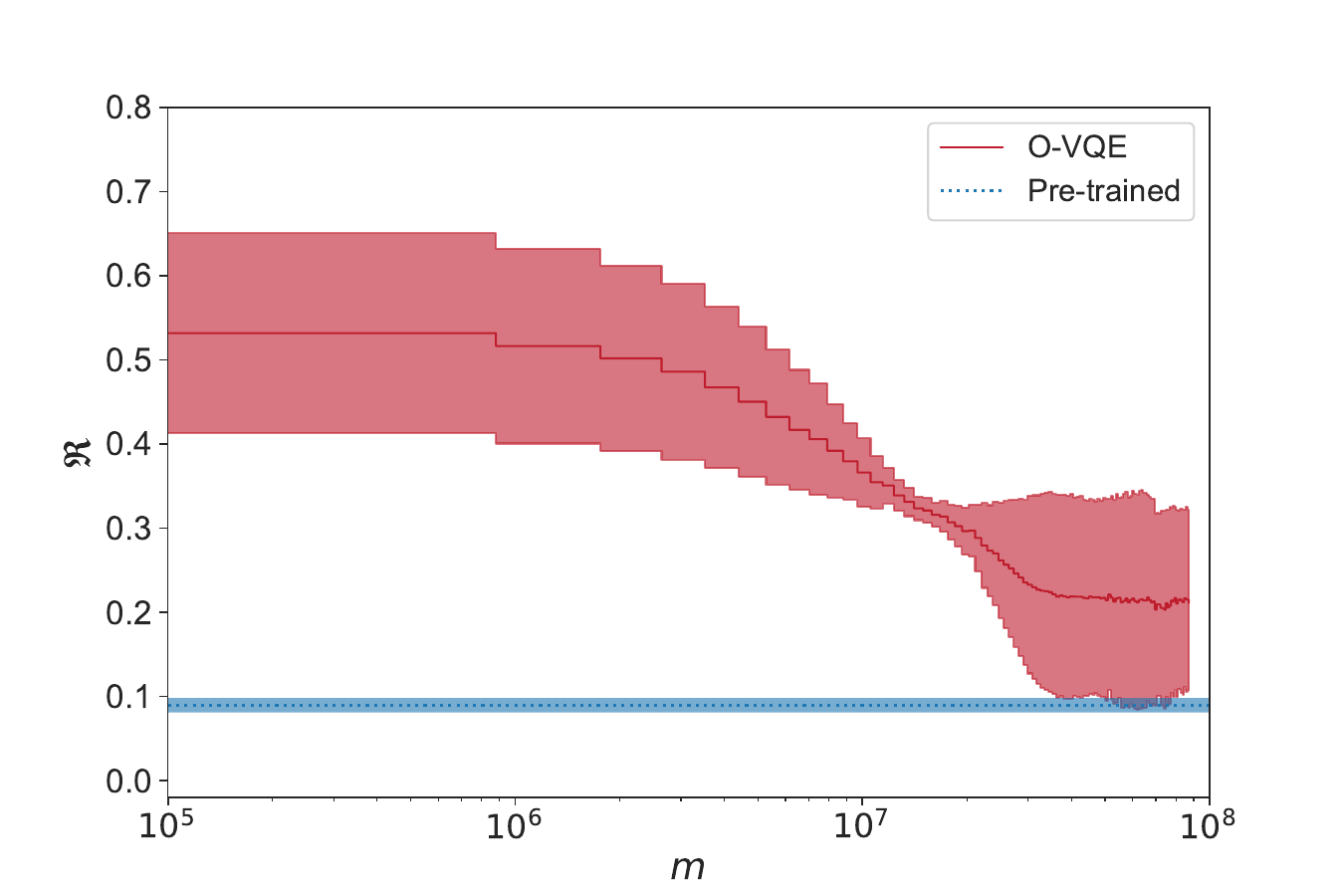}
	\caption{\small{\textbf{The normalized deviation $\mathfrak{R}$ versus the number of measurements used in the optimization of the original VQEs.}} The red solid line labeled by `O-VQE' refers to the normalized deviation $\mathfrak{R}(\bx)$ of the original VQEs. The blue dashed line labeled by `Pre-trained' corresponds to the experimental results of the pre-training phase for the predictive surrogate, serving as the baseline reference. The predictive surrogate requires only $m=20000$ measurements for training. The shaded regions represent the standard deviations across $5$ independent experiments.}
	\label{sfig:resource_reduction}
\end{figure*}

The measurement reduction analysis conducted here quantifies the measurement efficiency of our pre-training protocol, where the relevant results are supplementary to FIG.~2\text{c}. In particular, the analysis of resource consumption, or equivalently, the total number of measurements, is calculated as follows. 

For a conventional VQE, the gradient of each optimization step is computed by parameter-shift rules, requiring $2d$ circuit evaluations with $d=11$ being the parameter count used in our experiments. Moreover, the expectation value in the parameter-shift rules is estimated by $40,000$ shots to ensure the statistical precision. Accordingly, the number of measurements per-step cost is $2 \times 40,000 \times 11 = 880,000$ shots. When the number of optimization steps is $100$, the total number of measurements is $100 \times 880,000 = 88,000,000$. 

For the proposed pre-training protocol based on the predictive surrogate, the initial pre-training phase employs $n=2000$ training examples with $T=10$ snapshots each, requiring only $20,000$ shots—a mere $0.023\%$ of conventional VQE's resource consumption. Crucially, as demonstrated in FIG.~\ref{sfig:resource_reduction}, this phase (blue dashed) already achieves superior performance with normalized deviation $\mathfrak{R} = 0.090 \pm 0.009$, outperforming conventional VQE's $\mathfrak{R} = 0.211 \pm 0.109$(red solid) attained at full $88,000,000$-shot cost.

These results conclusively demonstrate that strategic classical pre-training fundamentally reshapes VQE optimization dynamics, enabling high-precision solutions with orders-of-magnitude reduction in quantum measurement requirements.

\subsection{Predicting FSPT phase transition} \label{apd:FSPT_add_results}

This subsection presents implementation details and supplemental results for FSPT phase identification. We first describe the circuit implementation for FSPT phase identification in SM~\ref{apd:FSPT_circuit}. We then introduce the omitted details of constructing the predictive surrogate $h_{\mathsf{qs}}$ in SM~\ref{append:subsec:exp_implement_detail_mb}. Next, we validate the prediction performance of surrogates \(h_{\mathsf{qs}}\) during the optimization process in SM~\ref{apd:eva_hmb}. We then demonstrate consistency between surrogate predictions and noisy numerical simulations for phase transitions identification in SM~\ref{apd:numerical_simu_FSPT}. Finally, we quantify measurement overhead reduction achieved via surrogate-enhanced critical region identification in SM~\ref{apd:FSPT_measurement_reduction}.
 
\subsubsection{Circuit implementation for FSPT phase identification}\label{apd:FSPT_circuit}

We follow the conventions of Ref.~\cite{zhang_digital_2022} and the description in SM~\ref{appendix:sec-F:subsec:FSPT} to construct the quantum circuit $U(\bx)$ for FSPT phase identification. Recall that the time-evolution unitary $U(\delta,\bm{J},t)$ is composed of two distinct components: $U_1(\delta,t) = e^{-\imath t\mathsf{H}_{\mathsf{tp},1}}$ and $U_2(\bm{J},t) = e^{-\imath t\mathsf{H}_{\mathsf{tp},2}}$, implemented sequentially over driving cycles $k\mathrm{T}_1$, namely
\begin{equation}
    U(\delta,\bm{J},k\mathrm{T}_1) = \begin{cases}
    \left[U_2(\bm{J},2\mathrm{T}_1)U_1(\delta,2\mathrm{T}_1)\right]^{k/2}, & \mbox{$k$ is even}, \\
    U_1(\delta,2\mathrm{T}_1)\left[U_2(\bm{J},2\mathrm{T}_1)U_1(\delta,2\mathrm{T}_1)\right]^{(k-1)/2}, & \mbox{$k$ is odd}.
            \end{cases} \nonumber
\end{equation}

The circuit implementation of $U(\delta,\bm{J},k\mathrm{T}_1)$ is shown in FIG.~\ref{sfig:cir_layout2}\text{a}. Specifically, the operator $U_1(\delta,2\mathrm{T}_1)$, generated by the Hamiltonian $\mathsf{H}_{\mathsf{tp},1}$, contains only single-qubit $\RX$ gates. By setting $2\mathrm{T}_1=2$, its explicit form is \[U_1(\delta,2)=\prod_{i=1}^N\RX_i(2\pi-\delta).\] In addition, the unitary operator $U_2(\bm{J})$, generated by the Hamiltonian $\mathsf{H}_{\mathsf{tp},2}$ involving the three-body interaction, could be approximately implemented using only two-body gates with precision over $0.9999$ by employing the neuroevolution algorithm proposed by Ref.~\cite{lu2021markovian}. In particular, the quantum circuit has the form  \[U_2(\bm{J},2)=\prod_{i=1}^{N-1} \CRZ_{i+1,i}(\pi) \prod_{i=2}^{N-1}\RY(-2J_i) \prod_{i=1}^{N-1} \CRZ_{i+1,i}(-\pi).\] 

The experimental implementation of $U(\bx)$ is shown in FIG.~\ref{sfig:cir_layout2}\text{b}, where the nearest-neighbor connectivity matches the interaction structure of $U_2$. When expressing  $U(\delta,\bm{J},k\mathrm{T}_1)$ as the quantum circuit $U(\bx)$ regarding the parameter $\bx$, the relation between $\delta,\bm{J},t=k\mathrm{T}_1$ and the parameter $\bx$ for an even number $k$ is \[\bx=(\delta,\cdots,\delta,J_2,\cdots,J_{N-1})^{k/2},\]
where the gate parameters $\delta$ and $\bm{J}=(J_2,\cdots,J_{N-1})$ are correlated.

\begin{figure*}[t]
    \centering\includegraphics[width=0.6\textwidth]{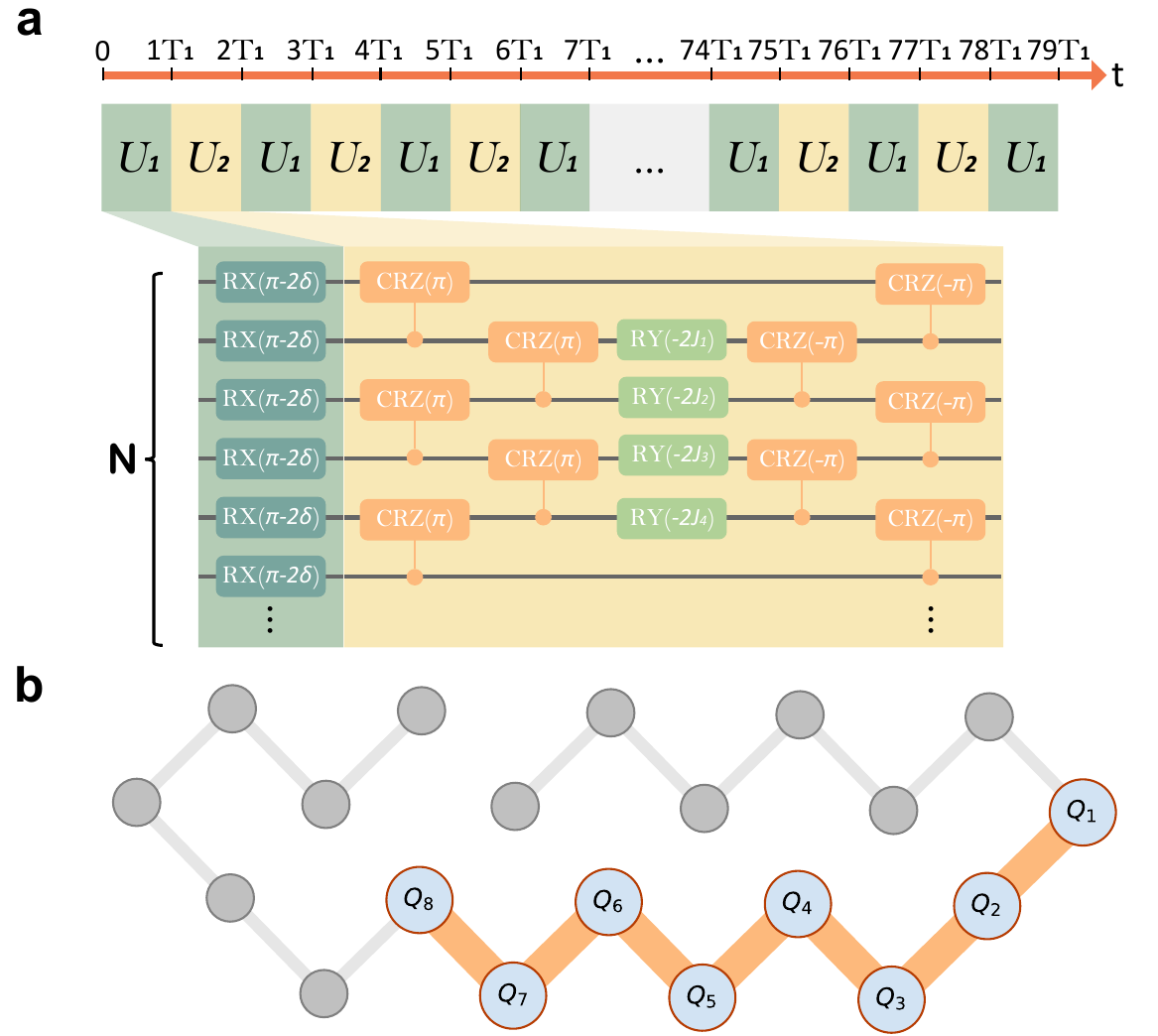}
    \caption{\small{\textbf{Implementation of quantum circuits for the FSPT phase identification.} \textbf{a.} The quantum circuit $U(\bx)$ consists of two alternating layers $U_1$ and $U_2$, with the circuit depth determined by the evolution time $t$. The layer $U_1$ consists of $N$ single-qubit rotation gates $\RX(\pi-2\delta)$. The layer $U_2$ consists of $(N-2)$ single-qubit rotation gates $\RY(-2J_i)$ with $i\in\{2,\cdots,N-1\}$ and $2(N-1)$ two-qubit rotation gates $\CRZ$ with parameter $\pm \pi$. \textbf{b.} The layout of the physical qubits in the employed quantum processors with nearest-neighbor connectivity.}}
    \label{sfig:cir_layout2}
\end{figure*}

\subsubsection{Experimental implementation details}\label{append:subsec:exp_implement_detail_mb}

We next present the details of the construction of the predictive surrogates $\{h_{\mathsf{qs}}^{(i,t)}\}$ for different observables $\{Z_i\}_{i=1}^N$ at different evolution times $t$, and the hyperparameter settings in identifying the critical point with $h_{\mathsf{qs}}$.

\smallskip
\noindent\underline{\textit{Construction of the predictive surrogates $h_{\mathsf{qs}}^{(i,t)}$.}} As explained in SM~\ref{append:subsec:workflow_FSPT}, we need to construct $Nn_k$ separate predictive surrogates $\{h_{\mathsf{qs}}^{(i,t)}\}$ to emulate the mean-value behavior of $\braket{Z_i(\bx)}$, where the parameter $\bx$ is specified by $(\delta,\bm{J},t)$. As the process for constructing these predictive surrogates is the same, our main focus is on describing the construction procedure of $h_{\mathsf{qs}}^{(1,t')}$ for the fixed observable $Z_1$ and a quantum circuit $U(\bx)$ with the evolution time being $t'$. For ease of notation, we henceforth denote $h_{\mathsf{qs}}^{(1,t')}$simply as $h_{\mathsf{qs}}$.

We begin by constructing the dataset $\mathcal{T}_{\mathsf{mb}} = \{(\bx^{(l)}, g(\bx^{(l)}))\}$, where $g(\bx^{(l)})$ represents the estimated mean value of $\braket{Z_1(\bx^{(l)})}$. To do this, we first generate classical input $\bx$, randomly sampling from the distribution $\mathbb{D}$, and use it as input to the specified $N$-qubit quantum processors to obtain the pre-measured state $\tilde{\mathcal{U}}(\bx)[(\ket{0}\bra{0})^N]$. We then apply the relevant Pauli measurement $Z_1$ to form a single training example $(\bx, g(\bx))$. This procedure is repeated $n$ times to build the training dataset $\mathcal{T}_{\mathsf{mb}}$.

Given access to  $\mathcal{T}_{\mathsf{mb}}$, the predictive surrogate $h_{\mathsf{qs}}(\bx) = \braket{\bm{\Phi}_{\Omega(\Lambda)}(\bx), \bm{\mathrm{w}}}$ is optimized by solving the ridge regression model in Eq.~\eqref{append:eq:ridge_reg_model_truncated}. The feature map $\bm{\Phi}_{\Omega(\Lambda)}$ is a subvector of the truncated feature vector $\bm{\Phi}_{\mathfrak{C}(\Lambda)}$, constructed by randomly and uniformly sampling $|\Omega(\Lambda)|$ elements without replacement from the set of trigonometric monomials $\{\Phi_{\bomega}(\bx) | \bomega \in \mathfrak{C}(\Lambda)\}$. This design with a dimension $|\bm{\Phi}_{\mathfrak{C}(\Lambda-1)}| < |\bm{\Phi}_{\Omega(\Lambda)}(\bx)| < |\bm{\Phi}_{\mathfrak{C}(\Lambda)}|$, which balances the computational cost of $\bm{\Phi}_{\mathfrak{C}(\Lambda)}(\bx)$ and the error introduced by $\bm{\Phi}_{\mathfrak{C}(\Lambda-1)}(\bx)$.

\smallskip
\noindent \underline{\textit{Hyperparameter settings.}} The training dataset $\mathcal{T}_{\mathsf{mb}}$ for constructing each $h_{\mathsf{qs}}^{(i,t)}$ is set with the size $n=250$. The training parameters are sampled from the Beta distribution with $\delta \sim \text{Beta}(\alpha=0.9,\beta=2)$ to enhance sampling density near the critical region $\delta=0$, and the disorder coupling parameters $\bm{J}$ are independently sampled from the uniform distribution over $[0,2]^N$. The evolution time is set to $t\in\{k\mathrm{T}_1:k\in [n_k]\}$ with $n_k=79$. For each parameter pair $(\delta,\bm{J})$ and time $t$, the local spin magnetization values are estimated from $T=40000$ measurements.  For constructing the predictive surrogates $h_{\mathsf{qs}}$, we set the regularization parameter to $\lambda = 1$ and construct the feature map $\bm{\Phi}_{\Omega(\Lambda)}$ by randomly sampling $1000$ features from $\mathfrak{C}(\Lambda)$ with the frequency threshold set to $\Lambda =7$. 

\subsubsection{Evaluation of the prediction performance of $h_{\mathsf{qs}}$ during the optimization process}\label{apd:eva_hmb}
\begin{figure*}[htbp]
	\centering
	\includegraphics[width=0.9\textwidth]{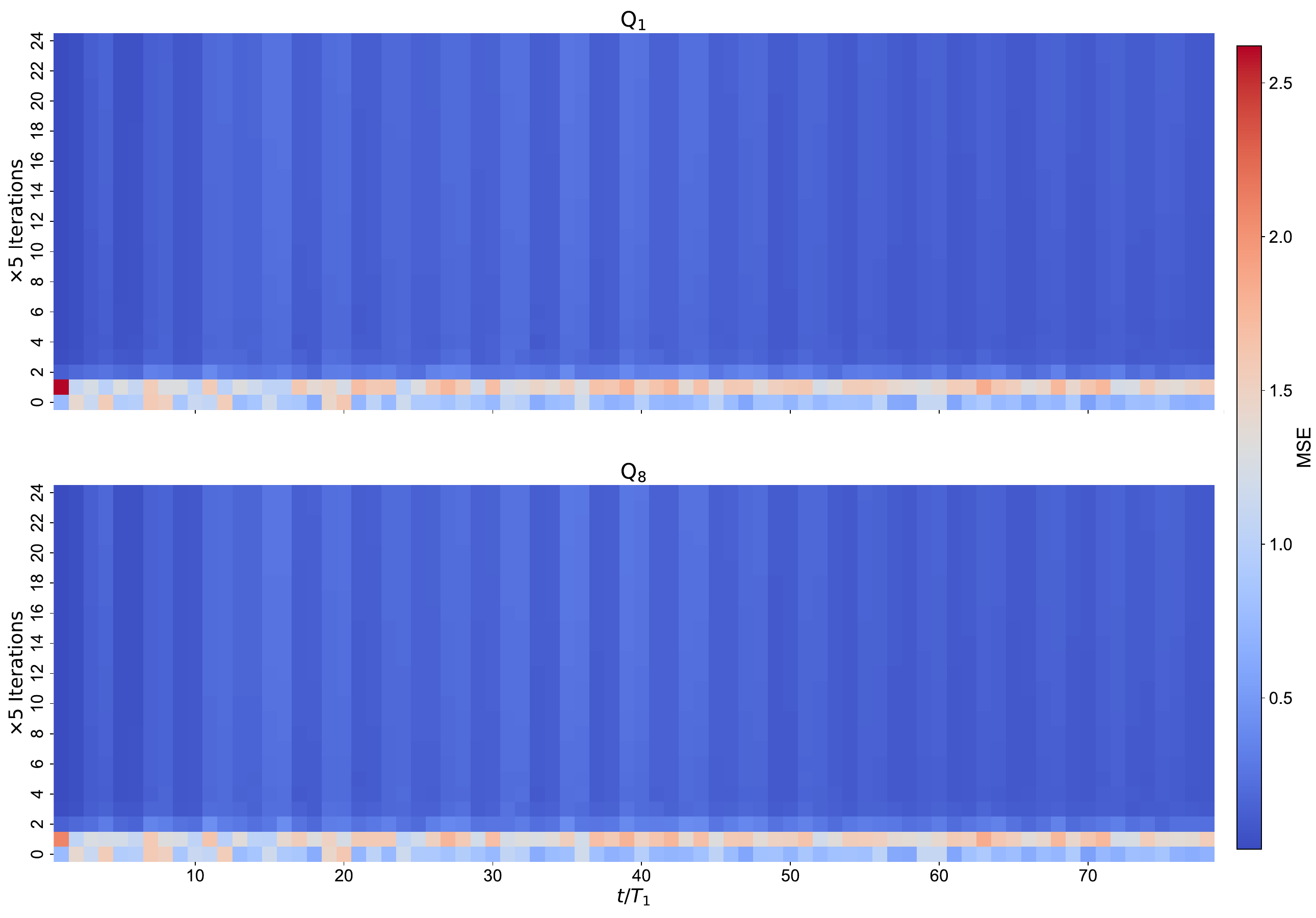}
	\caption{\small \textbf{The prediction performance of the surrogates during the optimization process.} The heat maps present the mean squared error (MSE) between the outputs of quantum processors and the outputs of the surrogates $h_{\mathsf{qs}}(\bx,\bm{\mathrm{w}}(j))$, where $\bm{\mathrm{w}}(j)$ refers to the optimized parameter after $j$ updates. The upper and lower panels plot the experimental results for the boundary qubits $Q_1$ and $Q_8$, respectively.
    The $x$-axis and $y$-axis refer to the evolution time $t/\mathrm{T}_1$ and the iteration step $j$, respectively. The label `$\times 5$' refers that the scales of the $y$-axis are scaled up by $5$ times. The results are aggregated from 10 independent trials. }
	\label{sfig:eva_hmb}
\end{figure*}

In this subsection, we evaluate the prediction performance of the surrogate $h_{\mathsf{qs}}(\bx,\bm{\mathrm{w}})=\braket{\bm{\Phi}_{\Omega(\Lambda)}(\bx), \bm{\mathrm{w}}}$ during the optimization process of $\bm{\mathrm{w}}$. Specifically, let $\bm{\mathrm{w}}(j)$ be the optimized parameter after $j$ updates in the optimization problem defined by Eq.~\eqref{append:eq:ridge}. We assess the prediction performance of the surrogates $h_{\mathsf{qs}}(\bx,\bm{\mathrm{w}}(j))$ during the training of  $\bm{\mathrm{w}}(j)$ with $j\in [480]$ using a validation set $\{(\bx^{(l)},\braket{Z_i(\bx^{(l)})})\}_{l=1}^{n_v}$ of size $n_v=50$. The predictive performance is measured by MSE between the output of the predictive surrogate $h_{\mathsf{qs}}(\bx,\bm{\mathrm{w}}(j))$ and the processor output $\braket{Z_i(\bx)}$, i.e.,
\begin{equation}
    \text{MSE}(j)=\frac{1}{n_v}\sum_{l=1}^{n_v}\left(h_{\mathsf{qs}}(\bx^{(l)},\bm{\mathrm{w}}(j))-\braket{Z_i(\bx^{(l)})}\right)^2.
\end{equation}

Given the critical role of boundary qubits in phase identification, we focus on the predictive surrogates $h_{\mathsf{qs}}^{(i,t)}$ for $i=1, N$ with $N=8$ across all evolution times $t \in \{ k\mathrm{T}_1: k =1,2...,79\}$. Figure.~\ref{sfig:eva_hmb} illustrates the training dynamics of $\text{MSE}(j)$ for various surrogates $h_{\mathsf{qs}}^{(i,t)}$ over $10$ independent trials. In each trial, the frequency set $\Omega(\Lambda)$ with $\Lambda=7$ is generated using distinct random seeds. The results show that $\text{MSE}(j)$ decreases rapidly within initial iterations for $j\in [20]$, then stabilizes at low values, indicating efficient convergence for both boundary qubits.

The complete experimental results for all qubits $i\in [N]$ with $N=8$ and various evolution times $t$ are provided in Tab~\ref{tab:mse_comparison}. The optimization process significantly reduces MSE across all configurations, with the optimized parameter $\bm{\mathrm{w}}$ typically achieving lower MSE up to two orders of magnitude than initial values. Notably, the standard deviations decrease substantially after optimization, indicating enhanced prediction stability. These consistent improvements confirm the effectiveness of the optimization procedure and the reliability of  $h_{\mathsf{qs}}$ for phase identification tasks. 

\subsubsection{Noisy numerical simulations for FSPT phase identification}\label{apd:numerical_simu_FSPT}
\begin{figure*}[t]
    \centering
    \includegraphics[width=0.5\textwidth]{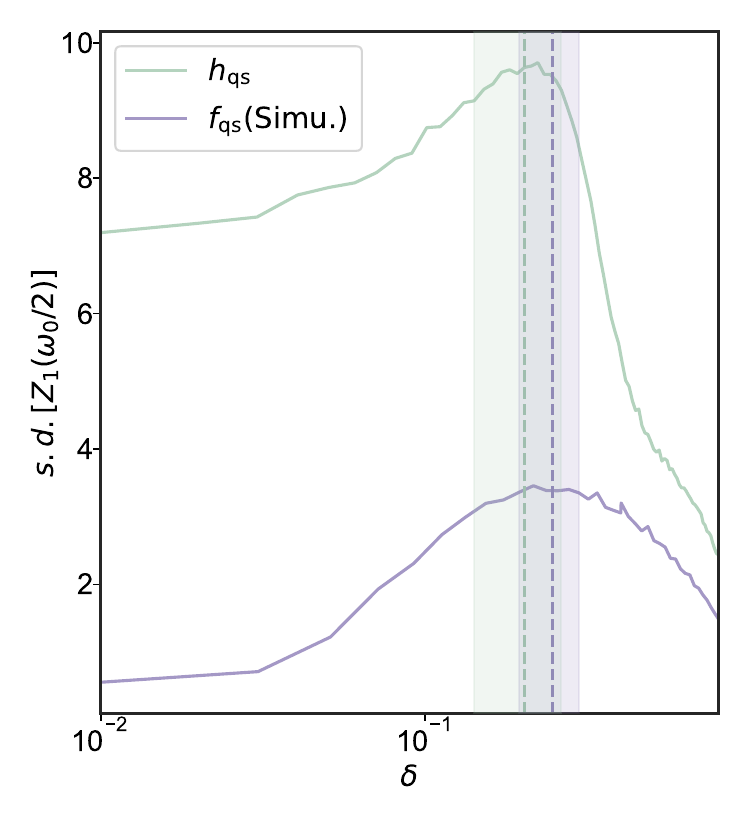}
    \caption{\small{
    \textbf{Experimental results for the critical region identification via noisy numerical simulations.} The green curve labeled `$h_{\mathsf{qs}}$' and the purple curve labeled `$f_{\mathsf{mb}}$(Simu.)' represent the 
    estimated subharmonic variance $\mathrm{s.d.}[Z_1(\omega_0/2)]$ as a function of the perturbation strength $\delta$, obtained from the surrogate predictions and noisy numerical simulations, respectively. The green and purple dashed vertical lines indicate the identified transition points, which refer to $0.202$ for surrogates and $0.246$ for the noisy numerical simulation. Moreover, the identified critical regions for these two methods are highlighted by the green and purple shadings, respectively.}}
    \label{sfig:trans}
\end{figure*}

Recall from FIG~4\text{c} in the main text that the critical region is confirmed through experiments involving a small set of parameters $\{(\delta^{(l)},\bm{J}^{(l,s)})\}_{l,s}^{n_v,S}$ with $n_v=6$ and $S=50$ on quantum processors. In this subsection, we expand this analysis to the entire identified critical region by conducting noisy numerical simulations, thereby further validating the predictive capability of the proposed surrogates. In particular, the driven perturbation $\delta^{(l)}$ is sampled from the $40$ equally spaced grid points within the range of $[0.01,0.8]$. For each $\delta^{(l)}$, the number of sampled coupling parameter $\bm{J}^{(l,s)}$ is set as $S=5000$. We separately estimate the subharmonic variance, denoted by $\mathrm{s.d.}[Z_{1}(\delta^{(l)}, \omega_0/2)]$, over the varied parameters $\{\bm{J}^{(l,s)}\}$ by using the predictive surrogate $h_{\mathrm{mb}}$ and classically simulating the noisy quantum circuits. Here, we implement the noisy simulation by applying depolarization and thermal relaxation noise channels to each base gate operation, matching the characterization of experimental devices as summarized in FIG.~\ref{sfig:device}.

The predicted subharmonic peak height as a function of $\delta$, obtained from the surrogates $h_{\mathsf{qs}}$ and the noisy classical simulation, is illustrated in FIG.~\ref{sfig:trans}. 
The results show that the critical region identified by noisy classical simulation exhibits remarkable consistency with that identified by the predictive surrogates. In particular, the predicted point with the largest variance is $0.202$ for the predictive surrogate and is $0.246$ for noisy classical simulations, with a minor discrepancy of $0.044$. Moreover, employing $95\%$ peak-height width criterion, the critical region is defined as the region of $\delta$ where $\mathrm{s.d.}[Z_{1}(\delta,\omega_0/2)]$ exceeds $95\%$ of its maximum value. This yields overlapping critical regions $\delta \in [0.152, 0.253]$ identified by predictive surrogates (as highlighted by the green shading) and $\delta \in [0.195, 0.297]$ identified by the noisy classical simulations (as highlighted by the purple shading). This minor discrepancy indicates that the proposed predictive surrogates provide conservative critical region estimates, serving as reliable lower-bound indicators for experimental phase transition identification.

\subsubsection{Measurement overhead reduction}\label{apd:FSPT_measurement_reduction}
This subsubsection quantitatively compares the measurement costs between our surrogate-based approach and conventional quantum methods. We first note that assessment methodologies are not inherently unique due to varying resource accounting frameworks; different definitions of measurement cost arise from choices in attributing computational overhead to classical versus quantum components. In this analysis, we adopt two specific comparisons: the experimental quantum measurement count in the surrogate workflow, and the equivalent quantum overhead if all computations performed by predictive surrogates were instead executed directly on the quantum processor.

For the surrogate-enabled workflow, the total cost combines: 1) $n = 250$ full quantum evolutions to train the $1 \times 79$ predictive surrogates $h_{\mathsf{qs}}$, and classical predictions identifying $\delta^* \in [0.101, 0.282]$ ($\Delta\delta_{\text{pred}} = 0.181$) across $100 \times 5,000$ virtual instances. 2)Experimental validation of the predicted transition region requires quantum measurements only within the reduced parameter range $\Delta\delta_{\text{pred}}$, corresponding to $18.1\%$ of the full range width ($\Delta\delta_{\text{full}} = 1.0$). Compared to conventional approaches requiring exhaustive sampling over $\delta \in [0, 1]$, our method achieves multiplicative savings by simultaneously reducing both spatial sampling range and required training evaluations.

To highlight the quantum resource savings inherent to our hybrid approach, we calculate the overhead if all surrogate predictions ($100 \times 5,000$ virtual instances) were executed on the quantum processor instead of classically. This hypothetical scenario would require $500,000$ full quantum evolutions. This overhead exceeds our actual experimental training cost of 250 quantum evolutions by a factor of $2,000$ (i.e., three orders of magnitude).

The most significant measurement reduction occurs during experimental validation. By concentrating quantum sampling within the surrogate-predicted critical region ($18.1\%$ of $\Delta\delta_{\text{full}}$), we achieve an $81.9\%$ reduction in parameter search width while maintaining precision. This is numerically verified in SM~\ref{apd:numerical_simu_FSPT}, where the predicted transition midpoint $\delta_{\text{surrogate}} = 0.202$ deviates by only $\Delta\delta=0.044$ ($4.4\%$ relative to $\Delta\delta_{\text{full}}$) from the noise-simulated $\delta^* = 0.246$. Our approach thus achieves sub-$5\%$ midpoint accuracy while reducing the experimental search space to $18.1\%$ of the original range, demonstrating efficient phase characterization through surrogate-guided parameter space pruning.

\subsection{Performance of the predictive surrogates with varied noise levels} \label{apd:vali_noise}
\begin{figure*}[t]
	\centering
	\includegraphics[width=0.85\textwidth]{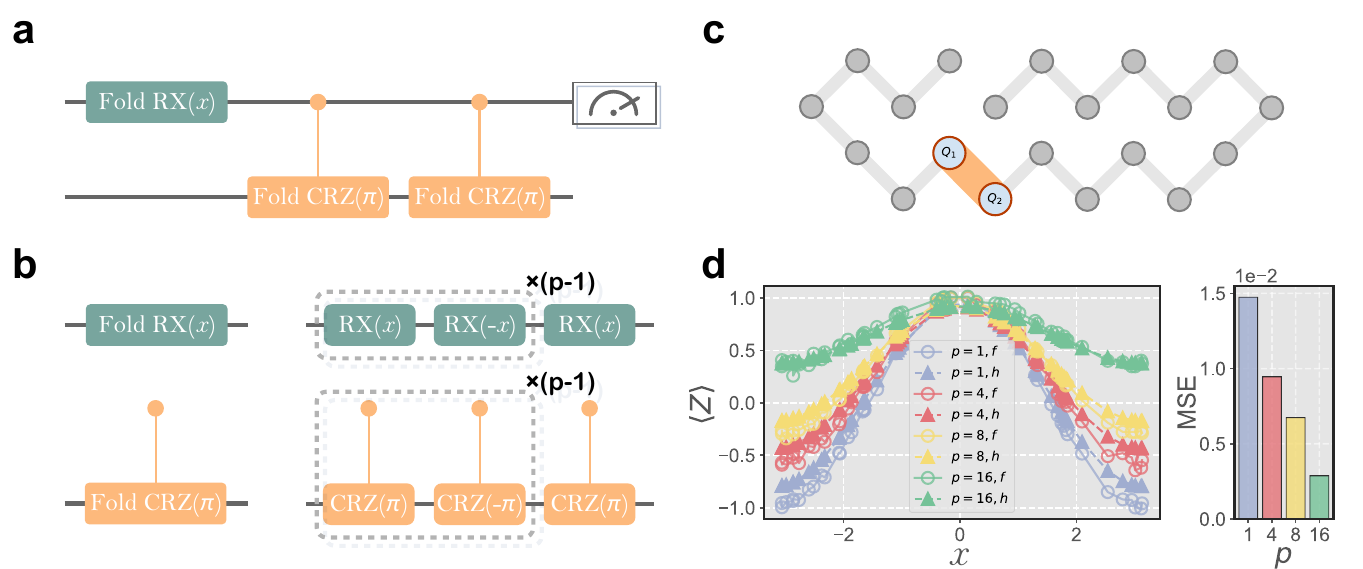}
	\caption{\small\textbf{Performance of $h_{\mathsf{cs}}$ with amplified noise levels.} 
	\textbf{a.} The implementation of the $2$-qubit benchmark circuit $U(x)$, where the $Z$-measurement is applied to the first qubit. 
	\textbf{b.} Gate folding implementation.  Each gate $G$ in $U(x)$ is replaced by $G \prod_{i=1}^p(G^\dagger G)$, preserving unitary equivalence while amplifying noise through $p$-fold repetition. 
	\textbf{c.} The two physical qubits employed in the superconducting quantum processor. 
	\textbf{d.} The left panel shows the  comparison of quantum processor outputs $f(x)$ (solid color lines) and predictive surrogate predictions $h_{cs}(x)$ (dashed color lines) for folding factors $p \in \{1,4,8,16\}$. 
	The right panel illustrates the mean squared error (MSE) between predictions and mean values obtained from the employed quantum processor.}
	\label{sfig:veri}
\end{figure*}

This subsection presents an experimental characterization of the predictive surrogate’s performance under varying noise levels. Recall that Theorem 1 establishes that, when system noise is modeled by a Pauli channel $\mathcal{N}_{P}$, the performance of the predictive surrogate $h_{\mathsf{cs}}$ improves as the noise level $P$ increases. Our primary objective here is to experimentally investigate whether, in practical scenarios, $h_{\mathsf{cs}}$ more accurately emulates quantum processor behavior as the noise levels increase.

We proceed with this exploration by comparing the output of the employed quantum processor $f(\tilderho(\bx), O)$ and the prediction of  $h_{\mathsf{cs}}(\bx)$ under distinct noise levels. In particular, the experimental setup is as follows. The employed quantum circuit takes the form as
\[U(\bx) =  \text{CRZ}_{1,2}(\pi) \cdot \text{CRZ}_{1,2}(\pi) \cdot \RX_1(\bx),\]
where $d=1$ and $\text{CRZ}_{1,2}$ refers to the controlled-$\RZ$ gate with the controlled qubit being $Q1$. The initial state is fixed to be $\rho_0=\ket{00}\bra{00}$ and the observable is set to $O=Z\otimes \mathbb{I}_2$. To vary the noise levels, we employ gate folding, a strategy inspired by quantum error mitigation~\cite{kandala2017hardware}. As shown in FIG.~\ref{sfig:veri}\text{b}, to amplify noise by $p$ folds,  each gate $G$ is replaced with the gate sequence $G \to G\prod_{i=1}^p(G^{\dagger}G)$, where each additional $G^{\dagger}G$ pair increases the effective noise experienced by the circuit. Accordingly, the explicit form of the quantum circuit implementing a $p$-fold noise amplification is given by
\[U(\bx;p) =  \Big (\text{CRZ}_{1,2}(\pi) \prod_{i=1}^p\big(\text{CRZ}_{1,2}(\pi)\text{CRZ}_{1,2}(\pi)\big) \Big) \cdot \Big (\text{CRZ}_{1,2}(\pi) \prod_{i=1}^p\big(\text{CRZ}_{1,2}(\pi)\text{CRZ}_{1,2}(\pi)\big) \Big) \cdot  \RX_1(\bx) \Big ( \prod_{i=1}^p\big(\RX_1(\bx)\RX_1(\bx)\big) \Big).\]
An intuition of the circuit implementation is shown in FIG.~\ref{sfig:veri}\text{a}. In ideal noiseless conditions, this circuit would produce the exact cosine response with $f(\rho(\bx), O) = \langle Z\rangle = \cos(\bx)$. 
   
In our experiments, we vary the folding factor with $p \in \{1,4,8,16\}$. 
For each $p$, we configure the predictive surrogate $h_{\mathsf{cs}}$ with the truncation value $\Lambda=1$ and adopt the similar hyperparameter settings in SM.~\ref{apd:hyper}. In particular, we uniformly sample both training and test points across the parameter space $\bx \in [-\pi,\pi]$, collecting $n=50$ training examples with $T=1000$ snapshots each, and $50$ test points with $T=20,000$ measurement shots each to ensure the statistical significance.

The achieved experimental results are demonstrated in FIG.~\ref{sfig:veri}\text{d}. An observation is that for all different settings of $p$, the optimized predictive surrogates accurately predict the noisy output of quantum processors. Crucially, with the increased $p$, the MSE between $f(\tilderho(\bx),O)$ and $h_{\mathsf{cs}}(\bx)$ decrease, which is from $0.0147$ at $p=1$ to $0.0028$ at $p=16$. This inverse relationship indicates that the proposed predictive surrogate can accurately emulate the mean-value behavior of quantum processors in practice, particularly as the noise level increases.

\begin{longtable}{|c|*{8}{c|}}
\caption{\small \textbf{Mean square error between the quantum processor outputs and the outputs of surrogates before and after training}.  
    Each row presents the results for varying evolution times $t = k\mathrm{T}_1$ with $k = 1,\dots,79$. Each column shows the results for different observables $\{Z_i\}$ with the qubit index $i\in [8]$. The upper and lower numbers in each cell refer to the mean and the standard deviation of MSE for the surrogates before and after training, respectively. The statistical results are obtained by running $10$ independent trials.} \label{tab:mse_comparison} \\
\toprule
\multirow{2}{*}{$t/T_1$} & \multicolumn{8}{c|}{Qubit Index} \\
\cmidrule{2-9}
 & 1 & 2 & 3 & 4 & 5 & 6 & 7 & 8 \\
\midrule
\endfirsthead

\multicolumn{9}{c}%
{{\bfseries Continued from previous page}} \\
\toprule
\multirow{2}{*}{$t/T_1$} & \multicolumn{8}{c|}{Qubit Index} \\
\cmidrule{2-9}
 & 1 & 2 & 3 & 4 & 5 & 6 & 7 & 8 \\
\midrule
\endhead

\midrule \multicolumn{9}{|r|}{{Continued on next page}} \\ \midrule
\endfoot

\bottomrule
\endlastfoot
1 & \makecell{ \scriptsize 0.7557 $\pm$ 0.9767 \\ \scriptsize 0.0127 $\pm$ 0.0000 } & \makecell{ \scriptsize 0.7557 $\pm$ 0.9767 \\ \scriptsize 0.0077 $\pm$ 0.0000 } & \makecell{ \scriptsize 0.7557 $\pm$ 0.9767 \\ \scriptsize 0.0090 $\pm$ 0.0000 } & \makecell{ \scriptsize 0.7557 $\pm$ 0.9767 \\ \scriptsize 0.0069 $\pm$ 0.0000 } & \makecell{ \scriptsize 0.7557 $\pm$ 0.9767 \\ \scriptsize 0.0070 $\pm$ 0.0000 } & \makecell{ \scriptsize 0.7557 $\pm$ 0.9767 \\ \scriptsize 0.0070 $\pm$ 0.0000 } & \makecell{ \scriptsize 0.7557 $\pm$ 0.9767 \\ \scriptsize 0.0068 $\pm$ 0.0000 } & \makecell{ \scriptsize 0.7557 $\pm$ 0.9767 \\ \scriptsize 0.0080 $\pm$ 0.0000 } \\ \midrule
2 & \makecell{ \scriptsize 1.4099 $\pm$ 1.3250 \\ \scriptsize 0.0355 $\pm$ 0.0045 } & \makecell{ \scriptsize 1.4099 $\pm$ 1.3250 \\ \scriptsize 1.3095 $\pm$ 0.0265 } & \makecell{ \scriptsize 1.4099 $\pm$ 1.3250 \\ \scriptsize 1.3616 $\pm$ 0.0662 } & \makecell{ \scriptsize 1.4099 $\pm$ 1.3250 \\ \scriptsize 1.2383 $\pm$ 0.0137 } & \makecell{ \scriptsize 1.4099 $\pm$ 1.3250 \\ \scriptsize 1.2001 $\pm$ 0.0079 } & \makecell{ \scriptsize 1.4099 $\pm$ 1.3250 \\ \scriptsize 1.2681 $\pm$ 0.0155 } & \makecell{ \scriptsize 1.4099 $\pm$ 1.3250 \\ \scriptsize 1.1927 $\pm$ 0.0116 } & \makecell{ \scriptsize 1.4099 $\pm$ 1.3250 \\ \scriptsize 0.0349 $\pm$ 0.0045 } \\
\end{longtable}

\end{document}